\documentclass{theoretics}

\usepackage{mathtools}
\usepackage{mathrsfs}
\usepackage{thm-restate}

\definecolor{darkgreen}{rgb}{0,0.5,0}
\usepackage[capitalise,noabbrev,nameinlink]{cleveref}

 \crefname{theorem}{Theorem}{Theorems}
 \Crefname{lemma}{Lemma}{Lemmas}
 \Crefname{figure}{Figure}{Figures}
 \Crefname{sidefigure}{Figure}{Figures}
 \Crefname{claim}{Claim}{Claims}

  \Crefname{observation}{Observation}{Observations}

\newcommand{\Pout}{P^{\mathsf{out}}}
\newcommand{\thetawest}{\theta_{\mathsf{west}}}
\newcommand{\thetaeast}{\theta_{\mathsf{east}}}

\newcommand{\RR}{\mathcal{R}}
\newcommand{\EE}{\mathcal{E}}
\newcommand{\FC}{F_\mathsf{c}}
\newcommand{\FO}{F_\mathsf{o}}
\newcommand{\EH}{E_\mathsf{h}}
\newcommand{\EV}{E_\mathsf{v}}
\newcommand{\PH}{\mathcal{S}_\mathsf{h}}

\newcommand{\Nnorth}{\mathcal{N}_\mathsf{north}}
\newcommand{\Nsouth}{\mathcal{N}_\mathsf{south}}
\newcommand{\rotation}{\mathsf{rotation}}
\newcommand{\direction}{\mathsf{direction}}

\newcommand{\eref}{e^\star}

\renewcommand{\mod}{\operatorname{mod}}

\newcommand{\poly}{\operatorname{poly}}

\makeatletter
\newcommand\IfRestateTF{%
  \ifx\label\thmt@gobble@label 
    \expandafter\@firstoftwo
  \else
    \expandafter\@secondoftwo
  \fi
}
\makeatother
\newcommand{\RestateRemark}{\IfRestateTF{{\normalfont\bfseries (Restated) }}{}}

\ThCSauthor{Yi-Jun Chang}{cyijun@nus.edu.sg}[0000-0002-0109-2432]
\ThCSaffil{National University of Singapore}
\ThCSthanks{A preliminary version of this article appeared at ICALP 2023 \cite{Chang23}.}
\ThCSshortnames{Y.\ Chang}  
\ThCSshorttitle{Ortho-Radial Drawing in Near-Linear Time}
\ThCSyear{2025}
\ThCSarticlenum{11}
\ThCSreceived{Apr 16, 2024}
\ThCSrevised{Jan 3, 2025}
\ThCSaccepted{Feb 27, 2025}
\ThCSpublished{Apr 8, 2025}
\ThCSdoicreatedtrue
\ThCSkeywords{Graph drawing, ortho-radial drawing, topology-shape-metric framework}

\title{Ortho-Radial Drawing in Near-Linear Time}

\begin{document}
\maketitle

\begin{abstract}
An \emph{orthogonal drawing} is an embedding of a plane graph into a grid. In a seminal work of Tamassia (SIAM Journal on Computing 1987), a simple combinatorial characterization of angle assignments that can be realized as bend-free orthogonal drawings was established, thereby allowing an orthogonal drawing to be described combinatorially by listing the angles of all corners. The characterization reduces the need to consider certain geometric aspects, such as edge lengths and vertex coordinates, and simplifies the task of graph drawing algorithm design.

Barth, Niedermann, Rutter, and Wolf (SoCG 2017) established an analogous combinatorial characterization for \emph{ortho-radial drawings}, which are a generalization of orthogonal drawings to \emph{cylindrical grids}. The proof of the characterization is existential and does not result in  an efficient algorithm. Niedermann, Rutter, and Wolf (SoCG 2019) later addressed this issue by developing quadratic-time algorithms for both testing the realizability of a given angle assignment as an ortho-radial drawing without bends and constructing such a drawing.

In this paper, we further improve the time complexity of these tasks to near-linear time. We establish a new characterization for ortho-radial drawings based on the concept of a \emph{good sequence}. Using the new characterization, we design a simple greedy algorithm for constructing ortho-radial drawings.
\end{abstract}


\section{Introduction}\label{sect:intro}

A \emph{plane graph} is a \emph{planar graph} $G=(V,E)$ with a \emph{combinatorial embedding} $\EE$. The combinatorial embedding $\EE$ fixes a circular ordering $\EE(v)$ of the edges incident to each vertex $v \in V$, specifying the counter-clockwise ordering of these edges surrounding $v$ in the drawing.
An \emph{orthogonal drawing} of a plane graph is a drawing of $G$ such that each edge is drawn as a sequence of horizontal and vertical line segments. For example, see \cref{fig:f1} for an orthogonal drawing of $K_4$ with four bends. Alternatively, an orthogonal drawing of $G$ can be seen as an embedding of $G$ into a grid such that the edges of $G$ correspond to internally disjoint paths in the grid.
Orthogonal drawing is one of the most classical drawing styles studied in the field of graph drawing, and it has a wide range of applications, including VLSI circuit design~\cite{bhatt1984framework,valiant1981universality}, architectural floor plan design~\cite{LIGGETT1981277}, and network visualization~\cite{batini1986layout,eiglsperger2004automatic,gutwenger2003new,kieffer2015hola}.

\paragraph{The topology-shape-metric framework}
One of the most fundamental quality measures of orthogonal drawings is the number of \emph{bends}.
The \emph{bend minimization} problem, which asks for an orthogonal drawing with the smallest number of bends, has been extensively studied over the past 40 years~\cite{CK12,BLV98,didimo2020optimal,GT96,storer1980node,tamassia1987embedding}.
In a seminal work, Tamassia~\cite{tamassia1987embedding} introduced the \emph{topology-shape-metric} framework to tackle the bend minimization problem.  Tamassia showed that an orthogonal drawing can be described combinatorially by an \emph{orthogonal representation}, which consists of an assignment of an angle in $\{90^\circ, 180^\circ, 270^\circ, 360^\circ\}$ to each corner and a designation of the \emph{outer face}. In this paper, a \emph{corner} is defined as a pair of edges incident to the same vertex that are consecutive in the given combinatorial embedding. Specifically, Tamassia~\cite{tamassia1987embedding} showed that an orthogonal representation can be realized as an orthogonal drawing with zero bends if and only if the following two conditions are satisfied:
\begin{enumerate}[(O1)]
    \item \label{item:O1} The sum of angles around each vertex is $360^\circ$.
    \item \label{item:O2} The sum of angles around each face with $k$ corners is $(k+2)\cdot 180^\circ$ for the outer face and is $(k-2)\cdot 180^\circ$ for the other faces.
\end{enumerate}

An orthogonal representation is \emph{valid} if it satisfies the above conditions \ref{item:O1} and \ref{item:O2}. Given a valid orthogonal representation, an orthogonal drawing realizing the orthogonal representation can be computed in linear time~\cite{hsueh1980symbolic,tamassia1987embedding}. 
This result (shape $\rightarrow$ metric) allows us to reduce the task of finding a bend-minimized orthogonal drawing (topology $\rightarrow$ metric) to the conceptually much simpler task of finding a bend-minimized valid orthogonal representation (topology $\rightarrow$ shape). 

By focusing on orthogonal representations, we may neglect certain geometric aspects of graph drawing such as edge lengths and vertex coordinates, making the task of algorithm design easier. In particular, given a fixed combinatorial embedding, the task of finding a bend-minimized orthogonal representation can be easily reduced to the  computation of a minimum cost flow~\cite{tamassia1987embedding}. Such a reduction to a flow computation is not easy to see if one thinks about orthogonal drawings directly without thinking about orthogonal representations. 

\begin{figure}[t!]
\centering
\includegraphics[width=\textwidth]{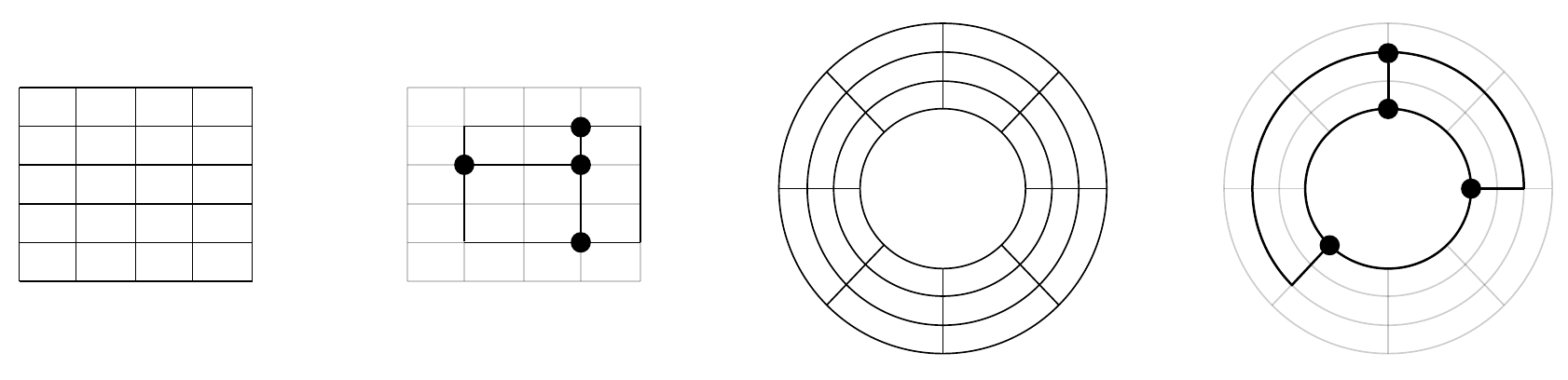}
\caption{A grid, an orthogonal drawing, a cylindrical grid, and an ortho-radial drawing.}\label{fig:f1}
\end{figure}

\subsection{Ortho-radial drawing}

\emph{Ortho-radial drawing} is a natural generalization of orthogonal drawing to \emph{cylindrical grids}, whose grid lines consist of concentric circles and straight lines emanating from the center of the circles. Formally, an ortho-radial drawing is defined as a planar embedding where each edge is drawn as a sequence of lines that are either a circular arc of some circle centered on the origin or a line segment of some straight line passing through the origin. We do not allow a vertex to be drawn on the origin, and we do not allow an edge to pass through the origin in the drawing. For example, see \cref{fig:f1} for an ortho-radial drawing of $K_4$ with two bends.

The study of ortho-radial drawing is motivated by its applications in network visualization~\cite{bast2021metro,fink2014concentric,xu2022automated}, particularly in the context of transit map layout~\cite{wu2020survey}. Ortho-radial drawings are especially well-suited for visualizing metro systems with radial and circular routes. Examples of such drawings can be found in~\cite{ovenden2003metro,roberts2005underground,roberts2012underground}.


There are three types of faces in an ortho-radial drawing.
The face that contains an unbounded region is called the \emph{outer face}.
The face that contains the origin is called the \emph{central face}. The remaining faces are called \emph{regular faces}. It is possible that the outer face and the central face are the same face.

Given a plane graph, an \emph{ortho-radial representation} is defined as an assignment of an angle to each corner together with a designation of the central face and the outer face.
Barth, Niedermann, Rutter, and Wolf~\cite{barth2023topology}  showed that an ortho-radial representation can be realized as an ortho-radial drawing with zero bends if the  following three conditions are satisfied:
\begin{enumerate}[(R1)]
    \item \label{item:R1} The sum of angles around each vertex is $360^\circ$.
    \item \label{item:R2} The sum $s$ of angles around each face $F$ with $k$ corners satisfies the following.
    \begin{itemize}
        \item $s = (k-2)\cdot 180^\circ$ if $F$ is a regular face.
        \item $s = k\cdot 180^\circ$ if $F$ is either the central face or the outer face, but not both.
        \item $s = (k+2) \cdot 180^\circ$ if $F$ is both the central face and the outer face.
    \end{itemize}
    \item \label{item:R3}There exists a choice of the \emph{reference edge} $\eref$ such that the ortho-radial representation does not contain a \emph{strictly monotone cycle}.
\end{enumerate}

Intuitively, this shows that the ortho-radial representations that can be realized as ortho-radial drawings with zero bends can be characterized similarly by examining the angle sum around each vertex and each face, with the additional requirement that the representation does not have a strictly monotone cycle.

The definition of a strictly monotone cycle is  technical and depends on the choice of the reference edge $\eref$, so we defer its formal definition to a subsequent section. 
The reference edge~$\eref$ is an edge in the contour of the outer face and is required to lie on the outermost circular arc used in an ortho-radial drawing. 
Informally, a strictly monotone cycle has a structure that is like a loop of ascending stairs or a loop of descending stairs, so a strictly monotone cycle cannot be drawn. The necessity of \ref{item:R1}--\ref{item:R3} is intuitive to see. The more challenging and interesting part of the proof in~\cite{barth2023topology} is to show that these three conditions are actually sufficient.

\subsection{Previous methods}
The journal paper~\cite{barth2023topology} is a combination of two works~\cite{barth2017,niedermann19}.
In the first work~\cite{barth2017}, the proof that conditions \ref{item:R1}--\ref{item:R3} are necessary and sufficient is only \emph{existential} in that it does not yield efficient algorithms to check the {validity} of a given ortho-radial representation and to construct an ortho-radial drawing without bends realizing a given ortho-radial representation.

Checking \ref{item:R1} and \ref{item:R2} can be done in linear time in a straightforward manner. The difficult part is to design an efficient algorithm to check \ref{item:R3}. The most naive approach of examining all cycles costs exponential time. 
The second work~\cite{niedermann19} addressed this gap by showing an $O(n^2)$-time algorithm to decide whether a strictly monotone cycle exists for a given reference edge $\eref$, where $n$ is the number of vertices in the input graph. They also show an $O(n^2)$-time algorithm to construct an ortho-radial drawing without bends, for any given ortho-radial representation with a reference edge $\eref$ that does not contain a strictly monotone cycle.

\paragraph{Rectangulation}
The main idea behind the proof in the first work~\cite{barth2017} is a reduction to the easier case where each regular face is \emph{rectangular}. For this case, they provided a proof that conditions \ref{item:R1}--\ref{item:R3} are necessary and sufficient, and they also provided an efficient drawing algorithm via a reduction to a flow computation given that \ref{item:R1}--\ref{item:R3} are satisfied. 

For any given ortho-radial representation with $n$ vertices, it is possible to add $O(n)$ additional edges to turn it into an ortho-radial representation where each regular face is rectangular. A major difficulty in the proof of~\cite{barth2017} is that they need to ensure that the addition of the edges preserves not only \ref{item:R1} and \ref{item:R2} but also \ref{item:R3}.
The lack of an efficient algorithm to check whether \ref{item:R3} is satisfied is precisely the reason that the proof of~\cite{barth2017} does not immediately lead to a polynomial-time algorithm. 

\paragraph{Quadratic-time algorithms}
The above issue was addressed in the second work~\cite{niedermann19}. They provided an $O(n^2)$-time algorithm to find a strictly monotone cycle if one exists, given a fixed choice of the reference edge $\eref$. This immediately leads to an $O(n^2)$-time algorithm to decide whether a given ortho-radial representation, with a fixed reference edge $\eref$, admits an ortho-radial drawing. Moreover, combining this $O(n^2)$-time algorithm with the proof of~\cite{barth2017} discussed above yields an $O(n^4)$-time drawing algorithm. The time complexity is due to the fact that $O(n)$ edge additions are needed for rectangulation, for each edge addition there are $O(n)$ candidate reference edges to consider, and to test the feasibility of each candidate edge they need to run the  $O(n^2)$-time algorithm to test whether the edge addition creates a strictly monotone cycle.

The key idea behind the $O(n^2)$-time algorithm for finding a strictly monotone cycle is a structural theorem that if there is a strictly monotone cycle, then there is a unique outermost one which can be found by a \emph{left-first} DFS starting from any edge in the outermost strictly monotone cycle. The DFS algorithm costs $O(n)$ time. Guessing an edge in the outermost monotone cycle adds an $O(n)$ factor overhead in the time complexity.

Using further structural insights on the augmentation process of~\cite{barth2017}, the time complexity of the above $O(n^4)$-time drawing algorithm can be lowered to $O(n^2)$~\cite{niedermann19}. The reason for the quadratic time complexity is that for each of the $O(n)$ edge additions, a left-first DFS starting from the newly added edge is needed to test whether the addition of this edge creates a strictly monotone cycle.

\subsection{Our new method}

For both validity testing (checking whether a given angle assignment induces a strictly monotone cycle) and drawing (finding a geometric embedding realizing a given ortho-radial representation), the two algorithms in~\cite{barth2023topology} naturally cost $O(n^2)$ time, as they both require performing left-first DFS  $O(n)$ times.

In this paper, we present a new method for ortho-radial drawing that is not based on rectangulation and left-first DFS. 
We design a simple $O(n \log n)$-time greedy algorithm that simultaneously accomplishes both validity testing and drawing, for the case where the reference edge $\eref$ is fixed. If a reference edge $\eref$ is not fixed, our algorithm costs $O(n \log^2 n)$ time, where the extra $O(\log n)$ factor is due to a binary search over the set of  candidates for the reference edge.
 At a high level, our algorithm tries to construct an ortho-radial drawing in a piece-by-piece manner. If at some point no progress can be made in that the current partial drawing cannot be further extended, then the algorithm can identify a strictly monotone cycle to certify the non-existence of a drawing. 
 
 \paragraph{Good sequences} The core of our method is the notion of a \emph{good sequence}, which we briefly explain below. An ortho-radial representation satisfying \ref{item:R1} and \ref{item:R2}, with a fixed reference edge $\eref$, determines whether an edge $e$ is a vertical edge (i.e., $e$ is drawn as a  segment of  a straight line passing through the origin) or  horizontal  (i.e., $e$ is drawn as a circular arc of some circle
centered on the origin). Let $\EH$ denote the set of horizontal edges, oriented in the clockwise direction, and let $\PH$ denote the set of connected components induced by $\EH$. Note that each component $S \in \PH$ is either a path or a cycle.

The exact definition of a good sequence is technical, so we defer it to a subsequent section. Intuitively, a good sequence is  an ordering of $\PH=(S_1, S_2, \ldots, S_{k})$, where $k = |\PH|$, that allows us to design
a simple linear-time greedy algorithm constructing an ortho-radial drawing in such a way that $S_1$ is drawn on the circle $r=k$, $S_2$ is drawn on the circle $r = k-1$, and so on. 
 
In general, a good sequence might not exist, even if the given ortho-radial representation admits an ortho-radial drawing. In such a case, we show that we may add \emph{virtual edges} to transform the ortho-radial representation into one where a good sequence exists. 
We will design a greedy algorithm for adding virtual edges and constructing a good sequence. In each step, we add virtual vertical edges to the current graph and append a new element $S \in \PH$ to the end of our sequence. In case we are unable to find any suitable $S \in \PH$ to extend the sequence, we can extract a strictly monotone cycle to certify the non-existence of an ortho-radial drawing. We emphasize that the cycle belongs to the original graph and does not use any of the virtual edges.

A major difference between our method and the approach based on rectangulation in~\cite{barth2023topology} is that the cost for adding a new virtual edge is only $O(\log n)$ in our algorithm. As we will later see, in our algorithm, in order to identify new virtual edges to be added, we only need to do some simple local checks such as calculating the sum of angles, and there is no need to do a full left-first DFS to test whether a newly added edge creates a strictly monotone cycle. 

\paragraph{Open questions} While we show a nearly linear-time algorithm for the (shape $\rightarrow$ metric)-step (i.e., from ortho-radial representations to ortho-radial drawings), essentially nothing is known about the (topology $\rightarrow$ shape)-step (from planar graphs to ortho-radial representations). 
While the task of finding a \emph{bend-minimized orthogonal representation} of a given plane graph can be easily reduced to the computation of a minimum cost flow~\cite{tamassia1987embedding}, such a reduction does not apply to ortho-radial representations, as network flows do not work well with the notion of strictly monotone cycles. It remains an open question whether a bend-minimized ortho-radial representation of a plane graph can be computed in polynomial time.

\subsection{Related work}
Orthogonal drawing is a central topic in graph drawing, see~\cite{DG13} for a survey. The bend minimization problem for orthogonal drawings of planar graphs of maximum degree $4$ without a fixed combinatorial embedding is NP-hard~\cite{formann1993drawing,GT96}. If the combinatorial embedding is fixed, the topology-shape-metric framework of Tamassia~\cite{tamassia1987embedding} reduces the bend minimization problem to a min-cost flow computation. The algorithm of Tamassia~\cite{tamassia1987embedding} costs $O(n^2 \log n)$ time. The time complexity was later improved to $O\left(n^{7/4} \sqrt{\log n}\right)$~\cite{GT96} and then to $O\left(n^{3/2} \log n\right)$~\cite{CK12}. A recent $O(n \poly \log n)$-time planar min-cost flow algorithm~\cite{dong2022nested} implies that the bend minimization problem can be solved in $O(n \poly \log n)$ time if the combinatorial embedding is fixed.

If the combinatorial embedding is not fixed, the NP-hardness result of~\cite{formann1993drawing,GT96} can be bypassed if the first bend on each edge does not incur any cost~\cite{BRW12} or if we restrict ourselves to some special class of planar graphs. In particular, for planar graphs with maximum degree 3, it was shown that the bend-minimization can be solved in polynomial time~\cite{BLV98}. After a series of improvements~\cite{chang2017bend,didimo2020optimal,didimo2018bend}, we now know that a bend-minimized orthogonal drawing of a planar graph with maximum degree $3$ can be computed in $O(n)$ time~\cite{didimo2020optimal}.

The topology-shape-metric framework~\cite{tamassia1987embedding} is not only useful in bend minimization, but it is also, implicitly or explicitly, behind the graph drawing algorithms for essentially all computational problems in orthogonal drawing and its variants, such as morphing orthogonal drawings~\cite{biedl2013morphing},  allowing vertices with degree greater than 4~\cite{BDPP99,KM98,papakostas2000efficient},  restricting the direction of edges~\cite{didimo2014complexity,durocher2014drawing}, drawing cluster graphs~\cite{brandes2004draw}, and drawing dynamic graphs~\cite{brandes1998dynamic}.

The study of ortho-radial drawing by Barth, Niedermann, Rutter, and Wolf~\cite{barth2023topology} extended the topology-shape-metric framework~\cite{tamassia1987embedding} to accommodate cylindrical grids. Before the work~\cite{barth2023topology}, a combinatorial characterization of drawable ortho-radial representation was only known for paths, cycles, and theta graphs~\cite{hasheminezhad2009ortho}, and for the special case where the graph is 3-regular and each regular face in the ortho-radial representation is a rectangle~\cite{hasheminezhad2010rectangular}.

\subsection{Organization}

In \cref{sect:prelim}, we discuss the basic graph terminology used in this paper, review some results in the previous work~\cite{barth2023topology}, and state our main theorems. In \cref{sect:drawing}, we introduce the notion of a good sequence and  show that its existence implies a simple  ortho-radial drawing algorithm. In \cref{sect:sequence}, we present a greedy algorithm that adds virtual edges to a given ortho-radial representation with a fixed reference edge so that a good sequence that covers the entire graph exists and can be computed efficiently. In \cref{sect:cycle}, we show that a strictly monotone cycle, which certifies the non-existence of a drawing, exists and can be computed efficiently if the greedy algorithm fails. In \cref{sect:binary-search}, we show that our results can be extended to the setting where the reference edge is not fixed at the cost of an extra logarithmic factor in the time complexity.
In \cref{sect:reduction}, we justify our assumption that the input graph is simple and  biconnected by showing a reduction from any graph to a  biconnected simple graph. We conclude in \cref{sect:conclusions} with discussions on possible future directions.

\section{Preliminaries}\label{sect:prelim}

Throughout the paper, let $G=(V,E)$ be a planar graph of maximum degree at most $4$ with a fixed combinatorial embedding $\EE$ in the sense that, for each vertex $v \in V$, a circular ordering $\EE(v)$ of its incident edges is given to specify the counter-clockwise ordering of these edges surrounding $v$ in a planar embedding. As we will discuss in \cref{sect:reduction}, we may assume that the input graph $G$ is \emph{simple} and \emph{biconnected}. 
In this section, we introduce some basic graph terminology and review some results from Barth, Niedermann, Rutter, and Wolf~\cite{barth2023topology}.

\paragraph{Paths and cycles} Unless otherwise stated, all edges, paths, and cycles are assumed to be directed. We write $\overline{e}$, $\overline{P}$, and $\overline{C}$ to denote the \emph{reversal} of an edge $e$, a path $P$, and a cycle $C$, respectively. We allow paths and cycles to have repeated vertices and edges. We say that a path or a cycle is \emph{simple} if it does not have repeated vertices.
Following~\cite{barth2023topology}, we say that a path or a cycle is \emph{crossing-free} if it satisfies the following conditions:
\begin{itemize}
    \item The path or the cycle does not contain repeated undirected edges. See \cref{fig:pathscylces} for an illustration: The cycle $C = (v_1, v_5, v_6, v_3, v_4, v_7, v_6, v_5, v_9, v_{10}, v_2)$ is not crossing-free as it traverses the undirected edge $\{v_5, v_6\}$ twice, from opposite directions.
    \item For each vertex $v$ that appears multiple times in the path or the cycle, the ordering of the edges incident to $v$ appearing in the path or the cycle matches either the order of these edges in $\EE(v)$ or its reversal. See \cref{fig:pathscylces} for an illustration: The path $(v_{11}, v_9, v_5, v_1, v_2, v_{10}, v_9, v_8)$ is not crossing-free, as it crosses itself at $v_9$; the path $(v_{8}, v_9, v_5, v_1, v_2, v_{10}, v_9, v_{11})$ is crossing-free, as the ordering of the edges incident to $v_9$ appearing in the path matches the order of these edges in $\EE(v_9)$.
\end{itemize}
 Although a crossing-free path or a crossing-free cycle might touch a vertex multiple times, the path or the cycle never crosses itself.
For any face $F$, we define the \emph{facial cycle} $C_F$ to be the clockwise traversal of its contour. 
In general, a facial cycle might not be a simple cycle as it can contain repeated edges. For example, the cycle $C = (v_1, v_5, v_6, v_3, v_4, v_7, v_6, v_5, v_9, v_{10}, v_2)$ in \cref{fig:pathscylces}, which is not simple, is the facial cycle of $F_2$.
If we assume that $G$ is biconnected, then each facial cycle of $G$ must be a simple crossing-free cycle. 

\begin{figure}[t!]
\centering
\includegraphics[width=0.7\textwidth]{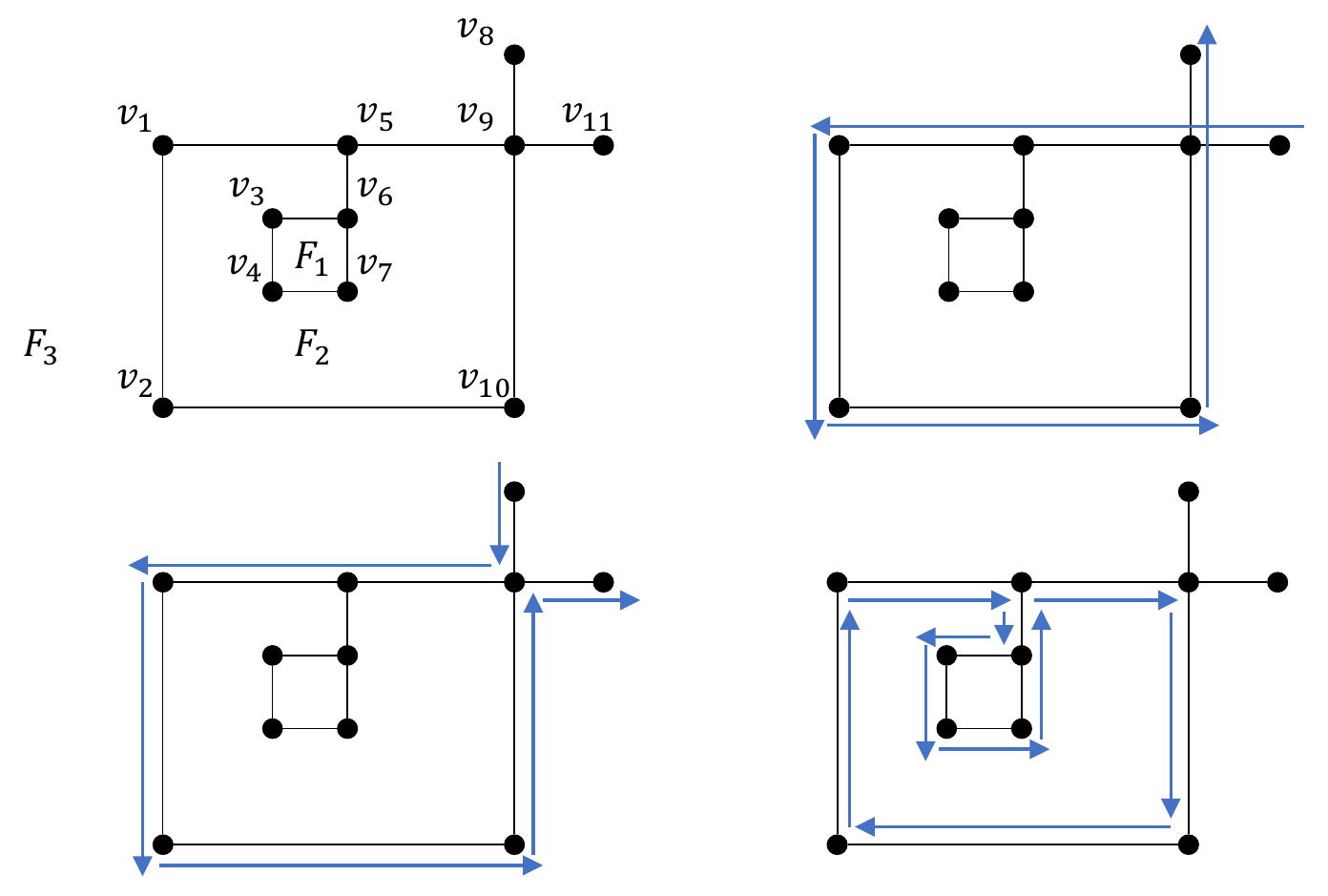}
\caption{A non-crossing-free path, a crossing-free path, and a facial cycle.}\label{fig:pathscylces}
\end{figure}

\paragraph{Ortho-radial representations and drawings} A \emph{corner} is an ordered pair of undirected edges $(e_1, e_2)$ incident to $v$ such that $e_2$ immediately follows $e_1$ in the counter-clockwise circular ordering $\EE(v)$. Given a planar graph $G=(V,E)$  with a fixed combinatorial embedding $\EE$,
an ortho-radial representation $\RR=(\phi, \FC, \FO)$ of $G$ is defined by the following components:
\begin{itemize}
    \item An assignment $\phi$ of an angle $a \in \{90^\circ, 180^\circ, 270^\circ\}$ to each corner of $G$.
    \item A designation of a face of $G$ as the central face $\FC$.
    \item A designation of a face of $G$ as the outer face $\FO$.
\end{itemize}
For the special case where $v$ has only one incident edge $e$, we view $(e,e)$ as a $360^\circ$ corner. This case does not occur if we consider biconnected graphs.

An ortho-radial representation $\RR=(\phi, \FC, \FO)$ is \emph{drawable} if the representation can be realized as an ortho-radial drawing of $G$ with zero bends, where the angle of each corner matches the assignment $\phi$, the central face $\FC$ contains the origin, and the outer face $\FO$ contains an unbounded region.

Recall that, by the definition of ortho-radial drawing, in an ortho-radial drawing with zero bends, each edge is drawn as a line segment of a straight line passing through the origin or drawn as a circular arc of a circle centered at the origin. We also consider the setting where the \emph{reference edge} $\eref$ is fixed.
In this case, there is an additional requirement that the reference edge $\eref$ has to lie on the outermost circular arc used in the drawing and follows the clockwise direction. If such a drawing exists, we say that $(\RR, \eref)$ is \emph{drawable}. 
See \cref{fig:orthodraw} for an example of a drawing of an ortho-radial representation $\RR$ with the reference edge $\eref=(v_{14},v_5)$. In the figure, we use $\circ$, $\circ \, \circ$, and $\circ \circ \circ$ to indicate a $90^\circ$, a $180^\circ$, and a $270^\circ$ angle assigned to a corner, respectively.

It was shown in~\cite{barth2023topology} that $(\RR, \eref)$ is drawable if and only if the ortho-radial representation~$\RR$ satisfies \ref{item:R1} and \ref{item:R2} with the reference edge $\eref$ does not contain a strictly monotone cycle. Since it is straightforward to test whether \ref{item:R1} and \ref{item:R2} are satisfied in linear time, from now on, unless otherwise stated, we assume that \ref{item:R1} and \ref{item:R2} are satisfied for the ortho-radial representation $\RR$ under consideration.

\begin{figure}[t!]
\centering
\includegraphics[width=0.8\textwidth]{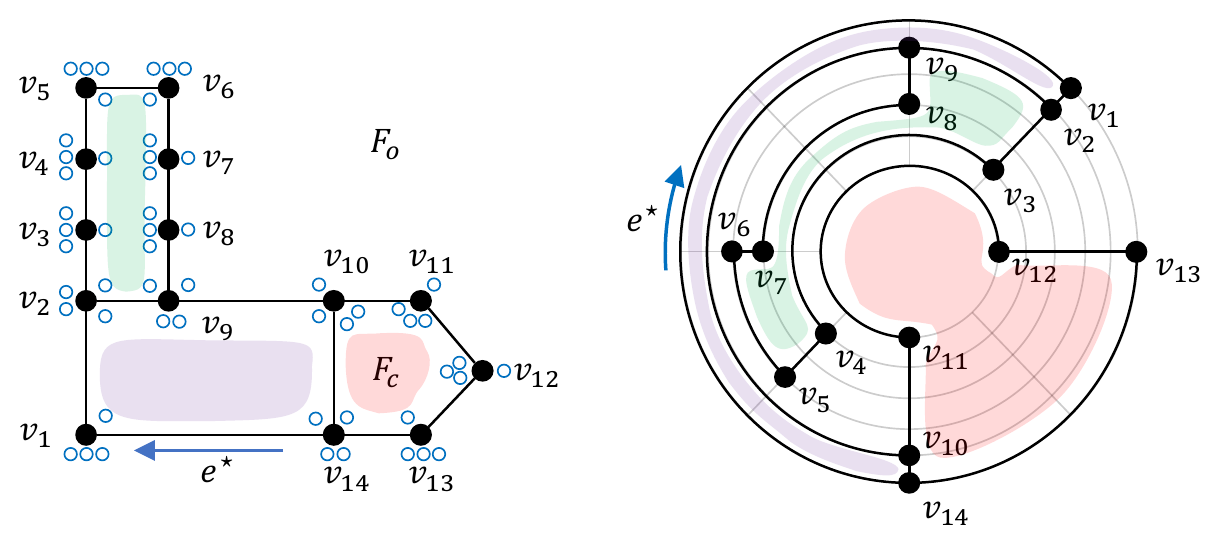}
\caption{A drawing of an ortho-radial representation with a reference edge, where the small blue circles in the left figure denote the angles in the representation that are realized in the right figure.}\label{fig:orthodraw}
\end{figure}

\paragraph{Combinatorial rotations} Consider a length-2 path $P=(u,v,w)$ that passes through $v$ such that $u \neq w$. Given the angle assignment $\phi$,  $P$ is either a $90^\circ$ left turn, a straight line, or a $90^\circ$ right turn. We define the \emph{combinatorial rotation} of $P$ as follows.
\[\rotation(P)=\begin{cases}
			-1, & \text{$P$ is a $90^\circ$ left turn,}\\
            0, & \text{$P$ is a straight line,}\\
            1, & \text{$P$ is a $90^\circ$ right turn.}
		 \end{cases}\]

More formally, let $S = (e_1, \ldots, e_k)$ be the contiguous subsequence of edges starting from $e_1 =\{u,v\}$ and ending at $e_k = \{v,w\}$ in the circular ordering $\EE(v)$ of the undirected edges incident to $v$. Then  $\sum_{j=1}^{k-1} \phi(e_j, e_{j+1}) - 180^\circ$ equals the degree of the turn of $P$ at the intermediate vertex $v$, so the combinatorial rotation of $P$ is $\rotation(P) = \left(\sum_{j=1}^{k-1} \phi(e_j, e_{j+1}) - 180^\circ\right) / \; 90^\circ$.

For the special case where $u=w$, the rotation of $P=(u,v,u)$ can be a  $180^\circ$ left turn, in which case $\rotation(P)=-2$, or a $180^\circ$ right turn, in which case $\rotation(P)=2$. For example, consider the directed edge $e = (u,v)$ where $P$ first goes from $u$ to $v$ along the right side of $e$  and then goes from $v$ back to $u$ along the left side of $e$. Then $P$ is considered a $180^\circ$ left turn, and similarly, $\overline{P}$ is considered a $180^\circ$ right turn. In particular, if $P=(u,v,u)$ is a subpath of a facial cycle $C$, then $P$ is always considered as a $180^\circ$ left turn, and so $\rotation(P)=-2$.

For any single-edge path $P$, we define $\rotation(P) = 0$.
For any crossing-free path $P$ of length more than $2$, we define $\rotation(P)$ to be the sum of the combinatorial rotations of all length-2 subpaths of $P$. Similarly, for any cycle $C$ of length more than $2$, we define $\rotation(C)$ to be the sum of the combinatorial rotations of all length-2 subpaths of $C$.  Same as~\cite{barth2023topology}, based on this notion, we may restate condition \ref{item:R2} as follows.

\begin{enumerate}[(R1$'$)]\setcounter{enumi}{1}
    \item For each face $F$, the combinatorial rotation of its facial cycle $C_F$ satisfies the following:
\[\rotation(C_F)=\begin{cases}
			4, & \text{$F$ is a regular face,}\\
            0, & \text{$F$ is either the central face or the outer face, but not both,}\\
            -4, & \text{$F$ is both the central face and the outer face.}
		 \end{cases}\]
\end{enumerate}

For example, consider the ortho-radial representation shown in \cref{fig:orthodraw}. The path $P=(v_{10}, v_{11}, v_{12}, v_{13}, v_{14})$  has $\rotation(P) = -1$ since it makes two $90^\circ$  left turns and one $90^\circ$ right turn. The cycle $C=(v_{10}, v_{11}, v_{12}, v_{13}, v_{14})$ is the facial cycle of the central face, and it has $\rotation(C) = 0$.

The equivalence between the new and the old definitions of \ref{item:R2} stems from the fact that a $90^\circ$ left turn corresponds to an angle of $270^\circ$.
If $F$ is a regular face with $k$ corners, then in the original definition of \ref{item:R2}, it is required that the sum $s$ of angles around $F$ is $s = (k-2)\cdot 180^\circ$. Since the facial cycle $C_F$ traverses the contour of $F$ in the clockwise direction, the number of $90^\circ$ right turns minus the number of $90^\circ$ left turns must be exactly $4$. Therefore, $s = (k-2)\cdot 180^\circ$ is the same as $\rotation(C_F) = 4$, as each $90^\circ$ right turn contributes $+1$ and each $90^\circ$ left turn contributes $-1$ in the calculation of $\rotation(C_F)$.  

\paragraph{Interior and exterior regions of a cycle}
Any cycle $C$ partitions the remaining graph into two parts. If $C$ is a facial cycle, then one part is empty. The direction of $C$ is clockwise with respect to one of the two parts. The part with respect to which $C$ is clockwise, together with $C$ itself, is called the \emph{interior} of $C$. Similarly, the part with respect to which $C$ is counter-clockwise, together with $C$ itself, is called the \emph{exterior} of $C$. In particular, if a vertex $v$ lies in the interior of $C$, then $v$ must be in the exterior of $\overline{C}$. 

This above definition is consistent with the notion of facial cycle in that any face $F$ is in the interior of its facial cycle $C_F$.
Depending on the context, the interior or the exterior of a cycle can be viewed as a subgraph, a set of vertices, a set of edges, or a set of faces. For example, consider the cycle $C=(v_1, v_2, v_{10}, v_9, v_5)$ of the plane graph shown in \cref{fig:pathscylces}. The interior of~$C$ is the subgraph induced by $v_8$, $v_{11}$, and all vertices in $C$. The exterior of $C$ is the subgraph induced by $v_3$, $v_{4}$, $v_6$, $v_7$, and all vertices in $C$. The cycle $C$ partitions the faces into two parts: The interior of~$C$ contains $F_3$, and the exterior of~$C$ contains $F_1$ and $F_2$.

Let $C$ be a simple cycle oriented in such a way that the outer face $\FO$ lies in its exterior.
Following~\cite{barth2023topology}, we say that  $C$ is \emph{essential} if the central face $\FC$ is in the interior  of $C$. Otherwise we say that  $C$ is \emph{non-essential}. The following lemma was proved in~\cite{barth2023topology}. 

\begin{lemma}[\cite{barth2023topology}]\label{lem-prelim-cycle}
Suppose \ref{item:R1} and \ref{item:R2} are satisfied.
Let $C$ be a simple cycle oriented in such a way that the outer face $\FO$ lies in its exterior, then the combinatorial rotation of $C$ satisfies the following condition.
\[\rotation(C)=\begin{cases}
			4, & \text{$C$ is an essential cycle,}\\
            0, & \text{$C$ is a non-essential cycle.}
		 \end{cases}\]
\end{lemma}

The intuition behind the lemma is that an essential cycle behaves like the facial cycle of the outer face or the central face, and a non-essential cycle behaves like the facial cycle of a regular face.

\paragraph{Subgraphs} When taking a subgraph $H$ of $G$, the combinatorial embedding, angle assignment, central face, and outer face of $H$ are naturally inherited from $G$. More precisely, let $e_1$, $e_2$, and~$e_3$ be three edges incident to $v$, appearing consecutively in the circular ordering $\EE(v)$. If $e_2$ is removed, then the angle assignment for the new corner $(e_1, e_3)$ is determined as  $\phi(e_1, e_2) + \phi(e_2, e_3)$. For example, suppose $\EE(v) = (e_1, e_2, e_3)$ with $\phi(e_1, e_2) = 90^\circ$,  $\phi(e_2, e_3) = 180^\circ$, and $\phi(e_3, e_1) = 90^\circ$ in $G$. If $v$ is incident only to the edges $e_1$ and $e_3$ in $H$, then the angle assignment $\phi_H$ for the two corners surrounding $v$ in $H$ will be  $\phi_H(e_1, e_3) = 270^\circ$ and $\phi_H(e_3, e_1) = 90^\circ$. 

Each face of $G$ is contained in exactly one face of $H$. A face in $H$ can contain multiple faces of $G$. A face of $H$ is said to be the central face if it contains the central face of $G$. Similarly, a face of $H$ is said to be the outer face if it contains the outer face of $G$. 

For example, consider the subgraph $H$ induced by $\{v_2, v_3, \ldots, v_9\}$ in the ortho-radial representation shown in \cref{fig:orthodraw}. In $H$, $v_9$ has only two incident edges $e_1 = \{v_8, v_9\}$ and $e_2 = \{v_2, v_9\}$, and the angle assignment $\phi_H$ for the two corners surrounding $v_9$ in $H$ will be   $\phi_H(e_1, e_2) = 90^\circ$ and $\phi_H(e_2, e_1) = 270^\circ$. The outer face and the central face of $H$ are the same.

\paragraph{Defining direction via reference paths}
 Following~\cite{barth2023topology}, for any two edges $e=(u,v)$ and $e'=(x,y)$, we say that a crossing-free path $P$ is a \emph{reference path} for $e$ and $e'$ if $P$ starts at $u$ or $v$ and ends at $x$ or $y$ such that $P$ does not contain any of the edges in $\{e, \overline{e}, e', \overline{e'}\}$. Given  a reference path $P$ for  $e=(u,v)$ and $e'=(x,y)$, we define the \emph{combinatorial direction} of $e'$ with respect to $e$ and $P$ as follows.   
\[
\direction(e, P, e')=\begin{cases}
\rotation(e \circ P \circ e'), & \text{$P$ starts at $v$ and ends at $x$,}\\
\rotation(\overline{e} \circ P \circ e')+2, & \text{$P$ starts at $u$ and ends at $x$,}\\
\rotation(e \circ P \circ \overline{e'})-2, & \text{$P$ starts at $v$ and ends at $y$,}\\
\rotation(\overline{e} \circ P \circ \overline{e'}), & \text{$P$ starts at $u$ and ends at $y$.}
		 \end{cases}
\]

Here $P \circ Q$ denotes the concatenation of the paths $P$ and $Q$. An edge $e$ is interpreted as a length-1 path. In the definition of $\direction(e, P, e')$, we allow the possibility that a reference path~$P$ consists of a single vertex. If $v=x$ and $u \neq w$, then we may choose $P$ to be the length-0 path consisting of a single vertex $v=x$, in which case $\direction(e, P, e')$ is simply the combinatorial rotation of the length-2 path $(u,v,y)$.
We do not consider the cases where $e = e'$ or $e = \overline{e'}$. 

Consider the reference edge $e = (v_{14},v_1)$ in the ortho-radial representation of \cref{fig:orthodraw}. We measure the direction of $e' = (v_8, v_9)$ from $e$ with different choices of the reference path $P$. If $P = (v_1, v_2, v_9)$, then $\direction(e, P, e')= \rotation(e \circ P \circ \overline{e'})-2 = -1$. If  $P = (v_{14}, v_{10}, v_9)$, then we also have $\direction(e, P, e')= \rotation(\overline{e} \circ P \circ \overline{e'}) = -1$. If we select $P=(v_1, v_2, v_3, v_4, v_5, v_6, v_7, v_8)$, then we get a different value of $\direction(e, P, e')= \rotation({e} \circ P \circ {e'}) = 3$. As we will discuss later, $\direction(e, P, e') \mod 4$ is invariant under the choice of $P$~\cite{barth2023topology}. 

In the definition of  $\direction(e, P, e')$, the additive $+2$ in $\rotation(\overline{e} \circ P \circ e')+2$ is due to the fact that the actual path that we intend to consider is $e \circ \overline{e} \circ P \circ e'$, where we make a $180^\circ$ right turn in $e \circ \overline{e}$, which contributes $+2$ in the calculation of the combinatorial rotation. Similarly, the additive $-2$ in $\rotation(e \circ P \circ \overline{e'})-2$ is due to the fact that the actual path that we intend to consider is $e \circ P \circ \overline{e'} \circ e'$, where we make a $180^\circ$ left turn in $\overline{e'} \circ e'$. There is no additive term in $\rotation(\overline{e} \circ P \circ \overline{e'})$  because of the cancellation of the $180^\circ$ right turn $e \circ \overline{e}$ and the $180^\circ$ left turn $\overline{e'} \circ e'$. The reason why $e \circ \overline{e}$ has to be a right turn and $\overline{e'} \circ e'$ has to be a left turn will be explained later.

See \cref{fig:direction} for an example of the calculation of an edge direction. The direction of $e = (u_1, u_2)$ with respect to $\eref= (v_1,v_2)$  and the reference path  $P=(v_1, v_5, v_4, u_1)$ can be calculated by 
$\rotation(\overline{\eref} \circ P \circ e')+2=1$ according to the formula above, where the additive $+2$ is due to the $180^\circ$ right turn at $\eref \circ \overline{\eref}$.

\paragraph{Edge directions} Imagining that the origin is the south pole, in an ortho-radial drawing with zero bends, each edge $e$ is drawn in one of the following four directions: 
\begin{itemize}
\item $e$ points towards the \emph{north} direction if $e$ is drawn as a line segment of a straight line passing through the origin, where $e$ is directed away from the origin.
\item $e$ points towards the \emph{south} direction if $e$ is drawn as a line segment of a straight line passing through the origin, where $e$ is directed towards the origin.
\item $e$ points towards the \emph{east} direction if $e$ is drawn as a circular arc of a circle centered at the origin in the clockwise direction.
\item $e$ points towards the \emph{west} direction if $e$ is drawn as a circular arc of a circle centered at the origin in the counter-clockwise direction.
\end{itemize}

We say that $e$ is a \emph{vertical} edge if $e$ points towards north or south.  Otherwise, we say that $e$ is a \emph{horizontal} edge. We argue that as long as  \ref{item:R1} and \ref{item:R2} are satisfied,  the direction of any edge $e$ is uniquely determined by the ortho-radial representation $\RR$ and the reference edge $\eref$.

For the reference edge $\eref$, it is required that $\eref$ points east, and so $\overline{\eref}$ points west. 
Consider any edge $e$ that is neither $\eref$ nor $\overline{\eref}$. It is clear that the value of $\direction(\eref, P, e)$ determines the direction of $e$ in that the direction of $e$ is forced to be east, south, west, or north if $\direction(\eref, P, e) \mod 4$ equals 0, 1, 2, or 3, respectively. For example, in the ortho-radial representation of \cref{fig:orthodraw}, the edge $e' = (v_8, v_9)$ is a vertical edge in the north direction, as we have calculated that  $\direction(\eref, P, e') \mod 4 = 3$.

\begin{lemma}[\cite{barth2023topology}]\label{lem-prelim-dir-1}
For any two edges $e$ and $e'$, the value of $\direction(e, P, e') \mod 4$ is invariant under the choice of the reference path $P$.
\end{lemma}

The above lemma shows that $\direction(\eref, P, e) \mod 4$ is invariant under the choice of the reference path $P$, so the direction of each edge in an ortho-radial representation is well-defined, even for the case that  $(\RR, \eref)$ might not be drawable. Given the reference edge $\eref$, we let $\EH$ denote the set of all horizontal edges in the east direction, and let $\EV$ denote the set of all vertical edges in the north direction.

\paragraph{Horizontal segments} 
We require that in a drawing of $(\RR, \eref)$, the reference edge $\eref$ lies on the outermost circular arc used in the drawing, so not every edge in $\overline{C_{\FO}}$ is eligible to be a reference edge. To determine whether an edge $e \in \overline{C_{\FO}}$ is eligible to be a reference edge, we need to introduce some terminology.

Given the reference edge $\eref$, the set $\EV$ of vertical edges in the north direction  and the set $\EH$ of horizontal edges in the east direction  are fixed.
Let $\PH$ denote the set of connected components induced by $\EH$. Each component $S \in \PH$ is either a path or a cycle, and so in any drawing of $\RR$, there is a circle $C$ centered at the origin such that $S$ must be drawn as $C$ or a circular arc of $C$. We call each component $S \in \PH$ a \emph{horizontal segment}.

Each horizontal segment $S \in \PH$ is written as a sequence of vertices $S=(v_1, v_2, \ldots, v_s)$, where $s$ is the number of vertices in $S$, such that $(v_i, v_{i+1}) \in \EH$ for each $1 \leq i < s$. If $S$ is a cycle, then we additionally have $(v_s, v_1) \in \EH$, so $S=(v_1, v_2, \ldots, v_s)$ is a circular order. When $S$ is a cycle, we use modular arithmetic on the indices so that $v_{s+1} = v_1$.
We write $\Nnorth(S)$ to denote the set of vertical edges $e = (x,y) \in \EV$ such that $x \in S$. Similarly, $\Nsouth(S)$ is the set of vertical edges $e = (x,y) \in \EV$ such that $y \in S$. We assume that the edges  in $\Nnorth(S)$ and $\Nsouth(S)$ are ordered according to the indices of their endpoints in $S$. The ordering is sequential if $S$ is a path and is circular if $S$ is a cycle. 
Consider the ortho-radial representation $\RR$ given in \cref{fig:orthodraw} as an example. The horizontal segment $S = (v_{10}, v_9, v_2)$ has $\Nsouth(S) = ((v_{11}, v_{10}), (v_8, v_9), (v_3, v_2))$ and $\Nnorth(S) = ((v_{10}, v_{14}), (v_2 ,v_1))$.

Observe that $\Nnorth(S) = \emptyset$ for the horizontal segment $S\in \PH$ that contains $\eref$ is a necessary condition that a drawing of  $\RR$ where $\eref$ lies on the outermost circular arc exists. This condition can easily be checked  in linear time.

\paragraph{Spirality} Intuitively, $\direction(e, P, e')$ quantifies the degree of \emph{spirality} of $e'$ with respect to $e$ and $P$. Unfortunately,   \cref{lem-prelim-dir-1} does not hold if we replace $\direction(e, P, e') \mod 4$ with $\direction(e, P, e')$. A crucial observation made in~\cite{barth2023topology} is that such a replacement is possible if we add some restrictions about the positions of $e$, $e'$, and $P$. See the following lemma.

\begin{lemma}[\cite{barth2023topology}]\label{lem-prelim-dir-2}
Let $C$ and $C'$ be essential cycles such that $C'$ lies in the interior of $C$.
Let $e$ be an edge on $C$.
Let $e'$ be an edge on $C'$.
The value of $\direction(e, P, e')$ is invariant under the choice of the reference path $P$, over all paths $P$ in the interior of $C$ and in the exterior of $C'$.
\end{lemma}

Recall that we require a reference path to be crossing-free. This requirement is crucial in the above lemma. If we allow $P$ to be a general path that is not crossing-free, then we may choose $P$ in such a way that $P$ repeatedly traverses a \emph{non-essential} cycle many times, so that $\direction(e, P, e')$ can be made arbitrarily large and arbitrarily small.

Setting $e = \eref$ and $C = \overline{C_{\FO}}$ in the above lemma, we infer that $\direction(\eref, P, e')$ is determined once we fix an essential cycle $C'$ that contains $e'$ and only consider reference paths $P$ that lie in the exterior of $C'$. The condition for the lemma is satisfied because $\overline{C_{\FO}}$ is the outermost essential cycle in that all other essential cycles are in the interior of $\overline{C_{\FO}}$. The reason why we set $C = \overline{C_{\FO}}$ and not $C = C_{\FO}$ is that $\FO$ has to be in the exterior of $C$. Note that the assumption that $G$ is biconnected ensures that each facial cycle is simple.

\begin{figure}[t!]
\centering
\includegraphics[width=\textwidth]{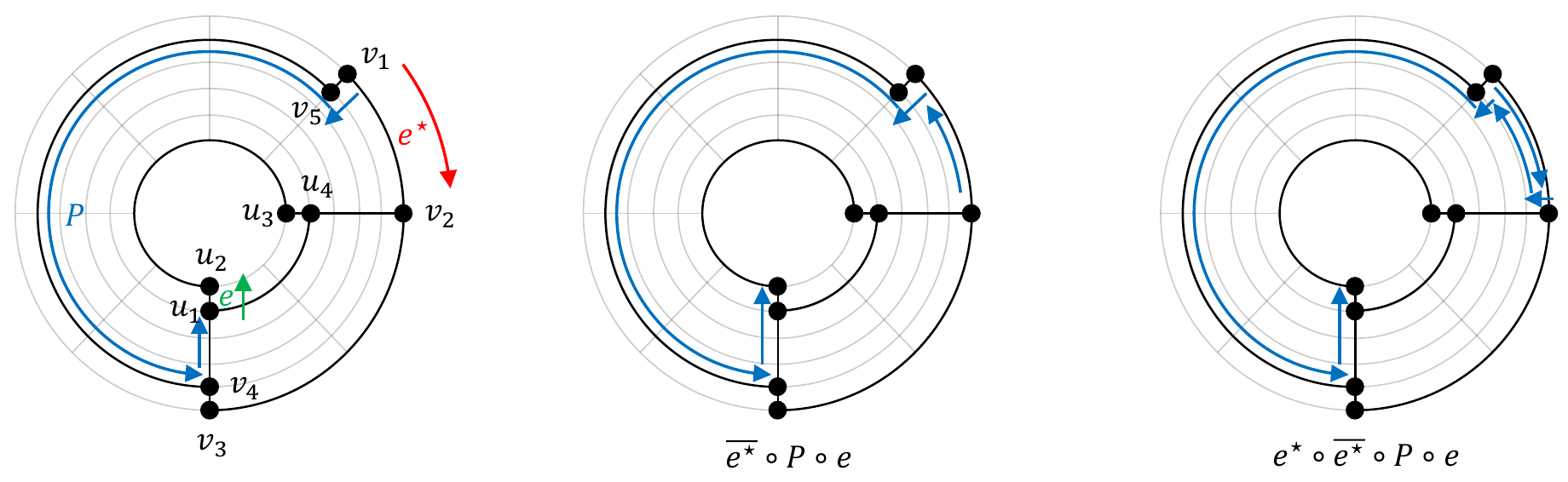}
\caption{The calculation of $\direction(\eref, P, e)$.}\label{fig:direction}
\end{figure}

Let $C$ be an essential cycle and let $e$ be an edge in $C$.
In view of the above, following~\cite{barth2023topology}, we define the \emph{edge label} $\ell_C(e)$ of $e$ with respect to $C$ to be the value of $\direction(\eref, P, e)$, for any choice of reference path $P$ in the exterior of $C$. For the special case that $e = \eref$ and $C = \overline{C_{\FO}}$, we let $\ell_C(e) = 0$. Intuitively, the value $\ell_C(e)$ quantifies the degree of spirality of $e$ from $\eref$ if we restrict ourselves to the exterior of $C$.
Consider the edge $e=(u_1, u_2)$ in the essential cycle $C = (u_1, u_2, u_3, u_4)$ in \cref{fig:direction} as an example. We have $\ell_C(e) = \direction(\eref, P, e) = 1$, since the reference path $P=(v_1, v_5, v_4, u_1)$ lies in exterior of $C$. 

We briefly explain the formula of $\direction(e, P, e')$: As discussed earlier, in the definition of  $\direction(e, P, e')$, the additive $+2$ in $\rotation(\overline{e} \circ P \circ e')+2$ is due to the fact that the actual path that we want to consider is $e \circ \overline{e} \circ P \circ e'$, where we make a $180^\circ$ right turn in $e \circ \overline{e}$. The reason why $e \circ \overline{e}$ has to be a right turn is because of the scenario considered in \cref{lem-prelim-dir-2}, where $e$ is an edge in $C$.
To ensure that we stay in the interior of $C$ in the traversal from $e$ to $e'$ via the path $e \circ \overline{e} \circ P \circ e'$, the $180^\circ$ turn of $e \circ \overline{e}$ has to be a right turn. The remaining part of the formula of $\direction(e, P, e')$ can be explained similarly.

\paragraph{Monotone cycles}
We are now ready to define the notion of strictly monotone cycles used in \ref{item:R3}.
We say that an essential cycle $C$ is \emph{monotone} if all its edge labels $\ell_C(e)$ are non-negative or all its edge labels $\ell_C(e)$ are non-positive. 
Let $C$ be an essential cycle that is monotone.
If $C$ contains at least one positive edge label, then we say that $C$ is \emph{increasing}. If $C$ contains at least one negative edge label, then we say that $C$ is \emph{decreasing}. Decreasing cycles and increasing cycles are collectively called \emph{strictly monotone}.

Intuitively, an increasing cycle is like a loop of descending stairs, and a decreasing cycle is like a loop of ascending stairs, so they are not drawable.
It was proved in~\cite{barth2023topology} that $(\RR, \eref)$ is drawable if and only if it does not contain a strictly monotone cycle. Recall again that, throughout the paper, unless otherwise stated, we assume that the given ortho-radial representation already satisfies \ref{item:R1} and \ref{item:R2}.

\begin{lemma}[\cite{barth2023topology}]\label{lem-monotone-cycle}
An ortho-radial representation $\RR$, with a fixed reference edge $\eref$ such that $\Nnorth(S) = \emptyset$ for the horizontal segment $S\in \PH$ that contains $\eref$, is drawable if and only if it does not contain a strictly monotone cycle.
\end{lemma}

In general, whether $(\RR, \eref)$ is drawable depends on the choice of the reference edge~$\eref$. For instance, in \cref{fig:smc}, while $(\RR, \eref)$ is drawable, $(\RR, e)$ is not, as the essential cycle $C = (v_1, v_2, \ldots, v_{10})$ is increasing. With $e$ as the reference edge, all the edge labels on the cycle $C$ are non-negative, with at least one being positive.

\begin{figure}[t!]
\centering
\includegraphics[width=0.7\textwidth]{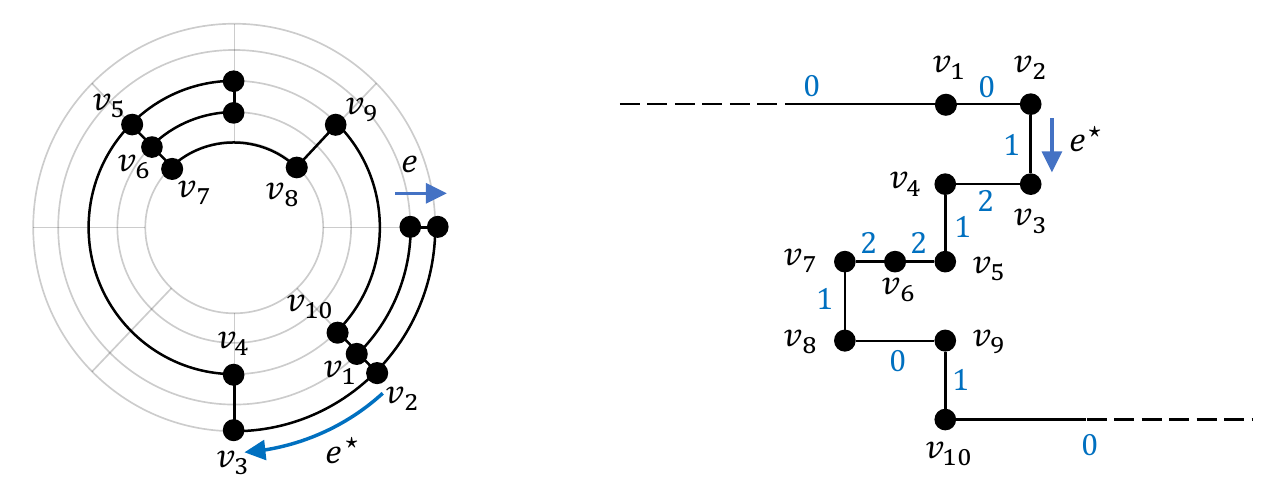}
\caption{Changing the reference edge to $e$ leads to a strictly monotone cycle.}\label{fig:smc}
\end{figure}

We are ready to state our main results.

\begin{restatable}{theorem}{thmone}\label{thm-main1}\RestateRemark
There is an $O(n \log n)$-time algorithm $\mathcal{A}$ that outputs either a drawing of $(\RR, \eref)$ or a strictly monotone cycle of  $(\RR, \eref)$, for any given ortho-radial representation $\RR$ of an $n$-vertex biconnected simple graph, with a fixed reference edge $\eref$  such that $\Nnorth(S) = \emptyset$ for the horizontal segment $S\in \PH$ that contains $\eref$.
\end{restatable}


The above theorem improves the previous algorithm of~\cite{barth2023topology} which costs $O(n^2)$ time.
If the output of $\mathcal{A}$ is a strictly monotone cycle, then the cycle certifies the non-existence of a drawing, by \cref{lem-monotone-cycle}. We also extend the above theorem to the case where the reference edge is not fixed.

\begin{restatable}{theorem}{thmtwo}\label{thm-main2}\RestateRemark
There is an $O(n \log^2 n)$-time algorithm $\mathcal{A}$ that decides whether an ortho-radial representation $\RR$ of an $n$-vertex biconnected simple graph is drawable. If $\RR$ is drawable, then $\mathcal{A}$ also computes a drawing of $\RR$.
\end{restatable}


The proof of \cref{thm-main1} is in \cref{sect:cycle}, and the proof of \cref{thm-main2} is in \cref{sect:binary-search}. 

\section{Ortho-radial drawings via good sequences}\label{sect:drawing}

In this section, we introduce the notion of a good sequence, whose existence enables us to construct an ortho-radial drawing through a simple greedy algorithm. Intuitively, a good sequence $A = (S_1, S_2, \ldots, S_{k})$ is a sequence of horizontal segments that allows us to safely place the horizontal segments one by one: $S_1$ is drawn on the circle $r=k$, $S_2$ is drawn on the circle $r = k-1$, and so on.

\paragraph{Sequences of horizontal segments} Let $A = (S_1, S_2, \ldots, S_{k})$ be any sequence of $k$ horizontal segments. In general, we do not require $A$ to cover the set of all horizontal segments in $\PH$. We consider the following terminology for each $1 \leq i \leq k$, where $k$ is the length of the sequence $A$.
\begin{itemize}
    \item Let $G_i$ be the subgraph of $G$ induced by the horizontal edges in $S_1, S_2, \ldots, S_{i}$ and the set of all vertical edges whose both endpoints are in $S_1, S_2, \ldots, S_{i}$.
 Let $F_i$ be the central face of $G_i$, and let $C_i$ be the facial cycle of $F_i$.
    \item We extend the notion $\Nsouth(S)$ to a sequence of horizontal segments, as follows. Let $\Nsouth(S_1, S_2, \ldots, S_{i})$ be the set of vertical edges $e = (x,y) \in \EV$ such that $y \in C_i$ and $x \notin C_i$. 
    \item Let $G_i^+$ be the subgraph of $G$ induced by all the edges in $G_i$ together with the edge set $\Nsouth(S_1, S_2, \ldots, S_{i})$. Let $F_i^+$ be the central face of $G_i^+$, and let $C_i^+$ be the facial cycle of $F_i^+$.
\end{itemize}

Observe that for each vertical edge $e = (x,y) \in \Nsouth(S_1, S_2, \ldots, S_{i})$, the south endpoint~$x$ appears exactly once in $C_i^+$. We circularly order the edges $e = (x,y) \in \Nsouth(S_1, S_2, \ldots, S_{i})$ according to the position of the south endpoint $x$ in the circular ordering of $C_i^+$.
Take the graph $G = G_6$ in \cref{fig:good-drawing} as an example. In this graph, there are $6$ horizontal segments, shaded in \cref{fig:good-drawing}: 
\begin{align*}
    &S_1=(v_{1,1}, v_{1,2}, v_{1,3}, v_{1,4}),
    &&S_2=(v_{2,1}, v_{2,2}, v_{2,3}),
    &&S_3=(v_{3,1}, v_{3,2}, v_{3,3}, v_{3,4}, v_{3,5}),\\
    &S_4=(v_{4,1}, v_{4,2}, v_{4,3}),
    &&S_5=(v_{5,1}, v_{5,2}, v_{5,3}, v_{5,4}, v_{5,5}),
    &&S_6=(v_{6,1}, v_{6,2}).
\end{align*}
With respect to the sequence $A=(S_1, S_2, \ldots, S_6)$, \cref{fig:good-drawing} shows the graphs $G_i$ and $G_{i}^+$, for all $1 \leq i \leq 6$. For example, for $i=2$, we have:
\begin{align*}
  \Nsouth(S_1, S_2)&=((v_{3,1},v_{1,1}) (v_{3,2},v_{2,1}), (v_{3,4},v_{2,3}), (v_{3,5},v_{1,4})), \\
  \Nnorth(S_2)&=((v_{2,1}, v_{1,2}),(v_{2,2},v_{1,3})) \\
  C_2 &= (v_{1,1}, v_{1,2}, v_{2,1}, v_{2,2}, v_{2,3}, v_{2,2}, v_{1,3}, v_{1,4}),\\
  C_2^+ &= (v_{1,1}, v_{3,1}, v_{1,1}, v_{1,2}, v_{2,1}, v_{3,2},v_{2,1}, v_{2,2}, v_{2,3}, v_{3,4},v_{2,3}, v_{2,2}, v_{1,3}, v_{1,4}, v_{3,5},v_{1,4}).
\end{align*}
Here $\Nsouth(S_1, S_2)$, $C_2$, and $C_2^+$ are circular orderings, and $\Nnorth(S_2)$ is a sequential ordering, as $S_2$ is a path.

\begin{figure}[t!]
\centering
\includegraphics[width=\textwidth]{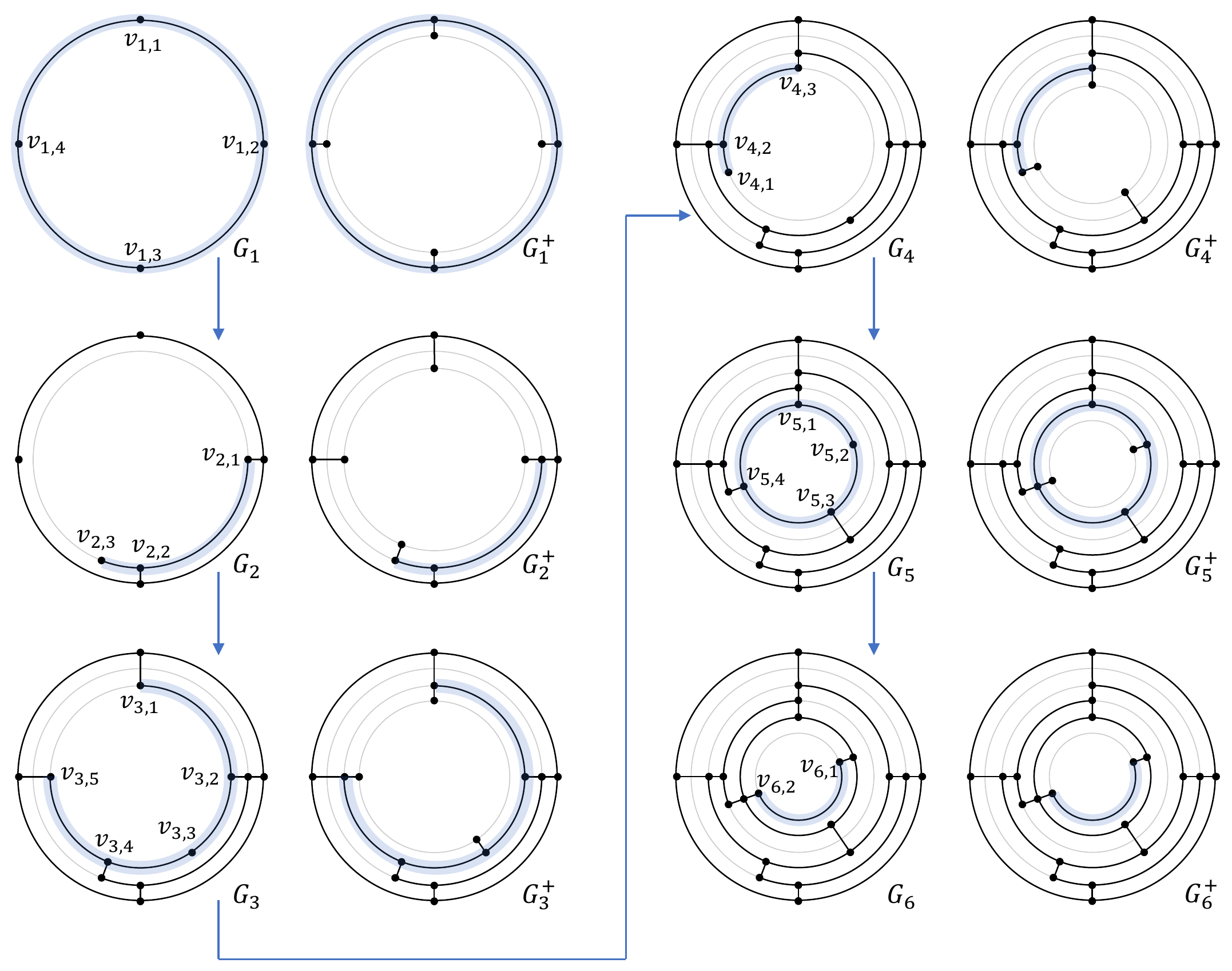}
\caption{Constructing a good drawing for a good sequence.}\label{fig:good-drawing}
\end{figure}

\paragraph{Good sequences} 
We say that a sequence of horizontal segments $A = (S_1, S_2, \ldots, S_{k})$ is \emph{good} if $A$ satisfies the following conditions.
\begin{enumerate}[(S1)]
    \item \label{item:S1} $S_1$ is the reversal of the facial cycle of the outer face $\FO$, i.e., $S_1 = \overline{C_{\FO}}$.
    \item \label{item:S2} For each $1 < i \leq k$, $\Nnorth(S_i)$ satisfies the following requirements.
\begin{itemize}
    \item $\Nnorth(S_i) \neq \emptyset$.
    \item If $S_i$ is a path, then $\Nnorth(S_i)$ is a contiguous subsequence of $\Nsouth(S_1, S_2, \ldots, S_{i-1})$.
    \item If $S_i$ is a cycle, then $\Nnorth(S_i) = \Nsouth(S_1, S_2, \ldots, S_{i-1})$. 
\end{itemize}    
\end{enumerate}

Clearly, if $A=(S_1, S_2, \ldots, S_{k})$ is good, then $(S_1, S_2, \ldots, S_{i})$ is also good for each $1 \leq i < k$. In general, a good sequence might not exist for a given $(\RR,\eref)$. In particular, in order to satisfy \ref{item:S1}, it is necessary that the cycle $\overline{C_{\FO}}$ is a horizontal segment. The sequence $A=(S_1, S_2, \ldots, S_6)$ shown in \cref{fig:good-drawing} is a good sequence.

\paragraph{Good drawings} Throughout the paper, we use the polar coordinate system, where $(r,\theta)$ is the point given by $x = r \cos{\theta}$ and $y= r \sin{\theta}$ in the Cartesian coordinate system. We always have $r \geq 0$.
Let $A = (S_1, S_2, \ldots, S_{k})$ be a good sequence of $k$ horizontal segments. For notational simplicity, we may also write $\Nsouth(A) = \Nsouth(S_1, S_2, \ldots, S_{k})$. Let $(e_1, e_2, \ldots, e_s)$ be the circular ordering of $\Nsouth(A)$, and let $e_j=(x_j, y_j)$, for each $1 \leq j \leq s$, where $s$ is the size of $\Nsouth(A)$.  We say that an ortho-radial drawing of $G_k$ with zero bends is \emph{good} if the drawing satisfies the following property. 
\begin{enumerate}[(D1)]
\item \label{item:D1} For each $1 \leq j \leq s$, the drawing does not use any point in $\{ (r,\theta) \mid 0\leq r < r_j \; \text{and} \; \theta = \theta_j\}$, where we let $(r_j, \theta_j)$ denote the position of  vertex $y_j$ in the drawing.
\end{enumerate}
It is implicitly required that a good drawing must be planar and preserve the combinatorial embedding of the plane graph $G_k$. In \cref{fig:good-drawing}, the drawing of the graph $G_i$, for each $1 \leq i \leq 6$, is a good drawing for the good sequence $(S_1, S_2, \ldots, S_i)$. In the following lemma, we show that any good drawing has the following favorable property.


\begin{enumerate}[(D1)]\setcounter{enumi}{1}
\item \label{item:D2} The clockwise circular ordering of $e_1, e_2, \ldots, e_s$ given by $\theta_1, \theta_2, \ldots, \theta_s$ in the drawing is the same as the circular ordering given by $\Nsouth(A)$.
\end{enumerate}

 \begin{lemma}\label{lem:good-drawing-condition}
If an ortho-radial drawing of $G_k$ for a good sequence $A=(S_1, S_2, \ldots, S_{k})$ satisfies \ref{item:D1}, then the drawing also satisfies \ref{item:D2}.
 \end{lemma}
 \begin{proof}
 See \cref{fig:D-condition} for an illustration of the proof.
 Suppose \ref{item:D2} is not satisfied, then we can find three indices $a$, $b$, and $c$ such that the clockwise ordering $(e_c, e_b, e_a)$  given  by their $\theta$-coordinates is in the opposite direction of their circular ordering $(e_a, e_b, e_c)$ in $\Nsouth(A) = (e_1, e_2, \ldots, e_s)$. 
 
 Let $G_k^\ast$ be the graph resulting from identifying the south endpoints of all the edges in $\Nsouth(A) = (e_1, e_2, \ldots, e_s)$ into a vertex $v^\ast$. 
 A planar drawing of $G_k^\ast$ can be found by extending the given ortho-radial drawing of $G_k$ by placing $v^\ast$ at the origin and drawing all the edges in $\Nsouth(A) = (e_1, e_2, \ldots, e_s)$ as straight lines. By \ref{item:D1}, the drawing of $G_k^\ast$ is crossing-free. Assuming that \ref{item:D2} is not satisfied, we will derive a contradiction by showing that this drawing cannot be crossing-free, so \ref{item:D2} must be satisfied.
 
For any $1 \leq i \leq s$ and $1 \leq j \leq s$, we write $P_{i,j}$ to denote the subpath of $C_k^+$ starting at $e_i$ and ending at $\overline{e_j}$. Any such a path in $G_k^\ast$ is a cycle, as it starts and ends at the same vertex $v^\ast$.

Consider the cycle $P_{a,b}$ in $G_k^\ast$. Our assumption on the $\theta$-coordinates for $\{e_a, e_b, e_c\}$ implies that $e_c$ lies in the interior of the cycle $P_{a,b}$ in the above drawing of $G_k^\ast$. Now consider the path $P_{c,a}$, which starts at $e_c$ and ends at $\overline{e_a}$. Let $v$ be the first vertex of $P_{a,b} - \{v^\ast\}$ that $P_{c,a}$ visits. Since  $P_{c,a}$ ends at $\overline{e_a}$, such a vertex exists. Let $e$ be the edge incident to $v$ from which $P_{c,a}$ enters~$v$. Since $C_k^+$ is a facial cycle of $G_k^+$, the circular ordering of the incident edges of $v$ in $C_k^+$  must respect the counter-clockwise ordering given by $\EE(v)$, so $e$ must be an edge in the exterior of the cycle $P_{a,b}$. Therefore, there exist an edge of $P_{c,a}$ and an edge of $P_{a,b}$ crossing each other, since otherwise $P_{c,a}$ cannot go from the interior of $P_{a,b}$ to the exterior of $P_{a,b}$. 
 \end{proof}

\begin{figure}[t!]
\centering
\includegraphics[width=0.4\textwidth]{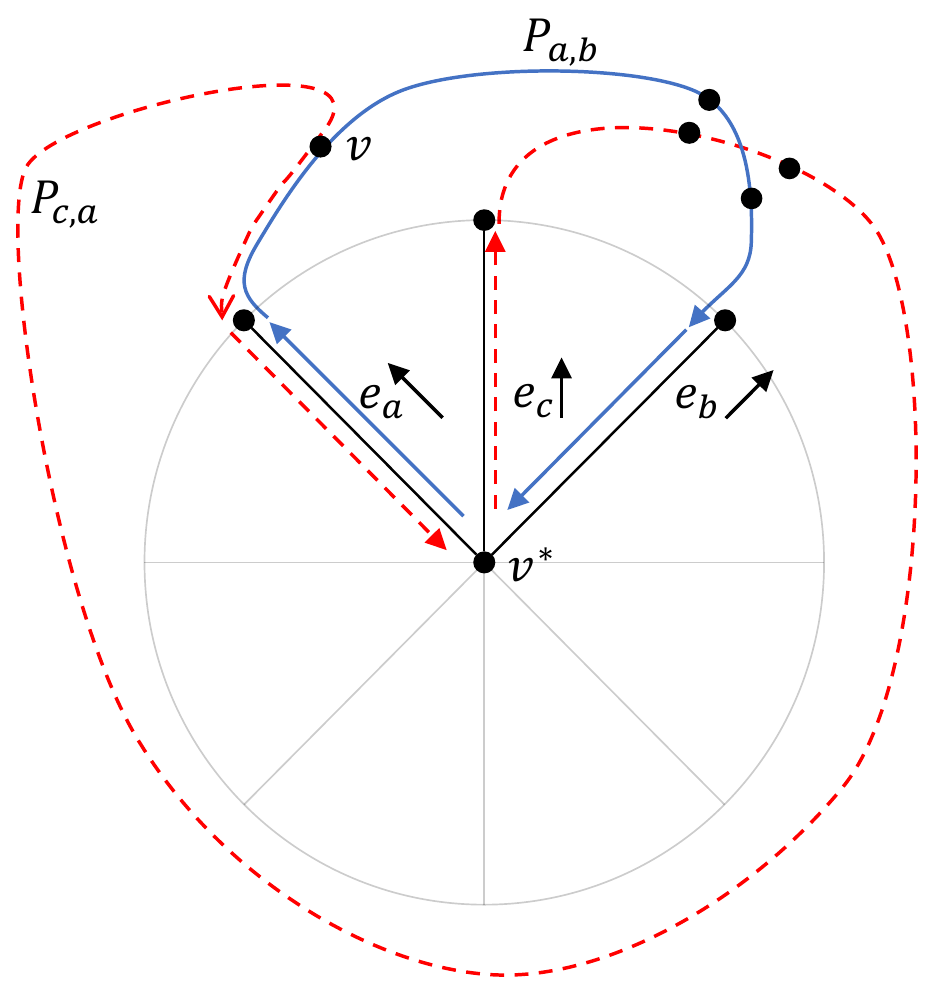}
\caption{Illustration for the proof of \cref{lem:good-drawing-condition}.}\label{fig:D-condition}
\end{figure}

 We show an efficient algorithm that computes a good drawing of $G_k$ for a given good sequence $A=(S_1, S_2, \ldots, S_{k})$. The time complexity of the algorithm is linear in the size of $G_k$. For the special case that $G_k = G$, this gives a linear-time algorithm for computing an ortho-radial drawing realizing the given $(\RR,\eref)$.
 
 \begin{lemma}\label{lem:good-drawing}
A good drawing of $G_k$ for a given good sequence $A=(S_1, S_2, \ldots, S_{k})$ can be constructed in time $O\left(\sum_{i=1}^k |S_i|\right)$.
 \end{lemma}
\begin{proof} The lemma is proved by an induction on the length $k$  of the sequence $A$. Refer to \cref{fig:good-drawing} for an illustration of the algorithm described in the proof.

\paragraph{Base case}
For the base case of $k=1$, a good drawing of $G_1 = S_1$ can be constructed, as follows. By \ref{item:S1},  $S_1 = \overline{C_{\FO}}$, which is the outermost essential cycle. Let $S_1 = (v_1, v_2, \ldots, v_t)$, where $t = |S_1|$ is the number of vertices in the cycle $S_1$. Then we may draw $G_1 = S_1$ on the unit circle by putting $v_j$ on the point $\left(1, -\frac{j}{t} \cdot 2\pi \right)$, for each $1 \leq j \leq t$. The minus sign is due to the fact that $S_1$ is oriented in the clockwise direction. The construction of the drawing takes $O\left(|S_1|\right)$ time as we need to compute these coordinates. Condition \ref{item:D1} is satisfied because the drawing does not use any point $(r,\theta)$ with $0\leq r<1$.

\paragraph{Inductive step}
For the inductive step, given that we have a good drawing of $G_{k-1}$, we will extend this drawing to a good drawing of $G_{k}$ by spending $O\left(|S_k|\right)$ time to properly assign the coordinates to the vertices in $S_k$. We select $r^\ast > 0$ to be any number that is smaller than the $r$-coordinates of the positions of all vertices in $G_{k-1}$ in the given drawing, so the circle $r = r^\ast$ is strictly contained in the central face $F_{k-1}$ of $G_{k-1}$ in the given drawing. 
Let $S_k = (v_1, v_2, \ldots, v_t)$, where $t = |S_k|$ is the number of vertices in $S_k$.  We will draw $S_k$ on the circle $r = r^\ast$.

\paragraph{Step 1: vertices with neighbors in the given drawing}
Let $(e_1, e_2, \ldots, e_s)$ be the circular ordering of $\Nsouth(S_1, S_2, \ldots, S_{k-1})$, and let $e_j=(x_j, y_j)$, for each $1 \leq j \leq s$, where $s$ is the size of $\Nsouth(S_1, S_2, \ldots, S_{k-1})$.  We let $(r_j, \theta_j)$ denote the position of vertex $y_j$ in the given drawing.

By \ref{item:S2}, $\Nnorth(S_k)$ is a subset of $\Nsouth(S_1, S_2, \ldots, S_{k-1})$. 
For each vertical edge $e_j = (x_j,y_j) \in \Nnorth(S_k)$, We assign the coordinates $(r^\ast, \theta_j)$ to $x_j$. By \ref{item:D1}, the given drawing does not use any point in $\{ (r,\theta) \mid 0\leq r < r_j \; \text{and} \; \theta = \theta_j\}$, so we may draw $e_j$ as a straight line connecting $x_j$ and~$y_j$.

By \cref{lem:good-drawing-condition}, \ref{item:D2} follows from \ref{item:D1}. By \ref{item:D2}, for the vertices in $S_k$ that we have drawn, that is, the set of vertices in $S_k$ that have incident edges in $\Nnorth(S_k)$, the clockwise ordering of their $\theta$-coordinates respect the ordering of the horizontal segment $S_k = (v_1, v_2, \ldots, v_t)$. If $S_k$ is a cycle, then $(v_1, v_2, \ldots, v_t)$ is seen as a circular ordering.

\paragraph{Step 2: the two endpoints} For the case $S_k$ is a path, we draw the two endpoints $v_1$ and $v_t$ of $S_k = (v_1, v_2, \ldots, v_t)$, as follows. By \ref{item:S2}, in this case, $\Nnorth(S_k)$ is a contiguous subsequence of $\Nsouth(S_1, S_2, \ldots, S_{k-1})$. Let $j_1$ and $j_2$ be the indices such that the subsequence starts at $e_{j_1}$ and ends at $e_{j_2}$. Let $\epsilon = \min_{1 \leq j \leq s} (\theta_j - \theta_{j-1})/3$.
If $v_1$ does not have an incident edge in $\Nnorth(S_k)$, then we assign the coordinates $\left(r^\ast, \theta_{j_1}+\epsilon \right)$ to $v_1$.
Similarly, if $v_t$ does not have an incident edge in $\Nnorth(S_k)$, then we assign the coordinates $\left(r^\ast, \theta_{j_2}-\epsilon \right)$ to $v_t$.
Our choice of $\epsilon$ ensures that the range $[\theta_{j_2}-\epsilon, \theta_{j_1}+\epsilon]$ of radians does not overlap with  $\theta_j$, for any $e_j \in \Nsouth(S_1, S_2, \ldots, S_{k-1}) \setminus \Nnorth(S_k)$.

\paragraph{Step 3: remaining vertices}
We draw the remaining vertices of $S_k$ as follows. 
Let $(v_{a}, \ldots, v_b)$ be any maximal-length contiguous subsequence of $S_k$ consisting of vertices that have not been drawn yet. We may simply draw them by placing them between $v_{a-1}$ and $v_{b+1}$ on the circle $r = r^\ast$. Formally, let $\thetawest$ be the $\theta$-coordinate of the position of $v_{a-1}$, and let $\thetaeast$ be the $\theta$-coordinate of the position of $v_{b+1}$.  
For each $a \leq j \leq b$, the coordinates of $v_j$ are assigned to be \[\left(r^\ast, \thetawest - (j-a+1) \cdot \frac{\thetaeast - \thetawest}{b-a+2}\right).\]
In general, it is possible to have $v_{a-1}=v_{b+1}$ when $S_k$ is a cycle and $|\Nnorth(S_k)| = 1$, in which case $v_{a-1}=v_{b+1}$ is the vertex in $S_k$ incident to the only edge in $\Nnorth(S_k)$. Note that this case cannot occur when the underlying graph $G$ is biconnected. For this case, we should let the $\theta$-coordinate of the position of $v_{b+1}$ to be the $\theta$-coordinate of the position of $v_{a-1}$ minus $2\pi$ in the above calculation.

\paragraph{Validity of the drawing} For the drawing of $G_k$ that we construct, we verify that condition \ref{item:D1} is satisfied. 
Consider any vertex $v$ in $G_k$ that has an incident edge in $\Nsouth(S_1, S_2, \ldots, S_k)$. Suppose that its coordinates in our drawing are $(r_v, \theta_v)$. To prove that \ref{item:D1} is satisfied, we just need to verify that our drawing does not use any point $(r,\theta)$ with $0\leq r < r_v$ and $\theta = \theta_v$. For the case $v$ is in $S_k$, we have $r_v = r^\ast$, and our choice of $r^\ast$ implies that our drawing does not use any point whose $r$-value is smaller than $r^\ast$. 

Now suppose that $v$ is not in $S_k$.  Then $v=y_j$ for some $e_j=(x_j, y_j) \in \Nsouth(S_1, S_2, \ldots, S_{k-1}) \setminus \Nnorth(S_k)$. Note that this case is possible only when $S_k$ is a path. By the induction hypothesis, the given drawing of $G_{k-1}$ does not use any point $(r,\theta)$ with $0\leq r < r_v$ and $\theta = \theta_v$, so we just need to verify that when we draw the horizontal segment $S_k$, the circular arc used to draw $S_k$ does not cross the line $\{ (r,\theta) \mid 0\leq r < r_v \; \text{and} \; \theta = \theta_v\}$. Indeed, the $\theta$-coordinates of this circular arc are confined to the range $[\theta_{j_2}-\epsilon, \theta_{j_1}+\epsilon]$, and our choice of $\epsilon$ in Step 2 ensures that this range does not overlap with  $\theta_j$, for any $e_j \in \Nsouth(S_1, S_2, \ldots, S_{k-1}) \setminus \Nnorth(S_k)$, so such a crossing is impossible. 

\paragraph{Runtime analysis} A good drawing of $G_1$ can be constructed in $O\left(|S_1|\right)$ time. Given a good drawing of  $G_{i-1}$, a good drawing of $G_{i}$ can be constructed in $O\left(|S_i|\right)$ time. Therefore,  given good sequence $A=(S_1, S_2, \ldots, S_{k})$, a good drawing of $G_k$ can be constructed in $O\left(\sum_{i=1}^k |S_i|\right)$ time.
\end{proof}

\paragraph{Remark} The drawing computed by the algorithm of \cref{lem:good-drawing} uses $k$ layers (i.e., concentric circles). It is possible to modify the algorithm so that the output is a drawing with the smallest number of layers. The idea is simply that when we process a new horizontal segment $S_i$ in the inductive step, instead of always creating a new layer, we draw $S_i$ in the outermost possible layer. More formally, we define the layer number $\ell_i$ for $S_i$ as follows. For the base case, we let $\ell_i = 0$. For the inductive step, we let $\ell_i = \ell_j + 1$, where $j$ is the index maximizing $\ell_j$ such that there exists a vertical edge whose south endpoint is in $S_i$ and whose north endpoint is in $S_j$. By the definition of the layer numbers,  any ortho-radial drawing requires at least $\ell^\ast = 1 + \max_{i \in [k]} \ell_i$ layers. An ortho-radial drawing with $\ell^\ast$ layers can be constructed by modifying the algorithm of \cref{lem:good-drawing} in such a way that $S_i$ is drawn on the circle $r = 1 - \ell_i/\ell^\ast$, which is the outermost possible layer where $S_i$ can be drawn.

\section{Constructing a good sequence}\label{sect:sequence}

 Given an ortho-radial representation $\RR$ of $G$ with a reference edge $\eref$ such that the horizontal segment $S^\star \in \PH$ with $\eref \in S^\star$ satisfies $\Nnorth(S^\star) = \emptyset$, in this section we describe an algorithm that achieves the following. If $(\RR,\eref)$ is drawable, then the algorithm adds virtual edges to~$\RR$  so that a good sequence $A=(S_1, S_2, \ldots, S_{k})$ such that $G_k = G$ exists and can be computed efficiently, and then a drawing of $(\RR,\eref)$ can be constructed using the drawing algorithm in the previous section. If $(\RR,\eref)$ is not drawable, then the algorithm returns a strictly monotone cycle in $G$ to certify that  $(\RR,\eref)$ is not drawable. Recall that we require  $\eref$ to be placed on the outermost circular arc in the drawing of $(\RR,\eref)$. A necessary condition for such a drawing to exist is $\Nnorth(S^\star) = \emptyset$.
 
 \paragraph{Preprocessing step 1: the outer face} To ensure that a non-empty good sequence exists, by \ref{item:S1}, it is required that $S  = \overline{C_{\FO}}$ is a horizontal segment in $\PH$. If this requirement is not met, then we will add virtual edges to $\RR$ to satisfy this requirement. Let $S^\star \in \PH$ be the horizontal segment that contains the reference edge $\eref$, and we have $\Nnorth(S^\star) = \emptyset$. We add a virtual horizontal edge $e_f$ that connects the two endpoints of $S^\star$, so $S^\star$ together with $e_f$ becomes the new contour of the outer face and is a horizontal segment in $\PH$. See \cref{fig:preprocessing} for an illustration.
 
 Observe that the addition of a virtual edge, in general, does not change the value of edge label $\ell_C(e)$ of any edge $e$ in any essential cycle $C$ that already exists in the original graph, as long as the addition of the virtual edge does not destroy \ref{item:R1} or \ref{item:R2}. The reason is that the calculation of $\ell_C(e)$ is invariant under the choice of the reference path $P$ in the calculation of $\ell_C(e)$, and there is always a reference path $P$ that already exists in the original graph and does not involve any virtual edge.  

 \paragraph{Preprocessing step 2: smoothing} As our goal is to find a good sequence $A=(S_1, S_2, \ldots, S_{k})$ such that $G_k = G$, it is necessary that $A$ contains all the  horizontal segments in $\PH$ and each vertex $v \in V$ is incident to a horizontal segment in $\PH$. As we assume that the underlying graph is biconnected, the only possibility that a vertex $v \in V$ is not incident to any horizontal segment is that $\deg(v) = 2$ and $v$ is incident to two vertical edges $(u,v)$ and $(v,w)$. We may get rid of any such vertex $v$ by \emph{smoothing} it, that is, we replace $(u,v)$ and $(v,w)$ with a single vertical edge $(u,w)$. See \cref{fig:preprocessing} for an illustration. It is straightforward to see that smoothing does not affect the drawability of $(\RR,\eref)$, and a drawing of the graph after smoothing can be easily transformed into a drawing of the graph before smoothing.
 From now on, we assume that each vertex $v \in V$ is incident to a horizontal segment in $\PH$, and so all we need to do is to find a good sequence that covers all the horizontal segments.

\begin{figure}[t!]
\centering
\includegraphics[width=0.6\textwidth]{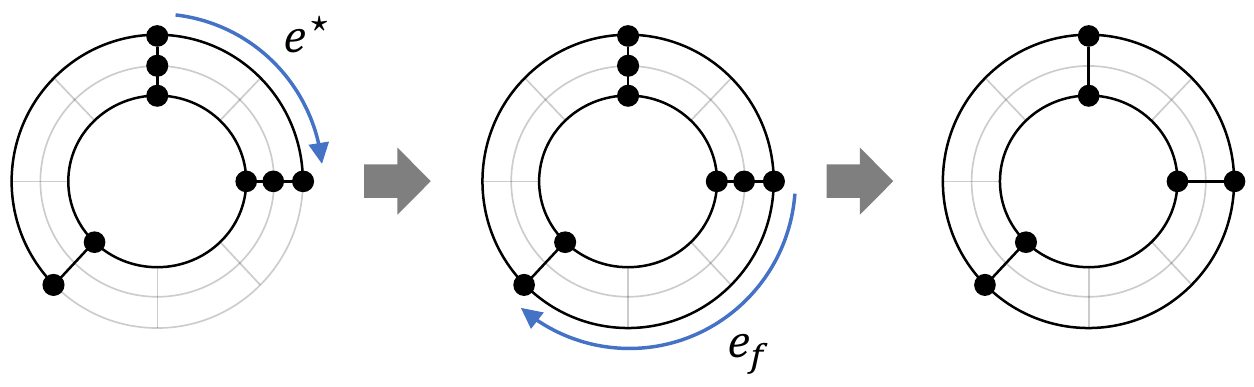}
\caption{The preprocessing steps.}\label{fig:preprocessing}
\end{figure}

 \paragraph{Eligibility for adding virtual edges} Let $S \in \PH$ such that $\Nnorth(S) = \emptyset$ and  $S  \neq \overline{C_{\FO}}$. Such a horizontal segment $S$  can never be added to a good sequence as \ref{item:S2} requires $\Nnorth(S)$ to be non-empty. To deal with this issue, we consider the following eligibility criterion for adding a virtual vertical edge incident to such a horizontal segment $S$.
 
 Let $A=(S_1, S_2, \ldots, S_{k})$ be a good sequence.
 Let $S \notin A$ be a horizontal segment such that $\Nnorth(S) = \emptyset$.
 Let $F$ be the face such that $\overline{S}$ is a subpath of $C_F$. 
 We say that $S$ is \emph{eligible} for adding a virtual edge if there exists an edge $e' \in C_F$ with  $e' \in S_i$ for some $1 \leq i \leq k$ such that either $\rotation(e' \circ \cdots \circ \overline{S}) = 2$ or  $\rotation(\overline{S} \circ \cdots \circ e') = 2$ along the cycle $C_F$. Intuitively, the condition $\Nnorth(S) = \emptyset$ ensures that immediately after adding the virtual edge, we may append $S$ to the end of the sequence $A$.
 
 \begin{figure}[t!]
\centering
\includegraphics[width=0.8\textwidth]{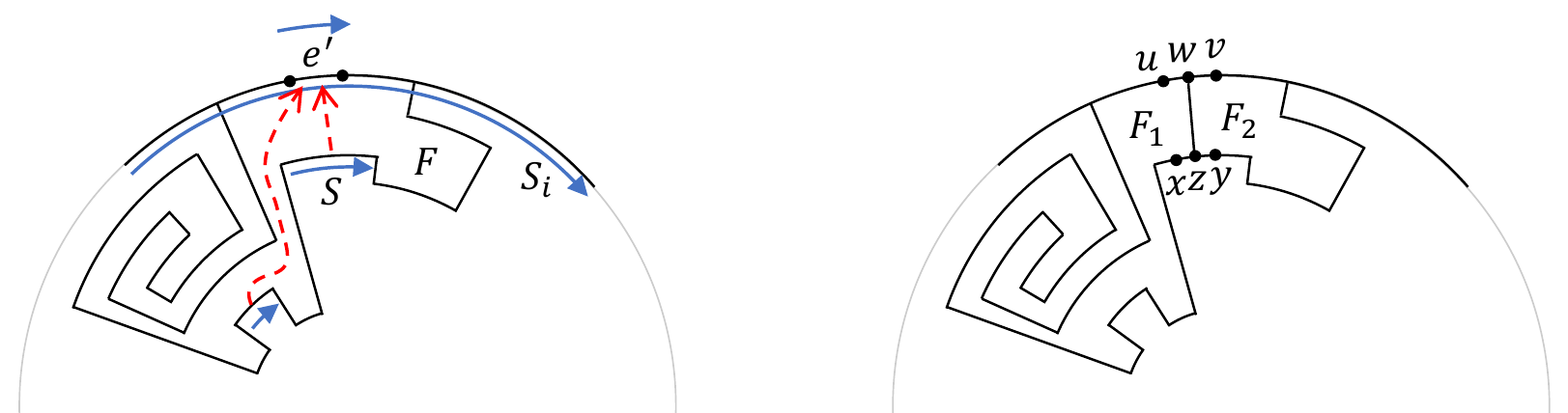}
\caption{Adding a virtual vertical edge in a regular face.}\label{fig:elig-reg}
\end{figure}
 
 For the case that $F$ is a regular face, $\rotation(C_F) = 4$, so $\rotation(e' \circ \cdots \circ \overline{S}) = 2$ if and only if $\rotation(\overline{S} \circ \cdots \circ e') = 2$. We also allow $F$ to be the central face, in which case at most one of  $\rotation(e' \circ \cdots \circ \overline{S}) = 2$ and  $\rotation(\overline{S} \circ \cdots \circ e') = 2$ can be true.
 
  We argue that if $S$ is eligible for adding a virtual edge with respect to the current good sequence $A$, then we may add a virtual vertical edge $e_f=(z,w)$ connecting a middle point $z$ of a horizontal edge $(x,y)$ in $S$ and a middle point $w$ of the horizontal edge $e' = (u,v)$.
See \cref{fig:elig-reg} for an illustration. In the figure, $F$ is a regular face, and there are two horizontal segments along the contour of $F$ that are eligible for adding a virtual edge due to $e' \in S_i$.

We argue that the addition of $e_f=(z,w)$ does not  destroy \ref{item:R1} and \ref{item:R2}. The verification of \ref{item:R1} is straightforward. We verify \ref{item:R2} for the case  $\rotation(\overline{S} \circ \cdots \circ e') = 2$, as the other case $\rotation(e' \circ \cdots \circ \overline{S}) = 2$ is similar. The addition of $e_f$ decomposes $F$ into two new faces $F_1$ and~$F_2$. 
Let $F_1$ be the one whose facial cycle contains $(u,w,z,x)$ as a subpath, and let $F_2$ be the one whose facial cycle contains $(y,z,w,v)$ as a subpath.

We first consider $F_1$.
  The rotation of the facial cycle of $F_1$ equals $\rotation(\overline{S} \circ \cdots \circ e') = 2$ plus the rotation of the path $(u,w,z,x)$, which is also $2$, as $(u,w,z,x)$ consists of two right turns. 
Therefore, the rotation of this facial cycle is $4$, meaning that $F_1$ is a regular face.
Now consider the other face $F_2$.  The rotation of the facial cycle of $F_2$ is identical to $\rotation(C_F)$ before the addition of $e_f$. The reason is that the rotation of the subpath from $y$ to $v$ is both $2$ in $C_F$ and in~$C_{F_2}$, as we assume that $\rotation(\overline{S} \circ \cdots \circ e') = 2$. If $F$ is a regular face, then the sum of rotations is $4$ for both $F$ and $F_2$, so $F_2$ is also a regular face. If $F$ is a central face, then the sum of rotations is $0$ for both $F$ and $F_2$, so $F_2$ is also a central face. 

Consider \cref{fig:elig-central} as an example: There are four horizontal segments in the contour of the central face $F$ that are eligible for adding a virtual vertical edge due to $e' \in S_i$.
The two horizontal segments highlighted in the left part of the figure are eligible due to  $\rotation(e' \circ \cdots \circ \overline{S}) = 2$ along the cycle $C_F$.  The two horizontal segments highlighted in the right part of the figure are eligible due to  $\rotation(\overline{S} \circ \cdots \circ e') = 2$ along the cycle $C_F$.

\paragraph{A greedy algorithm} Assuming that $C_{\FO} \in \PH$ and each $v \in V$ is incident to a horizontal segment, our algorithm for constructing a good sequence is as follows. We start with the trivial good sequence $A=(S_1)$, where  $S  = \overline{C_{\FO}}$, 
and then we repeatedly do the following two operations until no further such operations can be done.
\begin{itemize}
    \item Find a horizontal segment $S \in \PH$ such that appending $S$ to the end of the current sequence $A$ results in a good sequence, and then extend $A$ by adding $S$ to the end of $A$.
    \item Find a horizontal segment $S \in \PH$ that is eligible for adding a virtual edge with respect to the current good sequence $A$, and then add a virtual vertical edge incident to $S$ as discussed above.
\end{itemize}

There are two possible outcomes of the algorithm. If we obtain a good sequence that covers all horizontal segments $\PH$, then we may use \cref{lem:good-drawing} to compute a drawing of $(\RR, \eref)$. Otherwise, the algorithm stops with a good sequence that does not cover all horizontal segments $\PH$, and no more progress can be made, in which case in the next section we will show that a strictly monotone cycle in the original graph $G$ can be found. 

A straightforward implementation of the greedy algorithm, which checks all horizontal segments in each step, takes $O(n^2)$ time. In the following lemma, we present a more efficient implementation that requires only $O(n \log n)$ time.

\begin{figure}[t!]
\centering
\includegraphics[width=0.9\textwidth]{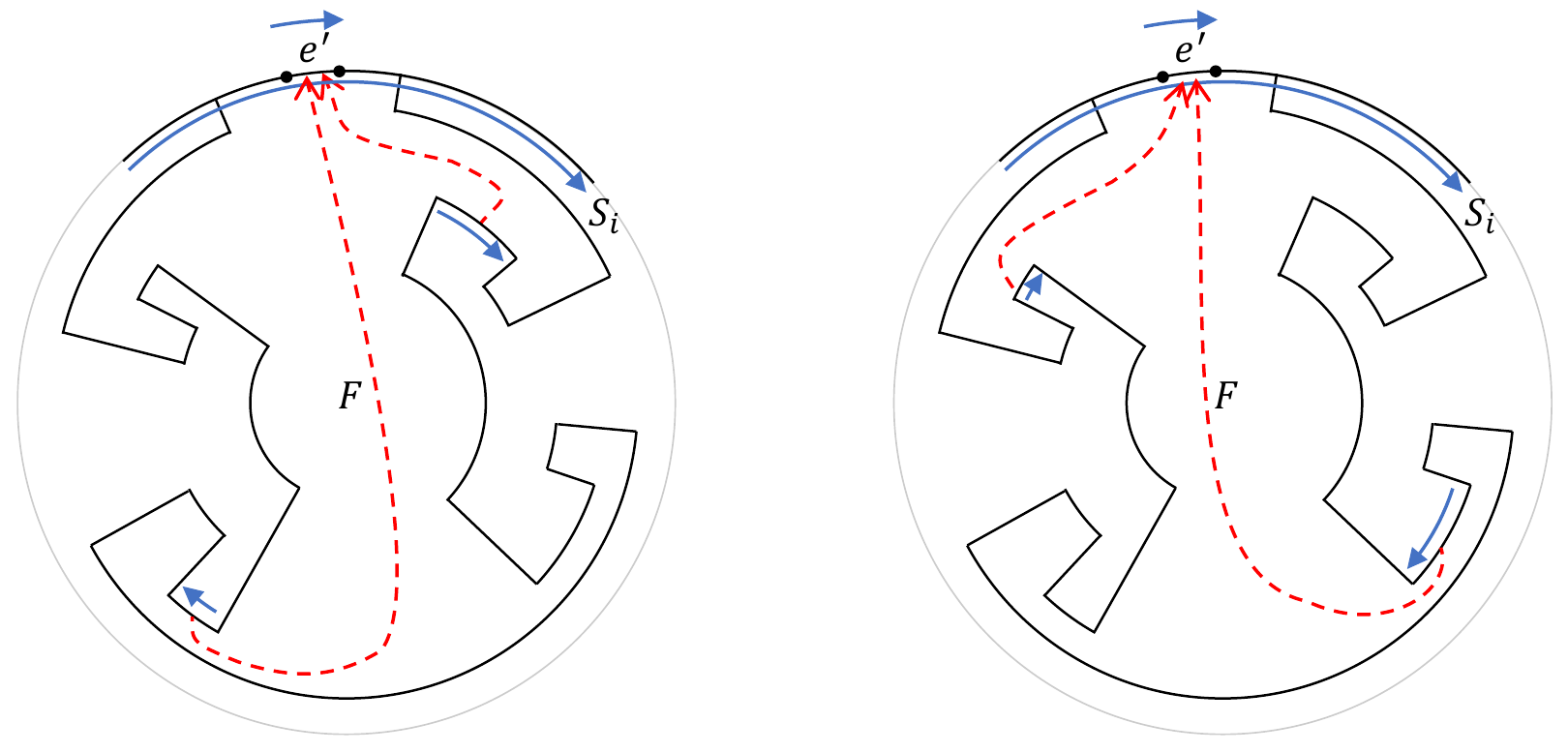}
\caption{Eligible horizontal segments in the contour of the central face.}\label{fig:elig-central}
\end{figure}

\begin{lemma}\label{lem-time-sequence}
The greedy algorithm can be implemented to run in $O(n \log n)$ time.
\end{lemma}
\begin{proof}
 Let $A=(S_1, S_2, \ldots, S_k)$ denote the current good sequence, which is initialized to an empty set $\emptyset$.
 During the algorithm, we  maintain the circular ordering  $\Nsouth(S_1, S_2, \ldots, S_{k})$ as a circular doubly linked list. Whenever a path $S \in \PH$ is inserted to $A$, this circular doubly linked list is updated by replacing the contiguous subsequence $\Nnorth(S)$ of  $\Nsouth(S_1, S_2, \ldots, S_{k})$ with $\Nsouth(S)$. Whenever a cycle $S \in \PH$ is inserted to $A$, we have $\Nsouth(S_1, S_2, \ldots, S_{k}) = \Nnorth(S)$, so the circular doubly linked list is updated to the circular ordering 
 of $\Nsouth(S)$.
 
 \paragraph{Horizontal segments}
 Throughout the algorithm, we maintain a set $W$ containing all horizontal segments $S \in \PH$ such that adding $S$ to $A$ results in a good sequence. If $S$ is a path, then $S$ can  be added to $A$ once $\Nnorth(S)$ is a  contiguous subsequence of $\Nsouth(S_1, S_2, \ldots, S_{k})$. If $S$ is a cycle, then $S$ can  be added to $A$ once $\Nnorth(S) = \Nsouth(S_1, S_2, \ldots, S_{k})$. Our goal is to design a suitable data structure and an algorithm to efficiently detect a horizontal segment $S$ that can be added to $A$, if such a segment exists.

 The set $W$ is initialized as $W=\{\overline{C_{\FO}}\}$ and is maintained as follows. 
 First, consider any $S \in \PH$ with $|\Nnorth(S)| = 1$.  Let $\Nnorth(S) = \{e\}$.  If $S$ is a path, then we add $S$ to $W$ once $e$ appears in $\Nsouth(S_1, S_2, \ldots, S_{k})$. If $S$ is a cycle, then we add $S$ to $W$ once $\Nsouth(S_1, S_2, \ldots, S_{k})=\{e\}$. Note that the case $S$ is a cycle with $|\Nnorth(S)| = 1$ is not possible when $G$ is biconnected.
 
 Next, consider any $S \in \PH$ with $|\Nnorth(S)| \geq 2$. For any two vertical edges $e$ and $e'$ such that~$e'$ immediately follows~$e$ in the ordering $\Nnorth(S)$, we maintain an indicator $X_{e, e'} \in \{0,1\}$ such that $X_{e, e'} = 1$ if~$e'$ also immediately follows~$e$ in $\Nsouth(S_1, S_2, \ldots, S_{k})$. Initially, all $X_{e, e'}$ are set to $0$. For each update to $\Nsouth(S_1, S_2, \ldots, S_{k})$, we check and update $X_{e, e'}$ for all edges~$e$ and~$e'$ that could be affected. For example, if an edge $\tilde{e}$ is removed from $\Nsouth(S_1, S_2, \ldots, S_{k})$, we check if the two edges  immediately preceding and following $\tilde{e}$ in $\Nsouth(S_1, S_2, \ldots, S_{k})$ belong to $\Nnorth(S)$ for the same horizontal segment $S$. If the answer is yes, then we set the corresponding indicator $X_{e, e'}=1$ for $S$. 
 By \ref{item:S2}, $S$ can be inserted to $A$ if and only if the value of each of its indicators is $1$.
 Therefore, we can decide whether $S$ should join $W$ by checking the sum $X_S$ of all its indicators $X_{e, e'}$. This sum $X_S$ is updated and checked whenever we update the value of an indicator $X_{e, e'}$ for $S$.
 
 The data structure and algorithm described above cost $O(n)$ time. Each insertion of a horizontal segment $S$ to the good sequence $A=(S_1, S_2, \ldots, S_k)$ incurs a number of insertions and deletions to the circular doubly linked list $\Nsouth(S_1, S_2, \ldots, S_{k})$, and each of these updates gives rise to $O(1)$  operations. The total time complexity is $O(n)$, as the number of updates is linear in $|\EV|=O(n)$.
 
 \paragraph{Virtual edges}
 Next, we consider the task of adding virtual vertical edges. Whenever a horizontal segment $S' \in \PH$ is inserted to the good sequence $A=(S_1, S_2, \ldots, S_k)$, we check each edge $e' \in S'$ to see if $e'$ causes some horizontal segment $S \in \PH$ to become eligible for adding virtual edges with respect to $(S_1, S_2, \ldots, S_k, S')$. For all such $S$, we add a virtual edge $e_f$ incident to $S$ and then add $S$ to $W$, as the addition of $e_f$ causes the insertion of $S$ to result in a good sequence.
 
 In the subsequent discussion, we fix an edge $e' \in S'$ and consider the task of finding a horizontal segment $S \in \PH$ that is eligible for adding a virtual edge with respect to $(S_1, S_2, \ldots, S_k, S')$ due to $e'$, if such a horizontal segment $S$ exists. Let $F$ be the face where $e' \in C_F$. We only need to check the set of all $S \in \PH$ such that $\overline{S}$ is a subpath of $C_F$ and $\Nnorth(S) = \emptyset$. After finding such an $S$ and adding a virtual edge $e_f$ incident to $S$, the face $F$ is divided into two faces $F_1$ and $F_2$, and the edge $e'$ is also divided into two edges $e_1'$ and $e_2'$, where $e_1' \in C_{F_1}$ and $e_2' \in C_{F_2}$. We will recursively apply the algorithm for both $e_1'$ and $e_2'$. By applying the algorithm for all $e' \in S'$, we can ensure that, at the end of this recursive process, no more virtual edges can be added. 
 
 A straightforward algorithm for the above task involves examining all $S \in \PH$ such that $\overline{S}$ is a subpath of $C_F$ and $\Nnorth(S) = \emptyset$. This approach requires 
$O(n)$ time in the worst case to identify a single horizontal segment $S \in \PH$ eligible for adding a virtual edge, which is costly. In the following, we present a more efficient algorithm and data structure.  
 
  \paragraph{Regular faces}
 We first focus on the case where the face $F$ with $e' \in C_F$ is a regular face. As discussed earlier, our task is to find a horizontal segment $S \in \PH$ with $\Nnorth(S) = \emptyset$ such that $\overline{S}$ is a subpath of~$C_F$ and either $\rotation(e' \circ \cdots \circ \overline{S}) = 2$ or  $\rotation(\overline{S} \circ \cdots \circ e') = 2$ along the cycle~$C_F$, if such an~$S$ exists. As $F$  is a regular face, $\rotation(C_F)=4$, so $\rotation(e' \circ \cdots \circ \overline{S}) = 2$ and  $\rotation(\overline{S} \circ \cdots \circ e') = 2$ are equivalent. 

We maintain the following data structure for face $F$, which has two components:
\begin{itemize}
    \item The first part of the data structure is a circular doubly linked list of the edges $C_F=(e_1, e_2, \ldots, e_s)$, where $s$ is the number of edges of $C_F$. To facilitate the calculation of rotation of a subpath of $C_F$, we calculate and store the value of $r_i = \rotation(e_1 \circ \cdots \circ e_i)$ for each $1 \leq i \leq s$.
    \item In the part of the data structure, we organize the set of all $S \in \PH$ such that $\overline{S}$ is a subpath of $C_F$ and $\Nnorth(S) = \emptyset$ into an array of buckets $B_{\min_i r_i}, \ldots, B_{\max_i r_i}$ such that $S$ is added to bucket $B_j$ if the rotation from $e_1$ to the first edge $e_S$ of $\overline{S}$ is $j$. The horizontal segments in each bucket are organized as a doubly linked list, sorted according to the indices of $e_S$ in $C_F=(e_1, e_2, \ldots, e_s)$. 
\end{itemize}

 
\paragraph{Queries} We show that, given the data structure, in $O(1)$ time, we can find a horizontal segment $S \in \PH$ that is eligible for adding a virtual edge, if such a horizontal segment exists.

 The rotation  from $e_i$ to $e_j$ is $r_j - r_i$ if $1 \leq i \leq j \leq s$ and is $r_j - r_i +4$ if $1 \leq j < i \leq s$. Let $i$ be the index such that $e' = e_i$. Our goal is to search for a horizontal segment $S$ such that $\rotation(e' \circ \cdots \circ \overline{S}) = 2$, so we may limit our search space to $e_j$ such that 
 $r_j = r_i + 2$ or $r_j  = r_i -2$, meaning that we only need to check the two buckets $B_{r_i-2}$ and $B_{r_i+2}$.
 \begin{itemize}
     \item The bucket $B_{r_i-2}$ is considered because for the case $1 \leq j < i \leq s$, the rotation from $e_i$ to $e_j$ is $2$ if and only if $r_j - r_i +4 = 2$, which is equivalent to $r_j = r_i - 2$. Therefore, the search space for this bucket will be any indices within the range $[1,i-1]$, so all we need to do is to check the \emph{first} element of the linked list $B_{r_i-2}$ to see if its index lies in the range $[1,i-1]$.
     \item Similarly, for the bucket $B_{r_i+2}$, we just need to check the \emph{last} element of the linked list to see if the index lies in the range $[i+1, s]$.
 \end{itemize}
 The search can be done in $O(1)$ time.
 
 
 

 \paragraph{Updates}
 In case a desired horizontal segment $S$ is found, as discussed earlier, the face $F$ will be split into two faces $F_1$ and $F_2$. If we rebuild the above data structure for both faces from scratch, then the reconstruction costs $O(n)$ time in the worst case, which we cannot afford. A key observation here is that both $F_1$ and $F_2$ can be seen as the result of replacing a subpath of $F$ of rotation $2$ with a new path of rotation $2$, meaning that we can still reuse the same rotation values $\{r_i\}$ in $F_1$ and $F_2$. For any two edges $e_{i'}$ and $e_{j'}$ that still belong to the same face after splitting, the rotation from $e_{i'}$ to $e_{j'}$ in the new face can still be computed using the same formula from $r_{i'}$ and $r_{j'}$ defined with respect to the old face $F$. As there is no need to recompute the rotation values, the two circular doubly linked lists $C_{F_1}$ and $C_{F_2}$ can be computed from the given circular doubly linked list $C_F=(e_1, e_2, \ldots, e_s)$ in $O(1)$ time.

Next, we consider the second part of the data structure. We show that the two arrays of buckets for $F_1$ and $F_2$ can be computed in $O(\min\{|C_{F_1}|, |C_{F_2}|\})$ time, so we can charge the cost to the edges in the \emph{smaller} face, where each edge is charged a cost of $O(1)$. Since each edge is charged $O(\log n)$ times in total, the total time spent on updating the data structures can be upper bounded by $O(n \log n)$, as desired.

In $O(\min\{|C_{F_1}|, |C_{F_2}|\})$ time, we can decide whether $|C_{F_1}| \geq |C_{F_2}|$. Without loss of generality, we assume $|C_{F_1}| \geq |C_{F_2}|$. In $O(|C_{F_2}|) = O(\min\{|C_{F_1}|, |C_{F_2}|\})$ time, we can build the array of buckets for $F_2$ from scratch. The array of buckets for $F_1$ can be obtained from the given array of buckets for $F$ by removing the horizontal segments in $F_2$ from the linked lists one by one in $O(|C_{F_2}|) = O(\min\{|C_{F_1}|, |C_{F_2}|\})$ time.

   \paragraph{The central face} Now consider the remaining case where the face $F$ with $e' \in C_F$ is the central face. 
   In this case, the two conditions $\rotation(e' \circ \cdots \circ \overline{S}) = 2$ and  $\rotation(\overline{S} \circ \cdots \circ e') = 2$ are not equivalent, as $\rotation(C_F) = 0$. However, we can still search for an eligible horizontal segment based on the same approach by considering both two conditions. 
   
   Specifically, here   we want to find $S$ such that either  $\rotation(e' \circ \cdots \circ \overline{S}) = 2$ or $\rotation(\overline{S} \circ \cdots \circ e') = 2$, from the  set of all $S \in \PH$ with $\Nnorth(S) = \emptyset$ and $\overline{S}$ is a subpath of $C_F$. Again, we write $C_F=(e_1, e_2, \ldots, e_s)$, let $e_i  = e'$, and let $e_j$ be an edge in $\overline{S}$. Then the rotation  from $e_i$ to $e_j$ is $r_j - r_i$ for all $i$ and $j$, so  we only need to consider $e_j$ such that 
 $r_j = r_i + 2$ or $r_j  = r_i -2$. Any horizontal segment in the two buckets $B_{r_i-2}$ and $B_{r_i+2}$ are eligible, so the search can still be done in $O(1)$ time. 
 

 When a virtual edge is inserted, the face $F$ will be divided into two faces $F_1$ and $F_2$. Unlike the case of regular faces, we cannot reuse the $r$-values for both  $F_1$ and $F_2$. Here, only one $F^\ast \in \{F_1, F_2\}$ of the two new faces can be seen as the result of replacing a subpath of $F$ of rotation $2$ with a new path of rotation $2$, so the rotation of the facial cycle of $F^\ast$ is still $0$,  $F^\ast$ will be the new central face, and the old $r$-values computed for $F$ can still be used for $F^\ast$.
 For the other new face $F' \in \{F_1, F_2\} \setminus \{F^\ast\}$, it is the result of replacing a subpath of $F$ of rotation $-2$ with a new path of rotation $2$, so the rotation of the facial cycle of $F'$ will be $4$,  $F'$ will be a regular face, and the old $r$-values computed for $F$ cannot be used for $F'$. 
 
 To deal with the above issue, we simply construct the data structure of $F'$ from scratch, where the time spent is linear in the number of edges in the facial cycle of $F'$. The total cost for the reconstruction throughout the algorithm is upper bounded by $O(s)=O(n)$, where $s$ is the number of edges in $C_{\FC}$ in the original graph $G$.
\end{proof}

\section{Extracting a strictly monotone cycle}\label{sect:cycle}

In this section, we consider the scenario where the greedy algorithm in the previous section stops with a good sequence $A=(S_1, S_2, \ldots, S_k)$ that does not cover the set of all horizontal segments $\PH$, and our goal is to show that in this case a strictly monotone cycle of the original graph $G$ can be computed in $O(n)$ time.

We introduce the terminology that will be used throughout the section. Given a good sequence $A=(S_1, S_2, \ldots, S_k)$ of size $k$, let $(e_1, e_2, \ldots, e_s)$ be the circular ordering of $\Nsouth(A)$, and let $e_j=(x_j, y_j)$, for each $1 \leq j \leq s$, where $s$ is the size of $\Nsouth(A)$. We write $\tilde{G}$ to denote the graph resulting from running the greedy algorithm. That is, $\tilde{G}$ is the original graph $G$ plus all the virtual edges added during the greedy algorithm. Both $G_k$ and $G_k^+$ are seen as subgraphs of~$\tilde{G}$. For each $1 \leq i \leq s$, we write $F_{i, i+1}$ to denote the unique face $F$ of $\tilde{G}$ such that $C_F$ contains both  ${e_{i}}$ and $\overline{e_{i+1}}$. Note that  $v_{s+1}=v_1$ because $(e_1, e_2, \ldots, e_s)$ is a circular ordering. Since we assume that $G$ is biconnected, we cannot have $s= |\Nsouth(A)| = 1$.
We assume that $A=(S_1, S_2, \ldots, S_k)$ and $\tilde{G}$ are the end results of our greedy algorithm in that $A$ cannot be further extended to a longer good sequence and no more virtual edges can be added to $\tilde{G}$.

\begin{figure}[t!]
\centering
\includegraphics[width=0.85\textwidth]{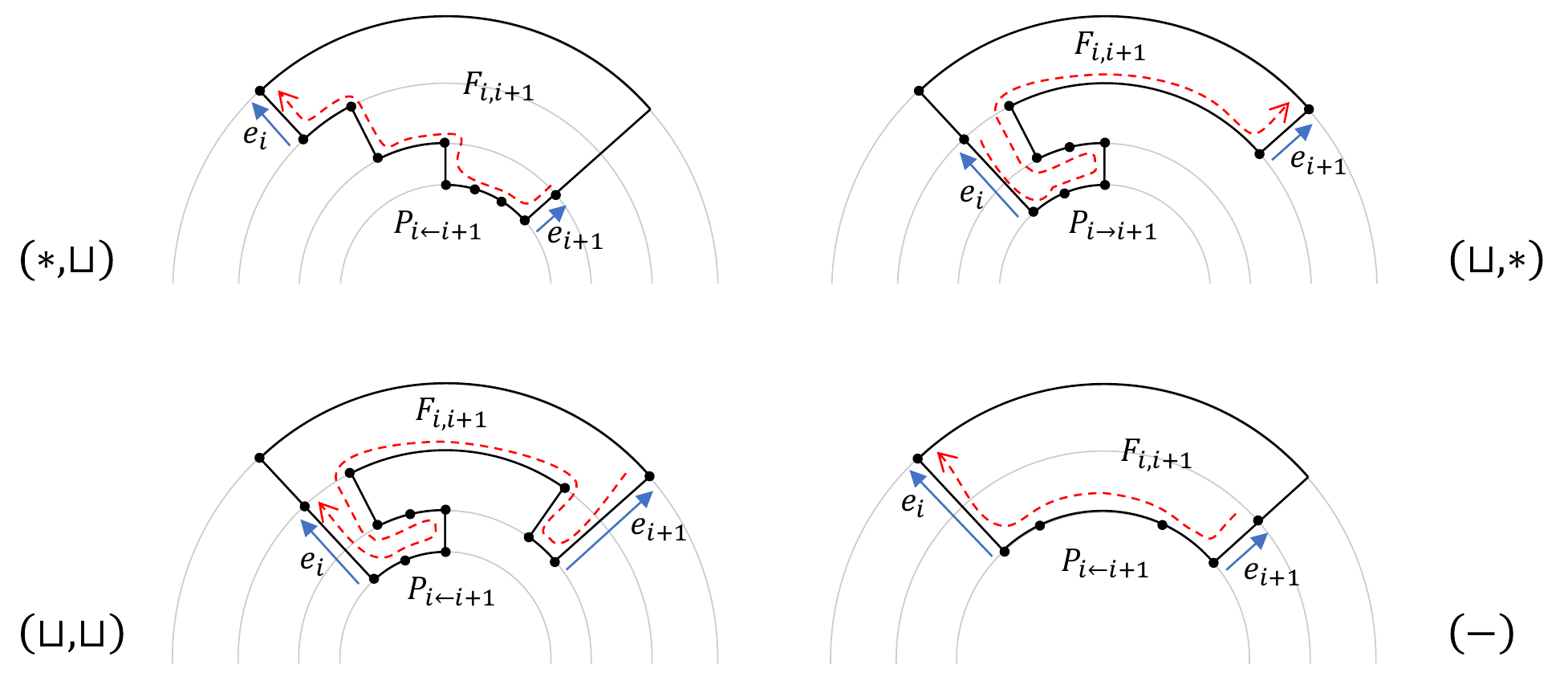}
\caption{Face types $(\ast,\sqcup)$, $(\sqcup, \ast)$, $(\sqcup, \sqcup)$, and $(-)$.}\label{fig:types}
\end{figure}

\paragraph{Face types}
Consider the face $F_{i,i+1}$, for some $1 \leq i \leq s$. We define $P_{i \leftarrow i+1}$ to be the subpath of $C_{F_{i,i+1}}$ starting at $\overline{e_{i+1}}$ and ending at ${e_{i}}$.  We write $P_{i \rightarrow i+1} = \overline{P_{i \leftarrow i+1}}$.
We write  $Z_{i\leftarrow i+1}=(z_1, z_2, \ldots)$ to denote the string of numbers such that $z_l$ is the rotation of the subpath of  $P_{i \leftarrow i+1}$ consisting of the first $l$ edges.
Similarly, we let $Z_{i\rightarrow i+1}=(z_1, z_2, \ldots)$ be the string of numbers such that $z_l$ is the rotation of the subpath of  $P_{i \rightarrow i+1}$ consisting of the first $l$ edges. Recall that we define $\rotation(P) = 0$ for any single-edge path $P$.
We define the types $(\ast,\sqcup)$, $(\sqcup,\ast)$, $(\sqcup,\sqcup)$, and $(-)$, as follows.
\begin{itemize}
\item  $F_{i,i+1}$ is of type $(\ast,\sqcup)$ if $0 \circ 1^c \circ 2$, for some $c \geq 1$, is a strict prefix of $Z_{i \leftarrow i+1}$. 
\item  $F_{i,i+1}$ is of type $(\sqcup,\ast)$ if $0 \circ (-1)^c \circ (-2)$, for some $c \geq 1$, is a strict prefix of $Z_{i \rightarrow i+1}$.
\item  $F_{i,i+1}$ is of type $(\sqcup, \sqcup)$ if $F_{i,i+1}$ is both of type $(\sqcup, \ast)$ and of type $(\ast, \sqcup)$.
\item $F_{i,i+1}$ is of type $(-)$ if $Z_{i \leftarrow i+1}= 0 \circ 1^c \circ 2$ for some $c \geq 1$.
\end{itemize}

We emphasize that, due to the \emph{strict} prefix requirement, any face of type $(-)$ is \emph{not} of types $(\ast,\sqcup)$, $(\sqcup,\ast)$, and $(\sqcup, \sqcup)$.
Faces of type $(-)$ can alternatively be defined as follows: $F_{i,i+1}$ is of type $(-)$ if the subpath $(x_{i+1},\ldots, x_{i})$ of the facial cycle of $F_{i,i+1}$ is a horizontal straight line in the west direction.
Considering $P_{i \rightarrow i+1} = \overline{P_{i \leftarrow i+1}}$, equivalently, $F_{i,i+1}$ is of type $(-)$ if $Z_{i \rightarrow i+1}= 0 \circ (-1)^c \circ (-2)$ for some $c \geq 1$.  


Intuitively, the face $F_{i,i+1}$ is of type $(\sqcup, \ast)$ if $P_{i \rightarrow i+1}$ makes two $90^\circ$ left turns before making any right turns, and the first $90^\circ$ left turn is made at $x_{i}$. These two $90^\circ$ left turns form a $\sqcup$-shape.
Similarly, the face $F_{i,i+1}$ is of type $(\ast, \sqcup)$ if ${P_{i \leftarrow i+1}}$ makes two $90^\circ$ right turns before making any left turns, and the first $90^\circ$ right turn is made at $x_{i+1}$.  These two $90^\circ$ right turns form a $\sqcup$-shape.

See \cref{fig:types} for illustrations of the four face types:
\begin{itemize}
    \item Upper-left: $Z_{i\leftarrow i+1}=(0,1,1,1,2,1,2,1,2)$, so $F_{i,i+1}$ is of type $(\ast,\sqcup)$.  
    \item Upper-right: $Z_{i\rightarrow i+1}=(0,-1,-1,-2,-3,-3,-2,-1,-2)$, so $F_{i,i+1}$ is of type $(\sqcup,\ast)$. 
    \item Lower-left: $Z_{i\leftarrow i+1}=(0,1,2,1,0,-1,-1,0,1,1,2)$ and $Z_{i\rightarrow i+1}=(0$, $-1$, $-1$, $-2$, $-3$, $-3$, $-2$, $-1$, $0$, $-1$, $-2)$, so $F_{i,i+1}$ is of type $(\sqcup,\sqcup)$.  
    \item Lower-right: $Z_{i\leftarrow i+1}=(0,1,1,1,2)$, so $F_{i,i+1}$ is of type $(-)$. 
\end{itemize}

We aim to extract a strictly monotone cycle by analyzing the face types. For example, if we can show that all faces $F_{i,i+1}$ are of types $(-)$ and $(\ast,\sqcup)$, with at least one face of type $(\ast,\sqcup)$, then intuitively an increasing cycle can be found.

\paragraph{Structural properties}
We analyze the structural properties of the edges $(e_1, e_2, \ldots, e_s)$ and their incident faces $F_{i,i+1}$. The following lemma proves the intuitive fact that the rotation from the reference edge $\eref$ to any $\overline{e_i}$ must be $90^\circ$ via any crossing-free path $P$ in $G_k^+$.
Note that such a path $P$ must exist.

\begin{lemma}\label{lem-e-direction}
Let $P$ be any crossing-free path in $G_k^+$ starting at the reference edge $\eref$ and ending at $\overline{e_i}$, for some $1 \leq i \leq s$. Then $\rotation(P) = 1$. 
\end{lemma}
\begin{proof}
See the left drawing of \cref{fig:lem-fig-1} for an illustration of the proof.
 Consider the ortho-radial representation $\RR'$ resulting from connecting the south endpoints $x_1, x_2, \ldots, x_s$ of the edges $e_1, e_2, \ldots, e_s$ in $G_k^+$ into a cycle $C = (x_1, x_2, \ldots, x_s)$ in such a way that the rotation from $\overline{e_i}$ to the new edge $(x_i, x_{i+1})$ is a $90^\circ$ left turn,   the rotation from $\overline{e_i}$ to the new edge $(x_i, x_{i-1})$ is a $90^\circ$ right turn, and any subpath of $C$ is a straight line.
 Observe that $C$ is a horizontal segment and is the facial cycle of the central face of $\RR'$. By \ref{item:D1} and \ref{item:D2}, it is straightforward to convert a good drawing of $G_k$ into an ortho-radial drawing of $\RR'$, so $\RR'$ is drawable. Since $C$ is a horizontal segment, the edge label  $\ell_C(e)$ of all edges $e \in C$  must be the same. Since $C$ cannot be a strictly monotone cycle, the only possibility is that $C$ is a monotone cycle, meaning that $\ell_C(e) = 0$ for all edges $e \in C$.  Now consider the path $P$ in the lemma statement and the edge $e=(x_i, x_{i+1})$ in~$C$.
 Since $\ell_C(e) = 0$, we  have $\rotation(P \circ e) = 0$. Since $P \circ e$ makes a $90^\circ$ left turn at $x_i$, we have $\rotation(P) = 1$. 
\end{proof}

In the subsequent discussion, we let $\Pout_{i\rightarrow j}$ denote the subpath of $C_k^+$ starting at $e_i$ and ending at $\overline{e_j}$, for any  $1 \leq i \leq s$ and $1 \leq j \leq s$.  For the special case that $j=i+1$, $\Pout_{i\rightarrow j}$ is a subpath of the facial cycle of $F_{i,i+1}$.  The following lemma proves the intuitive fact that the rotation of 
$\Pout_{i\rightarrow j}$ is $2$.

\begin{lemma}\label{lem-rotation-aux}
For any $1 \leq i \leq s$ and $1 \leq j \leq s$ with $i \neq j$, we have $\rotation(\Pout_{i\rightarrow j}) = 2$.
\end{lemma}
\begin{proof}
See the middle drawing of \cref{fig:lem-fig-1} for an illustration of the proof.
 Consider the ortho-radial representation $\RR'$ resulting from connecting the south endpoints $x_i$ and $x_j$ of the edges $e_i$ and  $e_j$ in $G_k^+$ into a horizontal edge $e'=(x_j, x_i)$ in such a way that the path $\overline{e_j} \circ e' \circ e_i$ makes two $90^\circ$ right turns. Similar to the proof of  \cref{lem-e-direction}, $\RR'$ is drawable by extending a good drawing of $G_k$.
 The path $\Pout_{i\rightarrow j}$ together with the path $\overline{e_j} \circ e' \circ e_i$ forms a facial cycle  of a regular face. By \ref{item:R2}, we must have $\rotation(\Pout_{i\rightarrow j}) = 2$, as the sum of rotations for a regular face has to be $4$. 
\end{proof}

\begin{figure}[t!]
\centering
\includegraphics[width=\textwidth]{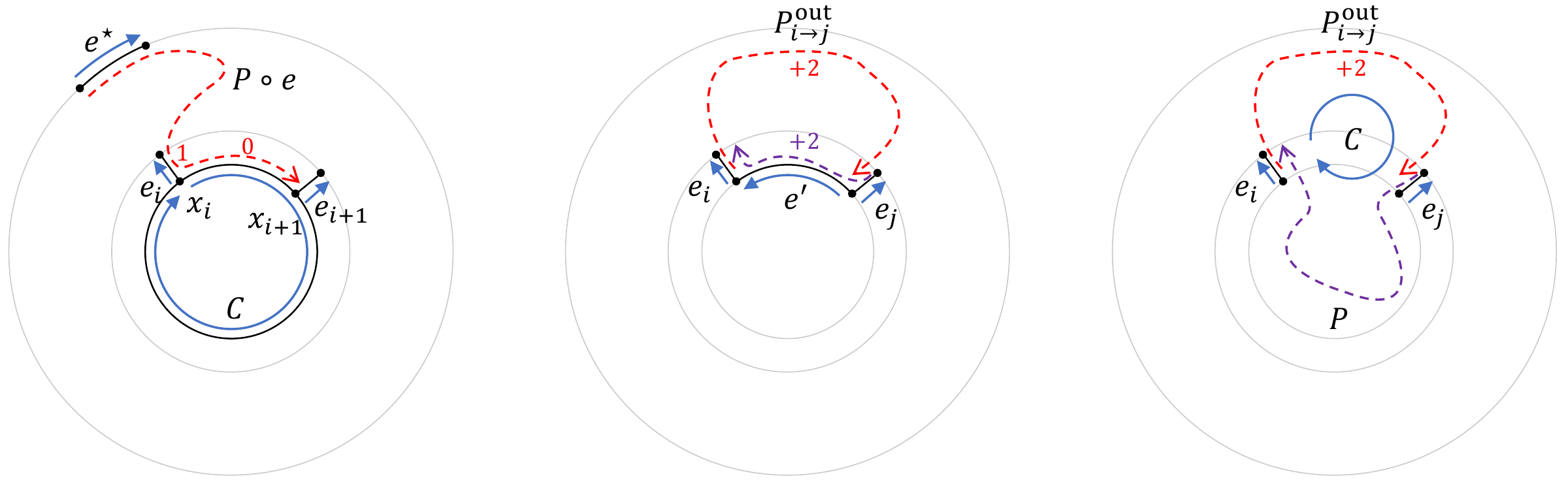}
\caption{Illustration for the proof of \cref{lem-e-direction,lem-rotation-aux,lem-P-rotation}.}\label{fig:lem-fig-1}
\end{figure}

The following lemma considers the rotation of a simple path traversing from $\overline{e_j}$ to $e_i$ that does not use any edges in $G_k$. The lemma shows that the rotation is either $-2$ or $2$, depending on whether the path, together with $\Pout_{i\rightarrow j}$, encloses the central face of $\tilde{G}$.

\begin{lemma}\label{lem-P-rotation}
Let $1 \leq i \leq s$ and $1 \leq j \leq s$ with $i \neq j$.
Let $P$ be any simple path in $\tilde{G}$ starting at $\overline{e_j}$ and ending at $e_i$ satisfying the following conditions.
\begin{itemize}
    \item $P$ lies in the interior of $C_k^+$.
    \item $P$ does not contain any vertex in $\Pout_{i\rightarrow j} \setminus \{x_i, y_i, x_j, y_j\}$.
\end{itemize}
Let $C$ be the cycle resulting from combining $P$ with $\Pout_{i\rightarrow j}$.
If the central face lies in the interior of $C$, then $\rotation(P) = -2$. Otherwise, $\rotation(P) = 2$.
\end{lemma}
\begin{proof}
See the right drawing of \cref{fig:lem-fig-1} for an illustration of the proof.
Consider the subgraph $H$ of $\tilde{G}$ induced by $G_k^+$ and all edges in $P$. Observe that $C$ is a facial cycle of $H$. If the central face of $\tilde{G}$ lies in the interior of $C$, then $C$ is the facial cycle of the central face of $H$, so $\rotation(C) = 0$ by \ref{item:R2}. By \cref{lem-rotation-aux}, $\rotation(\Pout_{i\rightarrow j}) = 2$, so we must have $\rotation(P) = -2$. If the central face of $\tilde{G}$ lies in the exterior of $C$, then $C$ is the facial cycle of a regular face of $H$, so $\rotation(C) = 4$ by \ref{item:R2}. Therefore, \cref{lem-rotation-aux} implies that $\rotation(P) = 2$. 
\end{proof}

The above lemma with $j=i+1$ and $P = P_{i \leftarrow i+1}$ implies that the last element in the sequence of numbers $Z_{i \leftarrow i+1}=(z_1, z_2, \ldots)$ is either $-2$ or $2$, depending on whether $F_{i,i+1}$ is the central face or a regular face. We will later show that $F_{i,i+1}$ must be a regular face.

\begin{lemma}\label{lem-horizontal-exist}
Consider the face $F_{i,i+1}$ for any $1 \leq i \leq s$. The facial cycle $C$ of $F_{i,i+1}$ must contain an edge $e'$ from some horizontal segment $S_l$ in $A=(S_1, S_2, \ldots, S_k)$ such that $\rotation(e_i \circ \cdots \circ e')=1$  and $\rotation(e' \circ \cdots \circ \overline{e_{i+1}})=1$ along the cycle $C$.
\end{lemma}
\begin{proof}
 See the left drawing of \cref{fig:lem-fig-2} for an illustration of the proof. Consider the path $\Pout_{i\rightarrow i+1}$, which is a subpath of $C$ starting at $e_i$ and ending at $\overline{e_{i+1}}$. By \cref{lem-rotation-aux},  $\rotation(\Pout_{i\rightarrow i+1}) = 2$, so there exists an edge $e'$ in $\Pout_{i\rightarrow i+1}$ such that $\rotation(e_i \circ \cdots \circ e')=1$  and $\rotation(e' \circ \cdots \circ \overline{e_{i+1}})=1$ along the path $\Pout_{i\rightarrow i+1}$, or equivalently along the cycle $C$. Since $e_i$ and $e_{i+1}$ are vertical, such an edge $e'$ must be horizontal. Since $e'$ is in $G_k$, $e'$ belongs to some horizontal segment $S_l$ in $A=(S_1, S_2, \ldots, S_k)$. 
\end{proof}

Combining the above lemma with the assumption that no more virtual edges can be added, we show that the strings $Z_{i \leftarrow i+1}$ and $Z_{i \rightarrow i+1}$ must satisfy some structural properties.

\begin{lemma}\label{lem-rotation-calculation}
Consider any $1 \leq i \leq s$.
If $F_{i,i+1}$ is of type $(\ast, \sqcup)$, then, except for the first number of the string, all numbers in the string  $Z_{i \leftarrow i+1}$ are at least $1$.
If $F_{i,i+1}$ is of type $(\sqcup,\ast)$, then, except for the first number of the string, all numbers in the string  $Z_{i \rightarrow i+1}$ are at most $-1$.
\end{lemma}
\begin{proof}
See the right drawing of \cref{fig:lem-fig-2} for an illustration of the proof.
  The assumption that no more virtual edges can be added, 
 together with \cref{lem-horizontal-exist}, implies that we cannot have a horizontal segment $S \in \PH$ 
 with $\Nnorth(S)=\emptyset$ satisfying the following conditions:
 \begin{itemize}
     \item $S$ is a subpath of $\overline{C_{F_{i,i+1}}}$.
     \item For each edge $e \in \overline{S}$, we have $\rotation(\overline{e_{i+1}} \circ \cdots \circ e)=1$  or $\rotation(e \circ \cdots \circ {e_{i}})=1$ along the cycle $\overline{C_{F_{i,i+1}}}$.
 \end{itemize}
 If such a horizontal segment $S$ with 
$\rotation(\overline{e_{i+1}} \circ \cdots \circ e)=1$ for all $e \in \overline{S}$ exists, then $S$ is eligible for adding a virtual edge, due to the edge $e'$ in  \cref{lem-horizontal-exist}, as
\[\rotation(e' \circ \cdots \circ e) 
= \rotation(e' \circ \cdots \circ \overline{e_{i+1}}) + \rotation(\overline{e_{i+1}} \circ \cdots \circ e)
=1+1=2
\]
along the cycle  $\overline{C_{F_{i,i+1}}}$. 
For the remaining case that 
$\rotation(e \circ \cdots \circ {e_{i}})=1$ for all $e \in \overline{S}$,  for a similar reason, $S$ is also eligible for adding a virtual edge, due to the edge $e'$ in  \cref{lem-horizontal-exist}.

  Now suppose that $F_{i,i+1}$ is of type $(\ast, \sqcup)$ and some number $z_l$ in the string  $Z_{i \leftarrow i+1}$ is $0$ and $l \neq 1$. The type $(\ast, \sqcup)$ guarantees that the string $Z_{i \leftarrow i+1}$ starts with $0 \circ 1^{c'} \circ 2$, for some $c' \geq 1$. Between this number $2$ and the above number $z_l=0$,  there must exist a substring $2 \circ 1^c \circ 0$ in $Z_{i \leftarrow i+1}$, for some $c \geq 1$. The reversal of the subpath of $P_{i \leftarrow i+1}$ corresponding to the substring $1^c$ is a horizontal segment  $S$ such that $\Nnorth(S)=\emptyset$ and $\rotation(\overline{e_{i+1}} \circ \cdots \circ e)=1$ for all $e \in \overline{S}$. Such a horizontal segment  $S$  cannot exist, due to the above discussion. Therefore,  all numbers in the string  $Z_{i \leftarrow i+1}$ must be at least $1$, except for the first number, which is always $0$.  
  
  The proof for the second statement of the lemma is similar. Suppose that $F_{i,i+1}$ is of type $(\sqcup,\ast)$ and some number $z_l$ in the string  $Z_{i \rightarrow i+1}$ is at least $0$ and $l \neq 1$. Then we can find a substring $(-2) \circ (-1)^c \circ 0$, for some $c \geq 1$, of $Z_{i \rightarrow i+1}$, and then we obtain a contradiction, as the horizontal segment corresponding to the substring $(-1)^c$ cannot exist.
\end{proof}

\begin{figure}[t!]
\centering
\includegraphics[width=0.9\textwidth]{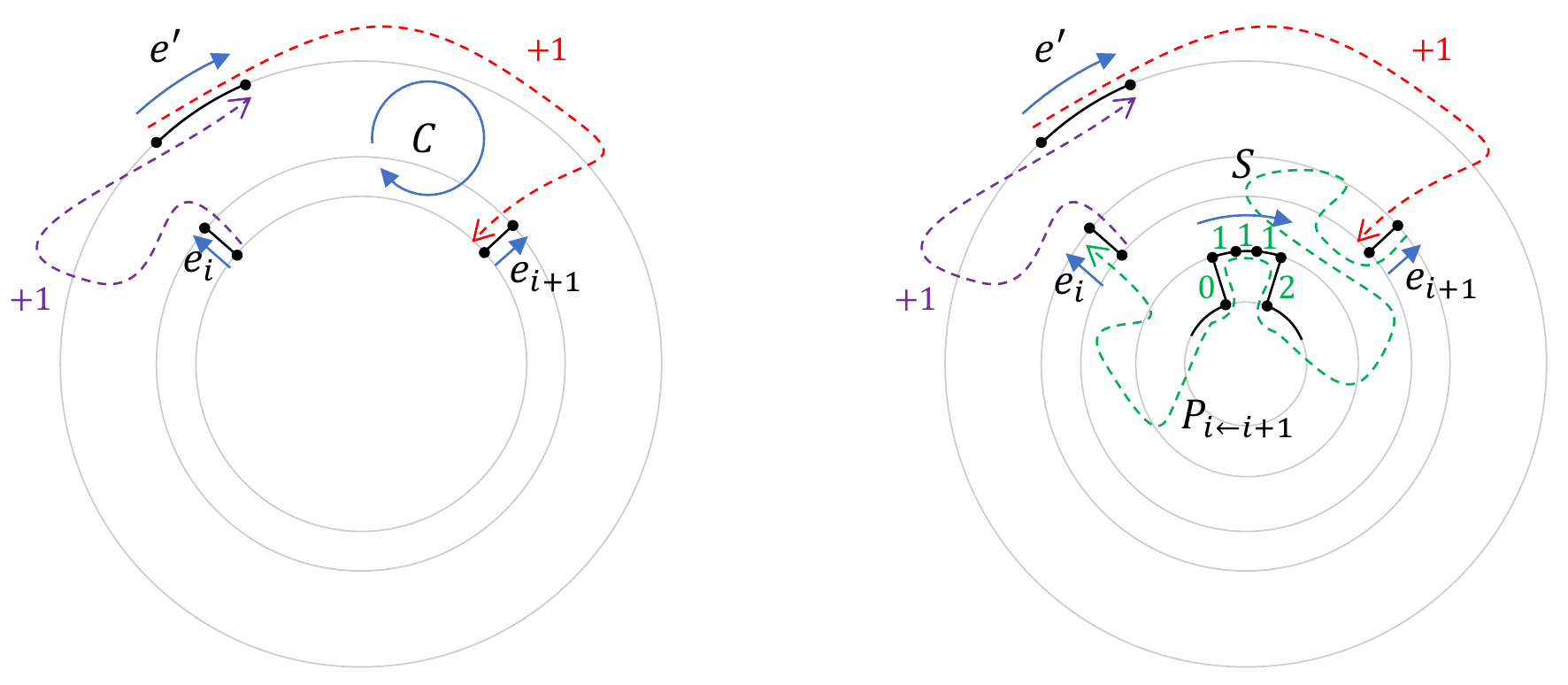}
\caption{Illustration for the proof of \cref{lem-horizontal-exist,lem-rotation-calculation}.}\label{fig:lem-fig-2}
\end{figure}

Intuitively, if a face $F_{i,i+1}$ is of type $(\sqcup,\sqcup)$, then a $\sqcap$-shape must exist in the middle of $P_{i \leftarrow i+1}$, and the horizontal segment corresponding to the middle part of the  $\sqcap$-shape must be eligible for adding a virtual edge, so $F_{i,i+1}$  cannot be of type $(\sqcup,\sqcup)$. In the following lemma, we prove this intuitive observation formally, by combining \cref{lem-P-rotation,lem-rotation-calculation}.

\begin{lemma}\label{lem-one-direction}
Consider any $1 \leq i \leq s$. Suppose that $F_{i,i+1}$ is of type $(\sqcup,\ast)$ or $(\ast, \sqcup)$. Then $F_{i,i+1}$ must be a regular face and $F_{i,i+1}$ cannot be of type $(\sqcup,\sqcup)$.
\end{lemma}
\begin{proof}
 By \cref{lem-P-rotation}, the rotation of the path $P_{i \leftarrow i+1}$ is $-2$ if $F_{i,i+1}$ is the central face, and it is $2$ if $F_{i,i+1}$ is a regular face.
Suppose that $F_{i,i+1}$ is of type $(\ast, \sqcup)$. Then \cref{lem-rotation-calculation} implies that the rotation of $P_{i \leftarrow i+1}$ is at least $1$, so $F_{i,i+1}$ must be a regular face and the rotation of $P_{i \leftarrow i+1}$  is exactly $2$, meaning that the string $Z_{i \leftarrow i+1}$ ends with the number $2$.  As a result, if  $F_{i,i+1}$ is also of type $(\sqcup,\ast)$, then $0 \circ 1^c \circ 2$, for some $c \geq 1$ will be a \emph{strict} suffix of $Z_{i \leftarrow i+1}$, violating \cref{lem-rotation-calculation}. Therefore, $F_{i,i+1}$ cannot be of type $(\sqcup,\sqcup)$.

To finish the proof, we just need to show that when $F_{i,i+1}$ is of type  $(\sqcup,\ast)$, $F_{i,i+1}$ also has to be a regular face.  Again, \cref{lem-P-rotation} implies that if $F_{i,i+1}$ is the central face, then the rotation of the path $P_{i \rightarrow i+1} = \overline{P_{i \leftarrow i+1}}$ is $2$. This contradicts \cref{lem-rotation-calculation}, since it requires the rotation of $P_{i \rightarrow i+1}$ to be at most $-1$. Therefore, $F_{i,i+1}$ is a regular face.  
\end{proof}

As discussed in the previous section, we may assume that each vertex in $\tilde{G}$ is incident to a horizontal segment. Consider the horizontal segment $S$ incident to the south endpoint $x_i$ of some $e_i \in \Nsouth(A)$. In view of \ref{item:S2}, there are two possible reasons for why adding $S$ to the current good sequence $A$ does not result in a good sequence. The first possible reason is that $\Nnorth(S)$ contains an edge that is not in $\Nsouth(A)$. The second possible reason is that there exist two edges $e$ and $e'$ such that $e'$ immediately follows $e$ in the ordering of $\Nnorth(S)$ and $e'$ does not immediately follow $e$ in the ordering of $\Nsouth(A)$. We show that the second reason is not possible. 

\begin{lemma}\label{lem-segment-ascending}
Let $S$ be any horizontal segment in $\tilde{G}$, and let  $e$ and $e'$ be any two edges such that $e'$ immediately follows $e$ in the ordering of $\Nnorth(S)$.
If both $e$ and $e'$ are in $\Nsouth(A)$, then $e'$ also immediately follows $e$ in the circular ordering of $\Nsouth(A)$.
\end{lemma}
\begin{proof}
See the left drawing of \cref{fig:lem-fig-3} for an illustration of the proof.  
 Suppose that the lemma statement does not hold. Then there exist two edges $e_i$ and $e_j$ in $\Nsouth(A)$ with $j \neq i+1$ and a horizontal segment $S \in \PH$ such that $e_j$ immediately follows $e_i$ in the ordering of $\Nnorth(S)$ and all the edges $e_{i+1}, e_{i+2}, \ldots, e_{j-1}$ are not in $\Nnorth(S)$. 
 Suppose such $e_i$, $e_j$, and $S$ exist. We select such $e_i$, $e_j$, and $S$ to minimize the number of edges after $e_i$ and before $e_j$ in the circular ordering $\Nsouth(A)$.
 
 Our choice of $e_i$, $e_j$, and $S$ implies that for each horizontal segment $S'$ such that $\Nnorth(S')$ contains an edge in $(e_{i+1}, e_{i+2}, \ldots, e_{j-1})$, the intersection of $\Nnorth(S')$ and $\Nsouth(A)$ must be a contiguous subsequence of $(e_{i+1}, e_{i+2}, \ldots, e_{j-1})$, since otherwise we should select $S'$ and not select~$S$. We  partition $(e_{i+1}, e_{i+2}, \ldots, e_{j-1})$ into groups according to the horizontal segment~$S'$ incident to the south endpoint $x_l$ of each edge $e_l=(x_l, y_l)$. In other words, if the south endpoints of two edges of $(e_{i+1}, e_{i+2}, \ldots, e_{j-1})$ are both incident to the same horizontal segment~$S'$, then these two edges are in the same group corresponding to~$S'$.
 Note that each group corresponds to a contiguous subsequence of $(e_{i+1}, e_{i+2}, \ldots, e_{j-1})$. 
 
 Consider a group $(e_a, e_{a+1}, \ldots, e_b)$, and let $S'$ be its corresponding horizontal segment, so the intersection of $\Nnorth(S')$ and $\Nsouth(A)$ is precisely the set of edges in the group. Since we assume that adding $S'$ to $A$ does not lead to a good sequence, $\Nnorth(S')$ must contain some edges that are not in $\Nsouth(A)$, so at least one of the following holds:
 \begin{itemize}
     \item $e_a$ is not the first element of  $\Nnorth(S')$, in which case the face $F_{a-1,a}$ is of type $(\ast,\sqcup)$ because $e_a$, $S'$, and the vertical edge $e'$ right before $e_a$ in $\Nnorth(S')$ form a $\sqcup$-shape that is a \emph{strict} prefix of $Z_{a-1 \leftarrow a}$, as $e' \neq e_{a-1}$.
     \item $e_b$ is not the last element of $\Nnorth(S')$, in which case the face $F_{b,b+1}$ is of type $(\sqcup,\ast)$.
 \end{itemize}
 

 We let $i = c_1 < c_2 <  \cdots < c_t = j-1$ be the sequence of all indices $a$ such that 
   $e_a$ is the last edge of a group or $e_{a+1}$ is the first edge of a group ($t=4$ in \cref{fig:lem-fig-3}). 
  Our choice of $e_i$, $e_j$, and $S$ implies that  $F_{i,i+1} = F_{c_1, c_1+1}$ must be of type $(\sqcup, \ast)$, because the path $P_{i \rightarrow i+1}$ will first make a $90^\circ$ left turn at $x_i$, go straight along $S$, and then make another $90^\circ$  left turn at $x_j$ to enter $e_j$.
  
  By \cref{lem-one-direction}, $F_{i,i+1} = F_{c_1, c_1+1}$ cannot be also of type $(\ast, \sqcup)$, as this forces the face to be of type $(\sqcup, \sqcup)$. In view of the discussion above, the fact that $F_{c_1, c_1+1}$ cannot be  of type $(\ast,\sqcup)$ forces the type of $F_{c_2, c_2+1}$ to be $(\sqcup, \ast)$. Similarly, we can argue that the types of $F_{c_3, c_3+1}, F_{c_4, c_4+1} ,\ldots$ must be  $(\sqcup, \ast)$, implying that the type of $F_{c_t, c_t+1} = F_{j-1,j}$ is also $(\sqcup, \ast)$.
  The same argument for showing that $F_{i,i+1}$ is of type $(\sqcup, \ast)$ can also be used to show that  $F_{j-1, j}$ must be of type $(\ast,\sqcup)$. Therefore,  $F_{j-1, j}$ is of type $(\sqcup, \sqcup)$, which is impossible due to \cref{lem-one-direction}, so the lemma statement holds. 
\end{proof}

\begin{figure}[t!]
\centering
\includegraphics[width=\textwidth]{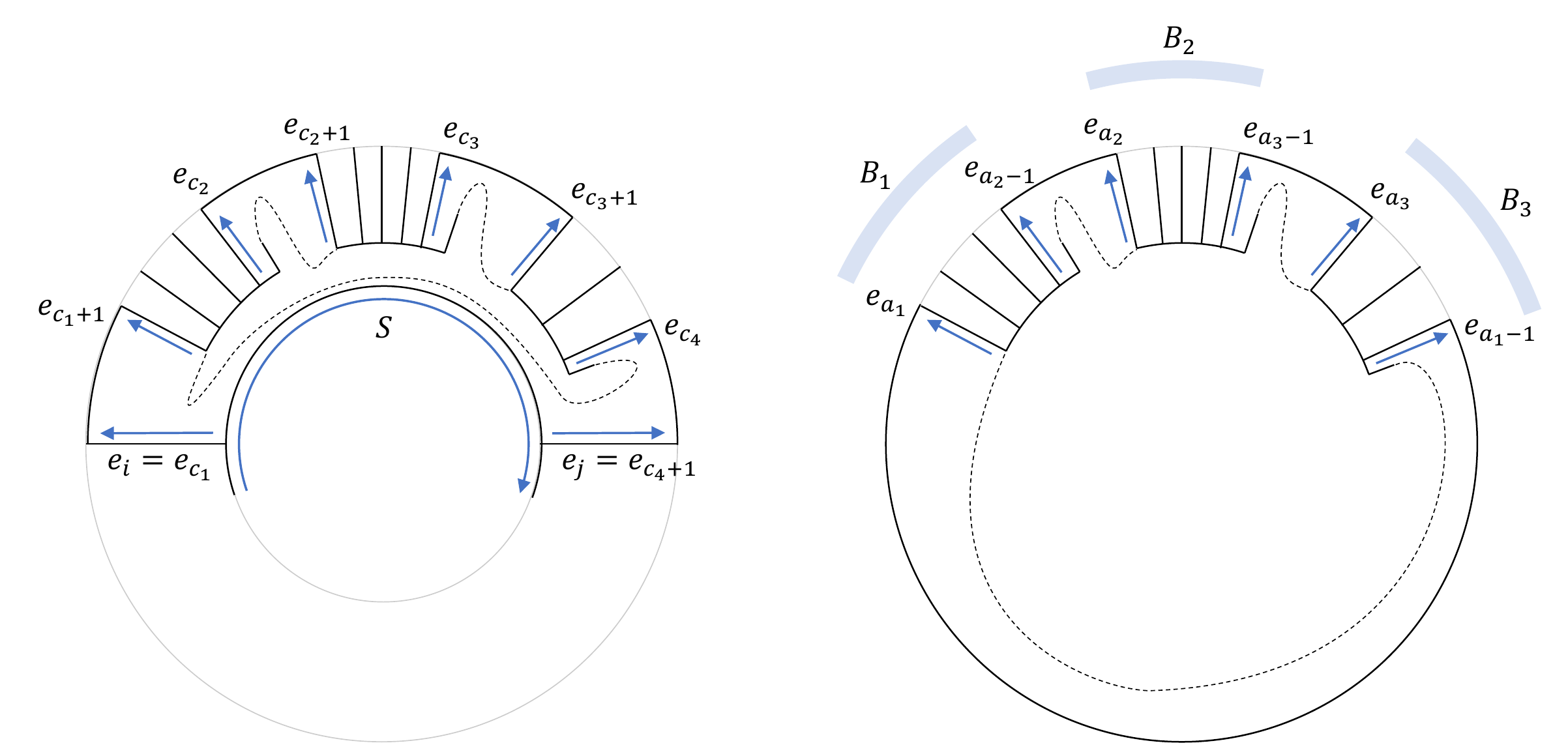}
\caption{Illustration for the proof of \cref{lem-segment-ascending,lem-all-ascending}.}\label{fig:lem-fig-3}
\end{figure}

By \Cref{lem-segment-ascending}, for each horizontal segment $S$ that is incident to the south endpoint $x_i$ of some $e_i$, $S$ must be a path, and $\Nnorth(S)$ must contain an edge $e$ that is not in $\Nsouth(A)$. 
In the following lemma, we use \Cref{lem-segment-ascending} to prove that we cannot simultaneously have a face of type $(\sqcup,\ast)$ and another face of type $(\ast,\sqcup)$.

\begin{lemma}\label{lem-all-ascending}
One of the following holds.
\begin{itemize}
    \item All faces $F_{i,i+1}$ are of  types $(-)$ and $(\sqcup,\ast)$, and at least one face $F_{i,i+1}$ is of  type $(\sqcup,\ast)$.
     \item All faces $F_{i,i+1}$ are of  types $(-)$ and $(\ast,\sqcup)$, and at least one face $F_{i,i+1}$ is of  type $(\ast,\sqcup)$.
\end{itemize}
\end{lemma}
\begin{proof}
We partition $\Nsouth(A) = (e_{1}, e_{2}, \ldots, e_{s})$ into contiguous subsequences according to the following rule: $e_i$ and $e_{i+1}$ belong to the same group if there exists a horizontal segment~$S'$ such that $e_{i+1}$ immediately follows $e_i$ in $\Nnorth(S')$. Each contiguous subsequence is called a group.


Let $t$ be the number of groups, and let $B_1, B_2, \ldots, B_t$ denote these 
 $t$ groups, circularly ordered according to their positions in the circular ordering $(e_{1}, e_{2}, \ldots, e_{s})$. Let $a_i$ denote the index such that the first edge of $B_i$ is $e_{a_i}$, so the last edge of $B_{i-1}$ is $e_{a_i - 1}$. See the right drawing of \cref{fig:lem-fig-3} for an illustration with $t=3$.

Let $S'$ be the horizontal segment corresponding to the group $B_i$ (i.e., $B_i$ is a contiguous subsequence of $\Nnorth(S')$). As discussed earlier, due to \cref{lem-segment-ascending}, $S'$ must be a path, and $\Nnorth(S')$ must contain an edge that is not in $B_i$. Therefore, we cannot have $B_i = \Nnorth(S')$, so at least one of the following holds:
\begin{itemize}
    \item There is an edge $e$ right before the first edge $e_{a_i}$ of $B_i$ in the sequential ordering of $\Nnorth(S')$. 
    Since $e$ is not in $\Nsouth(A)$, $F_{a_i - 1, a_i}$ is of type $(\ast,\sqcup)$, as $e_{a_i}$, $S'$, and $e$ form a $\sqcup$-shape.
    \item There is an edge $e$ right after the last edge $e_{a_{i+1}-1}$ of $B_i$ in the sequential ordering of $\Nnorth(S')$. 
    Since $e$ is not in $\Nsouth(A)$, $F_{a_{i+1}-1, a_{i+1}}$ is of type $(\sqcup,\ast)$.
\end{itemize}

By \cref{lem-one-direction}, for each $1 \leq i \leq t$, the face $F_{a_i - 1, a_i}$ cannot be of type $(\sqcup, \sqcup)$. Therefore, the only possibility is that they are all of type $(\ast,\sqcup)$ or all of type $(\sqcup,\ast)$.

Now consider any face $F_{l, l+1}$ that is not of the form $F_{a_i - 1, a_i}$ for some $1 \leq i \leq t$. That is, $e_l$ and $e_{l+1}$ are in the same group, meaning that their south endpoints are both incident to the same horizontal segment $S'$, and $e_{l+1}$ immediately follows $e_l$ in the sequential ordering $\Nnorth(S')$, so the face $F_{l, l+1}$ must be of type $(-)$. 
\end{proof}

Informally, the above lemma, together with \cref{lem-rotation-calculation}, implies that either all of $P_{i \rightarrow i+1}$ are monotonically ascending or all of $P_{i \leftarrow i+1}$ are monotonically ascending, so we should be able to extract a strictly monotone cycle by considering the edges in these paths. Before we do that, we first show that all faces $F_{i,i+1}$ are distinct regular faces.

\begin{lemma}\label{lem-face-regular}
For each $1 \leq i \leq s$,  $F_{i,i+1}$ is a regular face.
\end{lemma}
\begin{proof}
If $F_{i,i+1}$ is of type $(\sqcup,\ast)$ or $(\sqcup,\ast)$, then \cref{lem-one-direction} implies that $F_{i,i+1}$ is a regular face.
By \cref{lem-segment-ascending}, the only remaining case is when $F_{i,i+1}$ is of type $(-)$. Suppose that $F_{i,i+1}$ is of type $(-)$, and let $C$ be the facial cycle of $F_{i,i+1}$. Then $\rotation(C)$ equals the sum of $\rotation(\Pout_{i \rightarrow i+1})$ and  $\rotation(P_{i \leftarrow i+1})$. By the definition of the type $(-)$, we know that $\rotation(P_{i \leftarrow i+1})=2$. By  \cref{lem-rotation-aux}, we know that $\rotation(\Pout_{i \rightarrow i+1})=2$. Therefore, $\rotation(C)=4$, so $C$ is a regular face in view of \ref{item:R2}.
\end{proof}

\begin{figure}[t!]
\centering
\includegraphics[width=0.5\textwidth]{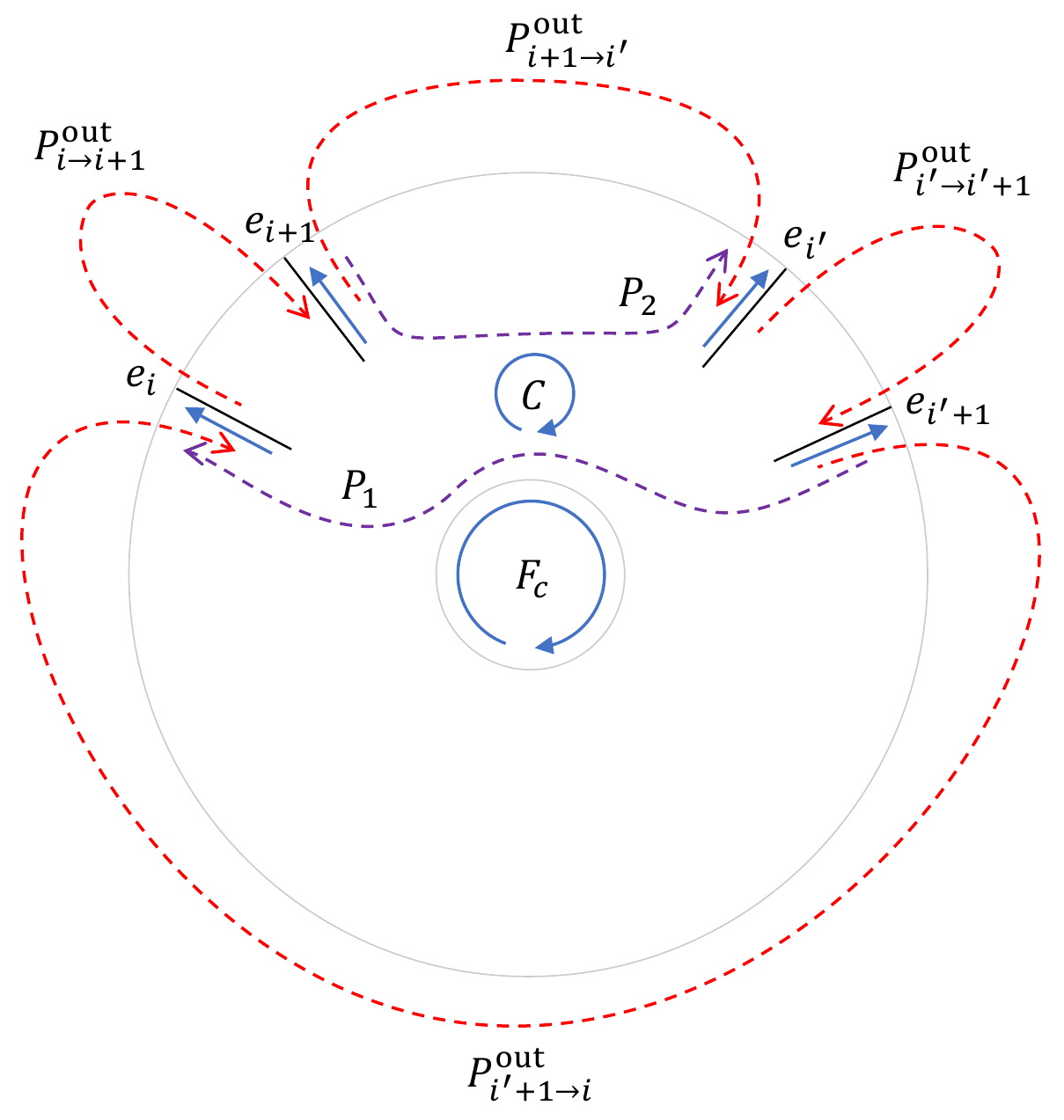}
\caption{Illustration for the proof of \cref{lem-face-distinct}.}\label{fig:lem-fig-4}
\end{figure}

\begin{lemma}\label{lem-face-distinct}
Any two faces $F_{i,i+1}$ and $F_{i',i'+1}$ with $i \neq i'$ are distinct. 
\end{lemma}
\begin{proof}
Suppose that $F_{i,i+1}$ and $F_{i',i'+1}$ are the same face $F$. Let $C$ be the facial cycle of $F$. We know that both $\Pout_{i \rightarrow i+1}$, which starts at $e_i$ and ends at $\overline{e_{i+1}}$, and $\Pout_{i' \rightarrow i'+1}$, which starts at $e_{i'}$ and ends at $\overline{e_{i'+1}}$, are subpaths of $C$. We may assume that $e_{i+1} \neq e_{i'}$ and $e_{i'+1} \neq e_i$, since otherwise $C$ is not a simple cycle, implying that the underlying graph is not biconnected.

We define $P_1$ to be the subpath of $C$ starting at $\overline{e_{i'+1}}$ and ending at $e_i$ and define $P_2$ to be the subpath of $C$ starting at $\overline{e_{i+1}}$ and ending at $e_{i'}$.
Therefore, $\rotation(C)$ is the sum of the rotation of the four paths $\Pout_{i \rightarrow i+1}$, $\Pout_{i' \rightarrow i'+1}$, $P_1$, and $P_2$. By \cref{lem-face-regular}, $F$ is a regular face, so $\rotation(C)=4$ due to \ref{item:R2}. 

The two paths $\overline{P_1}$ and $\Pout_{i'+1 \rightarrow i}$ form a cycle. The two paths $\overline{P_2}$ and $\Pout_{i+1 \rightarrow i'}$ form another cycle. Observe that the central face of $\tilde{G}$ lies in the interior of one of these two cycles. By symmetry, we may assume the central face lies in the interior of the cycle formed by $\overline{P_1}$ and $\Pout_{i'+1 \rightarrow i}$. In case the central face lies in the interior of the other cycle, we swap $i$ and $i'$. See \cref{fig:lem-fig-4} for an illustration.

By \cref{lem-P-rotation} with $P = \overline{P_2}$ and $\Pout_{i+1\rightarrow i'}$, we obtain that $\rotation(\overline{P_2}) = 2$.
Similarly, by \cref{lem-P-rotation} with $P = \overline{P_1}$ and $\Pout_{i'+1\rightarrow i}$,
we obtain that  $\rotation(\overline{P_1}) = -2$.
Therefore, we have $\rotation(P_1) = 2$ and $\rotation(P_2) = -2$. There is one subtle issue in applying \cref{lem-P-rotation}:  $P_2$ might contain vertices in $\Pout_{i+1\rightarrow i'} \setminus \{x_{i+1}, y_{i+1}, x_{i'}, y_{i'}\}$ and $P_1$ might contain vertices in $\Pout_{i'+1\rightarrow i} \setminus \{x_{i'+1}, y_{i'+1}, x_{i}, y_{i}\}$, so the condition for applying  \cref{lem-P-rotation} is not met. This issue can be overcome by choosing $i$ and $i'$ with $F_{i,i+1} = F_{i',i'+1}$ in such a way that the number of edges after $e_{i+1}$ and before $e_{i'}$ in the circular ordering $(e_{1}, e_{2}, \ldots, e_{s})$ is minimized. This forces $P_2$ to not contain any edges in $(e_{i+2}, \ldots, e_{i'-1})$ and their reversal, meaning that $P_2$ cannot contain any vertex in $\Pout_{i+1\rightarrow i'} \setminus \{x_{i+1}, y_{i+1}, x_{i'}, y_{i'}\}$.  Being able to apply  \cref{lem-P-rotation} to one of $\overline{P_1}$ and $\overline{P_2}$ is enough, since we already know that $\rotation(C)=4$, $\rotation(\Pout_{i \rightarrow i+1}) = 2$ and  $\rotation(\Pout_{i' \rightarrow i'+1})=2$ by \cref{lem-rotation-aux}. These rotation numbers force $\rotation(P_2) = -\rotation(P_1)$.

By \cref{lem-P-rotation}, the rotation of $P_{i' \rightarrow i'+1}$ is $-2$, so the last number of the string $Z_{i' \rightarrow i'+1}$ is $-2$.
Observe that $\overline{P_1}$ is a suffix of the path $P_{i' \rightarrow i'+1}$ and $\overline{P_1} \neq P_{i' \rightarrow i'+1}$, so $\rotation(\overline{P_1}) = -2$ implies that  
$Z_{i' \leftarrow i'+1}$ contains a number $0$ that is not the first number of the string. Consequently,  $F_{i',i'+1}$ cannot be of type $(-)$.
Also, $F_{i',i'+1}$  cannot be of type $(\sqcup,\ast)$, since a requirement for type $(\sqcup,\ast)$ due to \cref{lem-rotation-calculation} is that all numbers in the string  $Z_{i' \rightarrow i'+1}$ are at most $-1$, except for the first number of the string. \cref{lem-all-ascending} forces $F_{i',i'+1}$ to be of  type $(\ast,\sqcup)$. 

By a similar argument that considers the suffix $P_1$  of the path $P_{i \leftarrow i+1}$, we may also force $F_{i,i+1}$ to be of type $(\sqcup,\ast)$. The fact that both types $(\ast,\sqcup)$ and $(\sqcup,\ast)$ exist is a contradiction to \cref{lem-all-ascending}, so $F_{i,i+1}$ and $F_{i',i'+1}$ cannot be the same face.
\end{proof}

\paragraph{Strictly monotone cycles}
As a consequence of the above lemma, all paths $P_{i \rightarrow i+1}$ cannot use any edge that is in $G_{k}$. Although the starting and the ending edges of $P_{i \rightarrow i+1}$ might be virtual edges, the remaining edges of $P_{i \rightarrow i+1}$ must appear in the original graph $G$. We let $H$ denote the subgraph of $G$ induced by the set of the undirected version of all edges of $P_{i \rightarrow i+1}$, except for the first edge $\overline{e_i}$ and the last edge $e_{i+1}$, for all $1 \leq i \leq s$.  We now show that a strictly monotone cycle of $G$ can be found by considering the central face of $H$.

\begin{lemma}\label{lem-central-face}
The facial cycle $C$ of the central face of $H$ is a strictly monotone cycle of $G$.
\end{lemma}
\begin{proof}
By \cref{lem-face-regular,lem-face-distinct}, the cycle formed by concatenating the paths $P_{i \rightarrow i+1}$, excluding the first edge $\overline{e_i}$ and the last edges $e_{i+1}$, over all $1 \leq i \leq s$, separate the central face of $\tilde{G}$ from the outer face of $\tilde{G}$ and all faces $F_{i,i+1}$, for all $1 \leq i \leq s$.
Therefore, the central face and the outer face of $H$ are distinct, where the central face of $H$ contains the central face of $\tilde{G}$, and the outer face of $H$ contains the outer face of $\tilde{G}$ and also the faces $F_{i,i+1}$, for all $1 \leq i \leq s$. 

Let $C$ be the facial cycle of the central face of $H$.  We claim that $C$ must be a simple cycle, so the above discussion implies that $C$ is an essential cycle of $\tilde{G}$. To see this, consider any edge $e \in C$. Since the undirected version of $e$ is in $H$, at least one of $e$ and $\overline{e}$ is an edge in a path $P_{i \rightarrow i+1}$, for some $1 \leq i \leq s$.
We cannot have $\overline{e} \in P_{i' \rightarrow i'+1}$ for any $1 \leq i' \leq s$, since otherwise the face $F_{i',i'+1}$ is incident to $e$ from the right, meaning that $F_{i',i'+1}$ is contained in the central face of $H$, which is impossible. Therefore, for each $e \in C$, we have $e \in P_{i \rightarrow i+1}$ for some $1 \leq i \leq s$. From this discussion, we infer that we cannot simultaneously have $e \in C$ and $\overline{e} \in C$, and this forces $C$ to a simple cycle, as $C$ is a facial cycle. 

\paragraph{Edge labels}
Next, consider any $e \in C$. We calculate $\ell_C(e)$ with respect to the reference edge~$\eref$ in $\tilde{G}$.
By \cref{lem-all-ascending}, either all faces $F_{i,i+1}$ are of the type $(-)$ and $(\sqcup,\ast)$ or all faces $F_{i,i+1}$ are of the type $(-)$ and $(\ast,\sqcup)$. We claim that $\ell_C(e) \leq 0$ for all $e \in C$ in the first case, and $\ell_C(e) \geq 0$ for all $e \in C$ in the second case. 

Suppose all faces $F_{i,i+1}$ are of the type $(-)$ and $(\sqcup,\ast)$.
Fix an $e \in C$, and let $e_i \in \Nsouth(A)$ be chosen such that $e \in P_{i \rightarrow i+1}$. 
We calculate $\ell_C(e)$ by considering any crossing-free path $P=(\eref, \ldots, \overline{e_i})$ of $G_k$ from $\eref$ to $\overline{e_i}$ and considering the subpath $P'=\overline{e_i} \circ \cdots \circ e$ of $P_{i \rightarrow i+1}$. Then $\ell_C(e) = \rotation(P) + \rotation(P')$. By \cref{lem-e-direction}, we have $\rotation(P)=1$. By \cref{lem-rotation-calculation}, we have $\rotation(P') \leq -1$, as we cannot have $e = e_i$.
Therefore, $\ell_C(e) \leq 0$.

Suppose all faces $F_{i,i+1}$ are of the type $(-)$ and $(\ast, \sqcup)$. Similarly, by combining \cref{lem-rotation-calculation} with the fact that $F_{i,i+1}$ is a regular face, we obtain that the rotation from $\overline{e_i}$ to any intermediate edge $e$ of $P_{i \rightarrow i+1}$ along the path $P_{i \rightarrow i+1}$ is at least $-1$. Therefore, by the same argument as above, we obtain that $\ell_C(e) \geq 0$ for all $e \in C$.

\begin{sidefigure}[t!]
\centering
\includegraphics[width=0.3\textwidth]{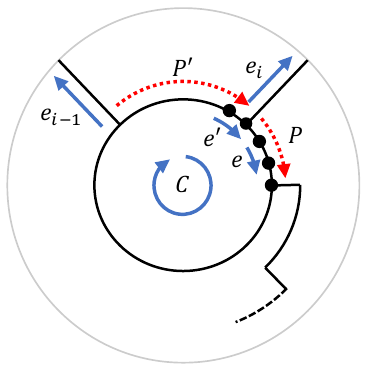}
\caption{Showing that $C$ is strictly monotone.}\label{fig:strict-monotone}
\end{sidefigure}

\paragraph{Strict monotonicity}
To finish the proof, we still need to show that $C$ is \emph{strictly} monotone. That is, we need to show that we cannot have $\ell_C(e) = 0$ for all $e \in C$. Now suppose that $\ell_C(e) = 0$ for all $e \in C$. We aim to show that this assumption implies that all faces $F_{i,i+1}$ are of the type $(-)$, contradicting \cref{lem-all-ascending}. See \cref{fig:strict-monotone} for an illustration of this part of the proof.

Again, here we assume that all faces $F_{i,i+1}$ are of the type $(-)$ and $(\sqcup,\ast)$. The proof for the other case where all faces $F_{i,i+1}$ are of the type $(-)$ and $(\ast, \sqcup)$ is similar. 
Consider any  $e \in C$.
We know that $e \in P_{i \rightarrow i+1}$ for some $1 \leq i \leq s$.
We extend $e$ into a subpath $P$ of $C$ such that  $P$  is a subpath of $P_{i \rightarrow i+1}$, and both the edge immediately preceding $P$ and the edge immediately following $P$ in $C$ are not in $P_{i \rightarrow i+1}$. 
Since the edge labels $\ell_C(\tilde{e})$ of all  edges $\tilde{e}$ of $P$ are $0$, the rotation from $\overline{e_i}$ to any edge  $\tilde{e}$ of $P$ along the path $P_{i \rightarrow i+1}$ must be $-1$. By also considering the two edges of $P_{i \rightarrow i+1}$ immediately before and after $P$, we obtain a sequence of rotation numbers $0 \circ (-1)^c \circ (-2)$ that is a substring of $Z_{i \rightarrow i+1}$.  Such a substring corresponds to a $\sqcup$-shape in the drawing.

There are two cases. If $F_{i,i+1}$ is of type $(-)$, then the above string $0 \circ (-1)^c \circ (-2)$ must be the entire string $Z_{i \rightarrow i+1}$. If $F_{i,i+1}$ is of type $(\sqcup,\ast)$, then in view of \cref{lem-rotation-calculation}, the above string $0 \circ (-1)^c \circ (-2)$ must be a prefix of $Z_{i \rightarrow i+1}$. In either case, the edge in $P_{i \rightarrow i+1}$ immediately before~$P$ must be $\overline{e_i}$. 

Now consider the edge $e'$ in $C$ immediately before $P$. This edge $e'$ must be the second last edge in $P_{i-1 \rightarrow i}$, that is, $e'$ is immediately before the last edge $e_i$ of $P_{i-1 \rightarrow i}$.  By repeating the above analysis with $e'$ instead of $e$, we obtain a path $P'$ that contains $e'$.
Because here $e'$ is the second last edge in $P_{i-1 \rightarrow i}$, the path $P'$ must cover all intermediate edges of $P_{i-1 \rightarrow i}$. In other words, $P_{i-1 \rightarrow i}$ must have the following structure: After the first edge $\overline{e_{i-1}}$, make a $90^\circ$ left turn to enter~$C$,  go straight along $C$, and then make another $90^\circ$ left turn from $e'$ to $e_i$, so $F_{i-1,i}$ is of type $(-)$. 

We may repeat the same analysis for $F_{i-2,i-1}$, $F_{i-3,i-2}$, and so on, to infer that all of them are of type $(-)$, contradicting \cref{lem-all-ascending}, so we cannot have $\ell_C(e) = 0$ for all $e \in C$. Therefore, $C$ is strictly monotone.
\end{proof}

\begin{figure}[t!]
\centering
\includegraphics[width=0.7\textwidth]{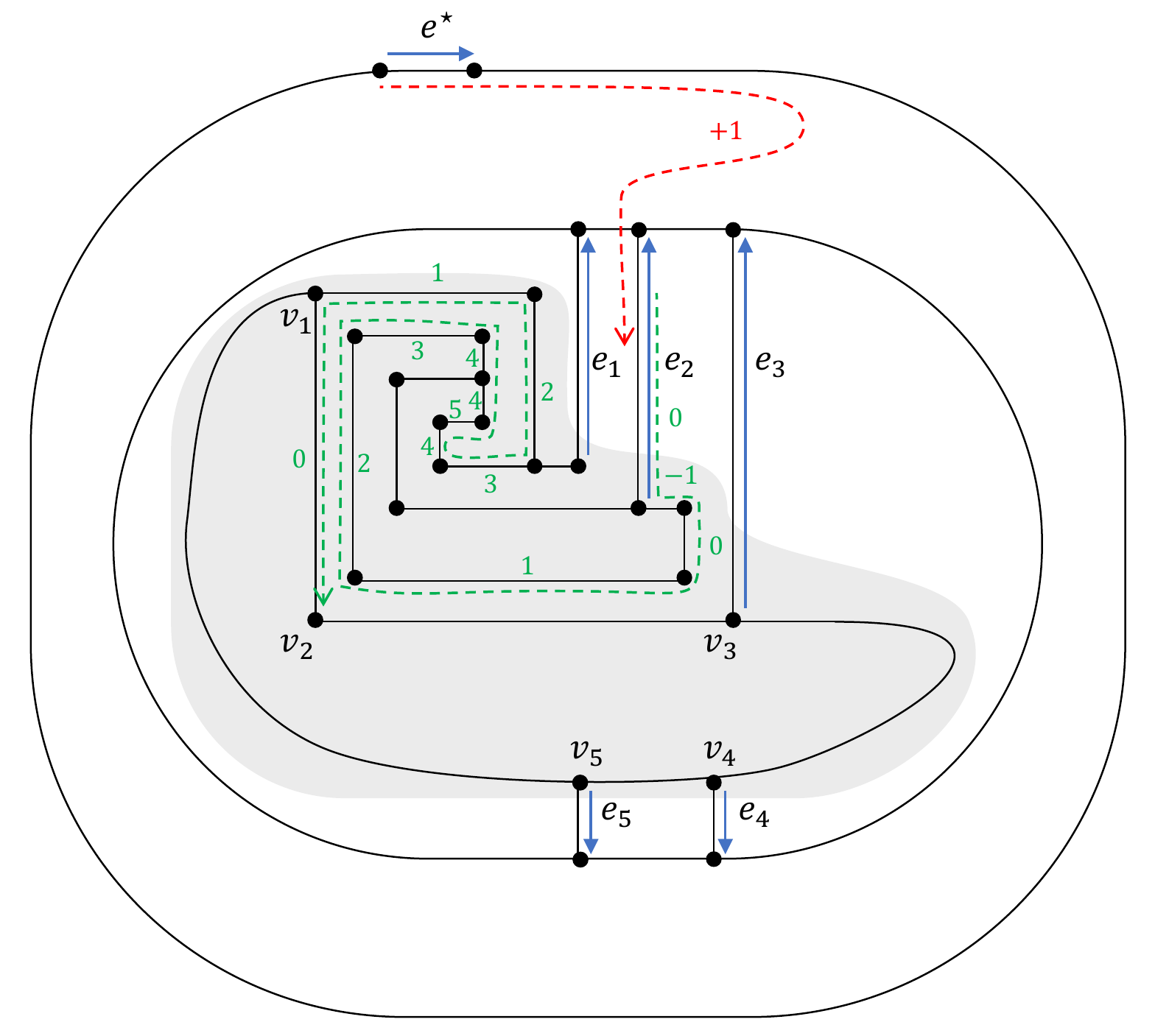}
\caption{Extracting a strictly monotone cycle $C=(v_1, v_2, \ldots, v_5)$.}\label{fig:strict-monotone-main}
\end{figure}

Consider \cref{fig:strict-monotone-main} for an example of extracting a strictly monotone cycle. In the figure, $H$ is the subgraph induced by the vertices in the shaded area. 
The cycle $C=(v_1, v_2, \ldots, v_5)$ is the facial cycle of the central face of $H$.  In this example, $\Nsouth(A) = (e_1, e_2, \ldots, e_5)$. The faces $F_{5,1}$, $F_{1,2}$, and $F_{2,3}$ are of type $(\ast, \sqcup)$. The faces $F_{3,4}$ and $F_{4,5}$ are of type $(-)$. The cycle $C$ is strictly monotone, as it is increasing. We can calculate that $\ell_C((v_1, v_2)) = 1$ by first going from $\eref$ to $\overline{e_2}$ via a crossing-free path $P$ and then going from $\overline{e_2}$ to $(v_1, v_2)$ along the path $P_{2 \rightarrow 3}$, as  $(v_1, v_2)$ is an intermediate edge of $P_{2 \rightarrow 3}$. The first part has rotation $1$ and the second part has rotation $0$, so the overall rotation is $1$. Similarly, we can calculate that $\ell_C(e)=0$ for each remaining edge $e$ in $C$. 
We summarize the discussion of this section as a lemma.

\begin{lemma}\label{lem-cycle-summary}
Suppose the greedy algorithm outputs a good sequence $A=(S_1, S_2, \ldots, S_k)$ that does not cover the set of all horizontal segments $\PH$. Then a strictly monotone cycle of the original graph $G$ can be computed in $O(n)$ time.
\end{lemma}
\begin{proof}
 By \cref{lem-central-face}, the facial cycle  $C$  of the central face of $H$ is a strictly monotone cycle of $G$. The computation of $H$ and $C$ can be done in $O(n)$ time.
\end{proof}

We are now ready to prove \cref{thm-main1}.

\thmone*

\begin{proof}
We run the $O(n \log n)$-time greedy algorithm of \cref{lem-time-sequence} for the given input $(\RR, \eref)$  such that $\Nnorth(S) = \emptyset$ for the horizontal segment $S\in \PH$ that contains $\eref$. 
There are two possible outcomes of the algorithm. If we obtain a good sequence that covers all horizontal segments $\PH$, then we may use \cref{lem:good-drawing} to compute a drawing of $(\RR, \eref)$. Otherwise, the algorithm stops with a good sequence that does not cover all horizontal segments $\PH$, and no more progress can be made, in which case a strictly monotone cycle can be computed in $O(n)$ time by \cref{lem-cycle-summary}.
\end{proof}

\section{Searching for the reference edge}\label{sect:binary-search}

In this section, we demonstrate how the drawing algorithm of \cref{thm-main1}, designed for an ortho-radial representation $\RR$ with a fixed reference edge $\eref$, can be extended to handle an ortho-radial representation $\RR$ without a fixed reference edge. This extension incurs an additional $O(\log n)$ factor in time complexity. The core idea is to use a binary search to find a reference edge $\eref$ for which $(\RR, \eref)$ is drawable, provided such an edge exists.

Throughout this section, we write $\ell_C^{\eref}(e)$ to denote the edge label of $e$ with respect to $C$ using $\eref$ as the reference edge. We first show the following lemma, which is essentially the same as~\cite[Lemma 4.2]{barth2023topology}. We still include a proof for the sake of completeness.

\begin{lemma}\label{lem:ref-edge}
For any two choices of the edges $\eref_1$ and $\eref_2$ in $\overline{C_{\FO}}$, for any essential cycle $C$ and any edge $e$ in $C$, we have $\ell_C^{\eref_2}(e)=\ell_C^{\eref_1}(e) + \ell_{\overline{C_{\FO}}}^{\eref_2}(\eref_1)$.
\end{lemma}
\begin{proof}
In the proof, we may assume that the three edges $e$, $\eref_1$, and $\eref_2$ are distinct.
If  $\eref_1 = \eref_2$, then the equality is trivial as $\ell_C^{\eref_2}(e)=\ell_C^{\eref_1}(e)$ and $\ell_{\overline{C_{\FO}}}^{\eref_2}(\eref_1) = 0$. If  $e = \eref_1$, then the equality is trivial as $\ell_C^{\eref_1}(e) = 0$ and $\ell_C^{\eref_2}(e)=  \ell_{\overline{C_{\FO}}}^{\eref_2}(\eref_1)$, as the reference path $P$ used to calculate $\ell_{\overline{C_{\FO}}}^{\eref_2}(\eref_1)$ can also be used to calculate $\ell_C^{\eref_2}(e)$.
If  $e = \eref_2$, then $\ell_C^{\eref_2}(e)=0$ and $\ell_C^{\eref_1}(e) + \ell_{\overline{C_{\FO}}}^{\eref_2}(\eref_1) = 0$, because $\ell_C^{\eref_1}(e)=\ell_{\overline{C_{\FO}}}^{\eref_1}(\eref_2)$ implies $\ell_C^{\eref_1}(e) = - \ell_{\overline{C_{\FO}}}^{\eref_2}(\eref_1)$.
We also observe that $e$ cannot be $\overline{\eref_1}$ and $\overline{\eref_2}$, as there is no essential cycle that contains  $\overline{\eref_1}$ or $\overline{\eref_2}$.

\paragraph{Finding a suitable reference path}
We pick a path $P$ satisfying the following conditions.
\begin{enumerate}
    \item $P$ is a simple path residing in the exterior of $C$, starting at $\eref_1$ or $\overline{\eref_1}$ and ending at $e$ or  $\overline{e}$.
    \item Once $P$ leaves $\overline{C_{\FO}}$ it never visit any vertex of $\overline{C_{\FO}}$ again. More formally, if we write $P = e_1 \circ e_2 \circ \cdots \circ e_s$, then the requirement is that there exists an index $i$ such that  $e_1 \circ e_2 \circ \cdots \circ e_i$ is a subpath of $\overline{C_{\FO}}$ and the only vertex of  $e_{i+1} \circ e_{i+2} \circ \cdots \circ e_s$ that belongs to $\overline{C_{\FO}}$ is its first endpoint.
    \item $P$ does contain $\eref_2$ and $\overline{\eref_2}$.
\end{enumerate}
Such a path $P$ can be found as follows (see \cref{fig:pathfind} for an illustration of the process). First, we select $P$ as any path satisfying the first condition. To make it satisfy the remaining two conditions, we let $v$ be the last vertex of $\overline{C_{\FO}}$ used in $P$, and we decompose $P$ in such a way that $P=P_s \circ P_t$, where $P_s$ ends at $v$ and $P_t$ starts at $v$. 

For the case $P_t = \emptyset$, we know that $e$ or $\overline{e}$ is on $\overline{C_{\FO}}$, so naturally there are two choices of subpaths of  $\overline{C_{\FO}}$ connecting $\eref_1$ or $\overline{\eref_1}$ to $e$ or  $\overline{e}$, and one of them does not contain $\eref_2$ and $\overline{\eref_2}$, so this gives us a desired path $P$. 

For the case $P_t \neq \emptyset$, we may replace $P_s$ with the path from $\eref_1$ to $v$ in $\overline{C_{\FO}}$ or the path from $\overline{\eref_1}$ to $v$ in ${C_{\FO}}$
to satisfy the second condition, and one of these two choices satisfies the third condition.

Let $P=\tilde{\eref_1} \circ \cdots \circ \tilde{e}$ be a path satisfying the above conditions. Here $\tilde{\eref_1}$ is either $\eref_1$ or $\overline{\eref_1}$, and $\tilde{e}$ is either  $e$ or  $\overline{e}$. We let $P^+=\tilde{\eref_2} \circ \cdots \circ \tilde{\eref_1} \circ \cdots \circ \tilde{e}$ be a simple path that extending $P$ in such a way that $\tilde{\eref_2}$ is either $\eref_2$ or $\overline{\eref_2}$, and the part $\tilde{\eref_2} \circ \cdots \circ \tilde{\eref_1}$ lies on $\overline{C_{\FO}}$ (see \cref{fig:pathfind}). Such an extension is possible due to the requirements of $P$ specified above.  

\begin{figure}[t!]
\centering
\includegraphics[width=0.9\textwidth]{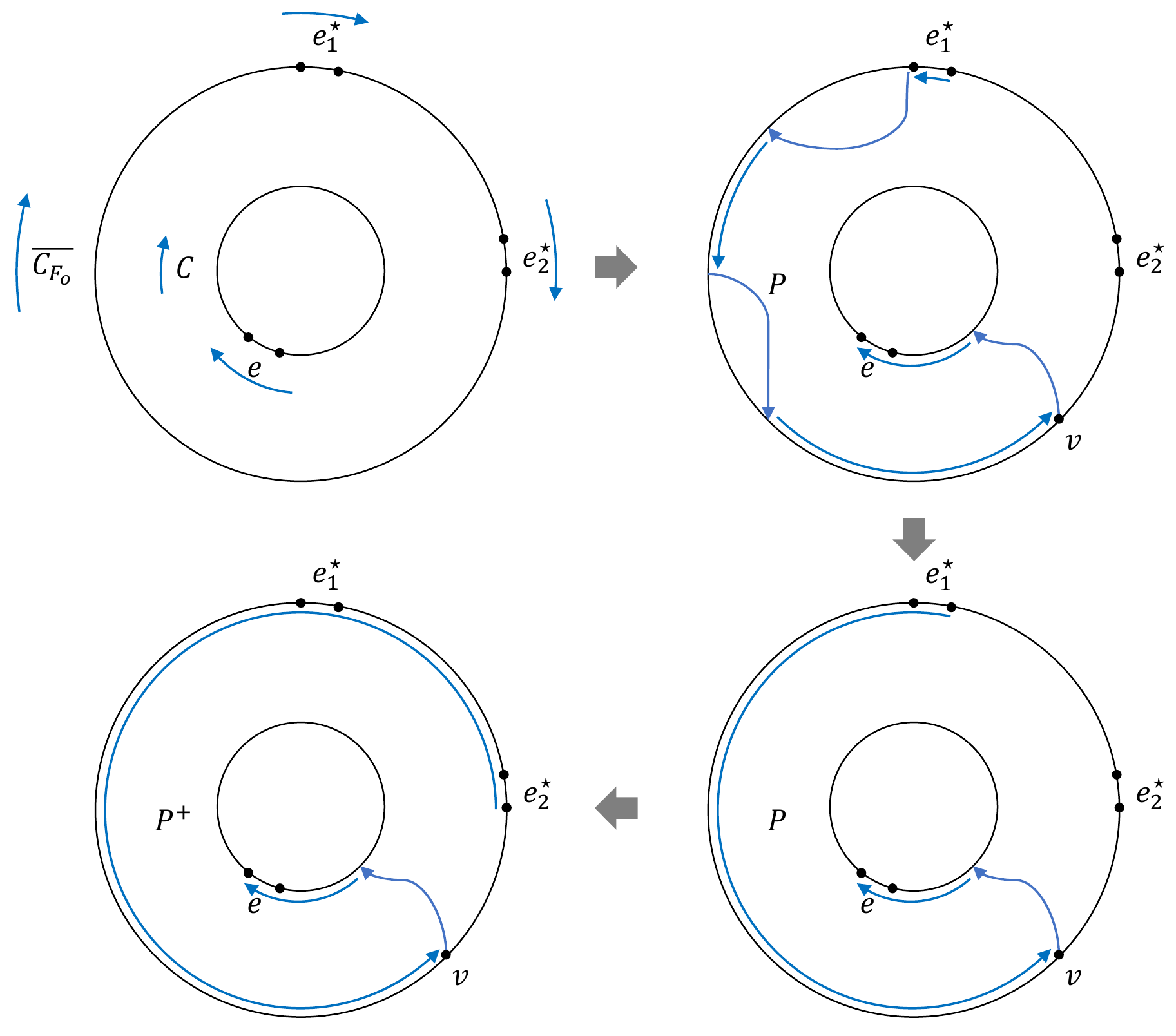}
\caption{Finding a suitable reference path in the proof of \cref{lem:ref-edge}.}\label{fig:pathfind}
\end{figure}

\paragraph{Edge label calculation}
We may use $P^+=\tilde{\eref_2} \circ \cdots \circ \tilde{\eref_1} \circ \cdots \circ \tilde{e}$, excluding the two endpoints, as a reference path for the calculation of $\ell_C^{\eref_2}(e)$.
By the formula of $\direction$, we may write 
\[\ell_C^{\eref_2}(e) = \rotation(P^+) +2b_2 - 2b,\]
where $b_2 \in \{0,1\}$ is the indicator of whether $\tilde{\eref_2} = \overline{\eref_2}$, and similarly $b \in \{0,1\}$ is the indicator of whether $\tilde{e} = \overline{e}$. We  break the calculation of $\rotation(P^+)$ into two parts:
\[\rotation(P^+) = \rotation(\tilde{\eref_2} \circ \cdots \circ \tilde{\eref_1}) + \rotation(\tilde{\eref_1} \circ \cdots \circ \tilde{e}).\]
Similarly, $\tilde{\eref_2} \circ \cdots \circ \tilde{\eref_1}$, excluding the two endpoints, can be used as a reference path for the calculation of $\ell_{\overline{C_{\FO}}}^{\eref_2}(\eref_1)$, so we can infer that
\[\ell_{\overline{C_{\FO}}}^{\eref_2}(\eref_1) = \rotation(\tilde{\eref_2} \circ \cdots \circ \tilde{\eref_1}) +2b_2 - 2b_1,\]
where $b_1 \in \{0,1\}$ is the indicator of whether $\tilde{\eref_1} = \overline{\eref_1}$. We may also write $\ell_C^{\eref_1}(e)$ in terms of $\rotation(\tilde{\eref_1} \circ \cdots \circ \tilde{e})$, as follows:
\[\ell_C^{\eref_1}(e) = \rotation(\tilde{\eref_1} \circ \cdots \circ \tilde{e}) +2b_1 - 2b.\]
Combining these formulas, we obtain the desired equality \[\ell_C^{\eref_2}(e)=\ell_C^{\eref_1}(e) + \ell_{\overline{C_{\FO}}}^{\eref_2}(\eref_1),\] 
as all $\pm2b$, $\pm2b_1$, and $\pm2b_2$ cancel out. 
\end{proof}

 In particular, if $\ell_{\overline{C_{\FO}}}^{\eref_2}(\eref_1) = 0$, then the above lemma implies that the edge label $\ell_C(e)$ is the same regardless of $\eref = \eref_1$ or $\eref = \eref_2$. 

In the following discussion, we fix an edge $e'$ in $\overline{C_{\FO}}$, and let $I$ be the range of possible values of $\ell_{\overline{C_{\FO}}}^{e'}(e)$ over all $e$ in $\overline{C_{\FO}}$, then $I$ is a contiguous sequence of integers with $0 \in I$. For each essential cycle $C$, let $I_C$ be the range of possible values of $\ell_C^{e'}(e)$  over all $e$ in $C$, then $I_C$ is also a contiguous sequence of integers.

Suppose the reference edge in use is $\eref$. Recall that an essential cycle $C$ is increasing if the edge labels $\ell_C^{\eref}(e)$, for all $e \in C$, are non-negative, and there exists $e \in C$ such that $\ell_C^{\eref}(e) \geq 1$.
Then \cref{lem:ref-edge} implies that $C$ is an increasing cycle (with $\eref$ as the reference edge) if and only if the following condition is met:
\begin{align*}
    &\ell_{\overline{C_{\FO}}}^{\eref}(e') + \min I_C \geq 0, && & \text{if $|I_C| \geq 2$,}\\
    &\ell_{\overline{C_{\FO}}}^{\eref}(e') + \min I_C \geq 1, && & \text{if $|I_C| = 1$.}
\end{align*}
Similarly, by \cref{lem:ref-edge}, $C$ is a decreasing cycle  (with $\eref$ as the reference edge) if and only if the following condition is met:
\begin{align*}
    &\ell_{\overline{C_{\FO}}}^{\eref}(e') + \max I_C \leq 0, && & \text{if $|I_C| \geq 2$,}\\
    &\ell_{\overline{C_{\FO}}}^{\eref}(e') + \max I_C \leq -1, && & \text{if $|I_C| = 1$.}
\end{align*}
In view of the above, there exist two integers $L_C$ and $U_C$, depending only on $I_C$, such that $C$ is an increasing cycle if and only if 
$\ell_{\overline{C_{\FO}}}^{\eref}(e') > U_C$,  and  $C$ is a decreasing cycle if and only if $\ell_{\overline{C_{\FO}}}^{\eref}(e') < L_C$. Therefore, $C$ is not a strictly monotone cycle (with $\eref$ as the reference edge) if and only if \[\ell_{\overline{C_{\FO}}}^{\eref}(e') \in [L_C, U_C].\] In particular, \ref{item:R3} is satisfied if and only if the intersection of $[L_C, U_C]$ over all essential cycles $C$ is non-empty. In other words,  \ref{item:R3} is satisfied if and only if $L^\ast  \leq U^\ast$, where we define $L^\ast = \max L_C$ and $\min U_C = U^\ast$, where the range of minimum and maximum is the set of all essential cycles $C$.

\paragraph{Binary search}
Given that \ref{item:R1}--\ref{item:R3} are satisfied for the given ortho-radial representation~$\RR$, we may use a binary search to find a reference edge $\eref$ such that  $(\RR, \eref)$ is drawable, as follows. We fix any edge $e'$ in $\overline{C_{\FO}}$ and compute the interval $I$ defined above. We start the binary search with the interval $I$.
In each step of the binary search, we let $x$ be a median of the current interval, and then we run the algorithm of \cref{thm-main1} with a choice of the reference edge $\eref$ such that $\ell_{\overline{C_{\FO}}}^{\eref}(e') = x$.

If $x > U^\ast$, then there exists an increasing cycle, and the algorithm must return an increasing cycle, as there is no decreasing cycle, because $x > U^\ast \geq L^\ast$. In this case, we update the upper bound of the current interval to $x-1$ as we learn that $x > U^\ast$. If $x <  L^\ast$, then there exists a decreasing cycle, and the algorithm must return a decreasing cycle, as there is no increasing cycle because $x < L^\ast \leq U^\ast$. In this case, we update the lower bound of the current interval to $x+1$ as we learn that $x < L^\ast$. If $x \in [L^\ast, U^\ast]$, then there are no increasing cycles and no decreasing cycles, so the algorithm will return a drawing of $(\RR, \eref)$. 

\paragraph{Finding a suitable reference edge}
There is still one remaining issue in implementing the above approach: The algorithm of \cref{thm-main1} requires that the reference edge $\eref$ lies on a horizontal segment $S \in \PH$ with $\Nnorth(S) = \emptyset$. Equivalently, a choice of the reference edge $\eref \in \overline{C_{\FO}}$ satisfies this requirement if and only if $\eref \in S$ such that $S$ meets  one of the following requirements.
\begin{itemize}
    \item $S = \overline{C_{\FO}}$ and $\ell_{\overline{C_{\FO}}}(e) = 0$ for all edges $e$ in $S$.
    \item $S$ is a subpath of $\overline{C_{\FO}}$ such that the following holds.
    \[
    \ell_{\overline{C_{\FO}}}(e) =\begin{cases}
    -1, & \text{for the edge $e$ immediately preceding $S$ in $\overline{C_{\FO}}$.}\\
    0, & \text{for all edges $e$ in $S$,}\\
    1, & \text{for the edge $e$ immediately following $S$ in $\overline{C_{\FO}}$.}
    \end{cases}
    \]
\end{itemize}

We claim that for any  $x \in  \left[L_{\overline{C_{\FO}}}, U_{\overline{C_{\FO}}}\right]$, we may find an edge $\eref$ with $\ell_C^{\eref}(e') = x$ and satisfying the above requirement, so we may use this edge $\eref$ as the reference edge when we run the algorithm of \cref{thm-main1} during the binary search. Note that if $x \notin \left[L_{\overline{C_{\FO}}}, U_{\overline{C_{\FO}}}\right]$, then the cycle $\overline{C_{\FO}}$ is already a strictly monotone cycle under any choice of $\eref$ with $\ell_C^{\eref}(e') = x$, so there is no need to run the algorithm of \cref{thm-main1}.

\begin{lemma}\label{lem:ref-edge-2}
For each  $x \in \left[L_{\overline{C_{\FO}}}, U_{\overline{C_{\FO}}}\right]$ there exists an edge $\eref \in \overline{C_{\FO}}$ with $\ell_{\overline{C_{\FO}}}^{\eref}(e') = x$ such that if $\eref$ is used as the reference edge, then $\eref$ belongs to a horizontal segment $S \in \PH$ with $\Nnorth(S) = \emptyset$.
\end{lemma}
\begin{proof}
For notational simplicity, in this proof, we write $C = \overline{C_{\FO}}$.
 We first start with any choice of  $\eref \in C$ with $\ell_C^{\eref}(e') = x$ and use~$\eref$ as the reference edge. Since $x \in \left[L_{C}, U_{C}\right]$, $C$ is not a strictly monotone cycle. If $\ell_C(e) = 0$ for all $e \in C$, then $S = C$ itself is already  a horizontal segment $S \in \PH$ with $\Nnorth(S) = \emptyset$, so we are done. Otherwise, $C$ contains edges with labels $-1$ and $1$. We select two edges $e^- \in C$  and $e^+ \in C$ in such a way that $\ell_C(e^-) = -1$, $\ell_C(e^+) = 1$, and the length of the subpath $P$ of $C$ starting at $e^-$ and ending at $e^+$ is minimized. Our choice of $P$ implies that all intermediate edges $\tilde{e}$ in $P$ have $\ell_C(\tilde{e}) = 0$, so they form a desired horizontal segment~$S$. If $\eref \in S$, then we are done. Otherwise,  since $\ell_{\overline{C_{\FO}}}^{\eref}(\tilde{e}) = \ell_{C}(\tilde{e}) = 0$ for any $\tilde{e} \in S$, \cref{lem:ref-edge} implies that all edge labels remain unchanged even if we change the reference edge from  $\eref$ to  $\tilde{e}$. Therefore, $S$ is still a horizontal segment $S \in \PH$ with $\Nnorth(S) = \emptyset$ even if the underlying reference edge is any $\tilde{e} \in S$. Any such an edge $\tilde{e}$ satisfies the requirement of the lemma, as  $\ell_C^{\tilde{e}}(e')=\ell_C^{\eref}(e')  - \ell_{C}^{\eref}(\tilde{e}) = x- 0 = x$ by \cref{lem:ref-edge} with $e=e'$, $\eref_1=\tilde{e}$, and $\eref_2 = \eref$.
\end{proof}

We are now ready to prove \cref{thm-main2}.

\thmtwo*

\begin{proof}
We just need to show that for the case where the given  ortho-radial representation $\RR$ is drawable, an ortho-radial drawing realizing $\RR$ can be computed in $O(n \log^2 n)$ time.
In this case, using \cref{lem:ref-edge-2}, the  binary search algorithm discussed above finds a reference edge $\eref$ such that the $O(n \log n)$-time algorithm of \cref{thm-main1} outputs an ortho-radial drawing realizing $(\RR, \eref)$. This drawing is also an ortho-radial drawing realizing $\RR$. The number of iterations of the binary search is $O\left(\log |I|\right) = O(\log n)$, so the overall time complexity is $O(n \log^2 n)$, as the cost per iteration is $O(n \log n)$, due to \cref{thm-main1}.
\end{proof}

\section{A reduction to biconnected simple graphs}\label{sect:reduction}

In this section, we show that, without loss of generality, we may assume that the input graph is simple and biconnected.  Specifically, given a planar graph $G=(V,E)$, a  combinatorial embedding $\EE$ of $G$, and an ortho-radial representation $\RR$ of $(G,\EE)$  satisfying \ref{item:R1} and \ref{item:R2}, we will construct a biconnected simple planar graph $G'=(V',E')$, a  combinatorial embedding $\EE'$ of $G'$, and an ortho-radial representation $\RR'$ of $(G',\EE')$  satisfying \ref{item:R1} and \ref{item:R2} such that $\RR$ is drawable if and only if $\RR'$ is drawable. See \cref{fig:biconn} for an illustration of the reduction. Moreover, the reduction costs only linear time in the following sense:
\begin{itemize}
    \item The construction of $G'$, $\EE'$, and $\RR'$ takes linear time.
    \item Given an ortho-radial drawing of $\RR'$, an ortho-radial drawing of $\RR$ can be found in linear time.
\end{itemize}

\begin{figure}[t!]
\centering
\includegraphics[width=\textwidth]{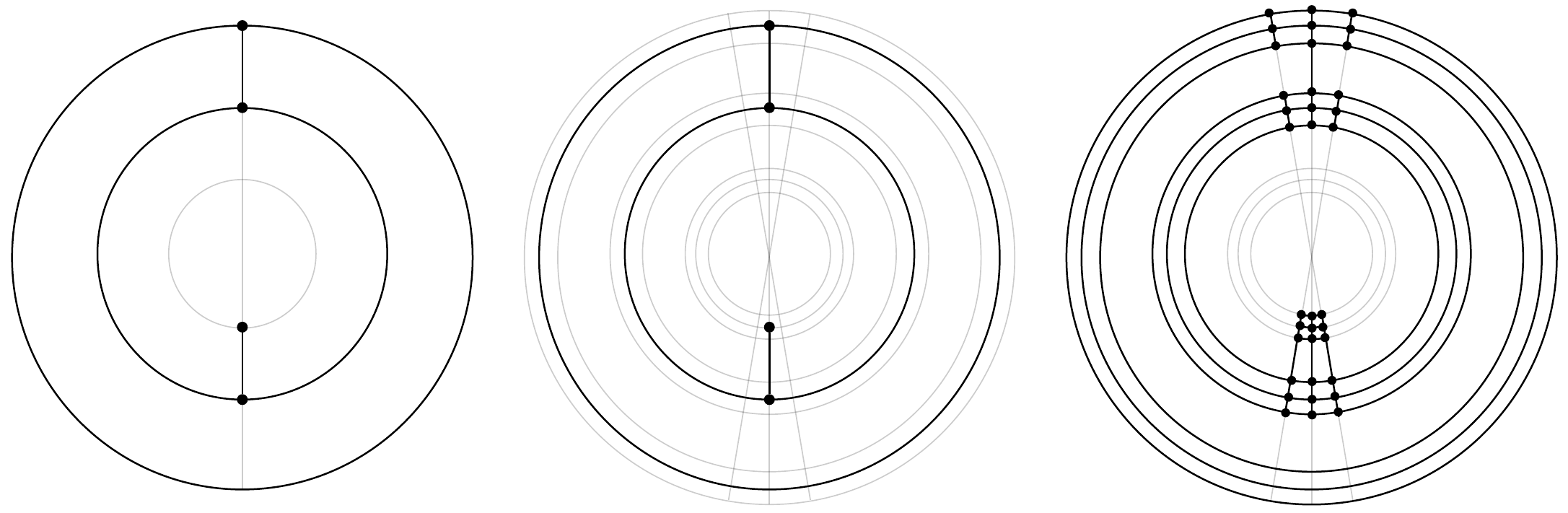}
\caption{Reduction to biconnected simple graphs by thickening.}\label{fig:biconn}
\end{figure}

\paragraph{The reduction}
Throughout this section, all edges are undirected, and we allow the graph $G$ to have multi-edges and self-loops.
The idea of the reduction is to improve the connectivity by thickening the graph. For each vertex $v \in V$, we replace it with a grid consisting of three horizontal lines and three vertical lines. Specifically, let 
\[X_v=\left\{v_{i,j} \; \mid \; i\in\{-1,0,1\} \; \text{and} \; j \in \{-1,0,1\}\right\} \ \ \ \text{and} \ \ \ V' = \bigcup_{v \in V} X_v.\]

 For the construction of the edge set $E'$, we first add an edge between $v_{i,j}$ and $v_{i',j'}$ if $|i-i'|+|j-j'|=1$, so $X_v$ becomes a grid graph, see \Cref{fig:grid}. 

 \begin{sidefigure}[ht!]
\centering
\includegraphics[width=0.23\textwidth]{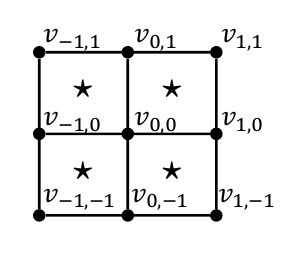}
\caption{The grid graph $X_v$.}\label{fig:grid}
\end{sidefigure}


The interior angles in the four 4-cycles in the grid graph, highlighted by $\star$ in \Cref{fig:grid}, 
 are all set to $90^\circ$ in $\RR'$ to ensure that $X_v$ must be drawn as a grid. We write the four boundary paths of the grid as follows. 
 \begin{align*}
     P_{1,0}^v &= (v_{1,1}, v_{1,0}, v_{1,-1}), & 
     P_{0,-1}^v &= (v_{1,-1}, v_{0,-1}, v_{-1,-1}),\\
     P_{-1,0}^v &= (v_{-1,-1}, v_{-1,0}, v_{-1,1}), & 
     P_{0,1}^v &= (v_{-1,1}, v_{0,1}, v_{1,1}).
 \end{align*}
 The concatenation of these four paths traverses the boundary of the grid in the clockwise direction.
 
  We associate each edge incident to $v$ with a distinct neighbor of $v_{0,0}$ in such a way that is consistent with the given counter-clockwise circular ordering $\EE(v)$ and the angle assignment~$\phi$ to the corners surrounding $v$ in the given ortho-radial representation $\RR$. See \Cref{fig:association}.
  By symmetry, there are five cases. In the first case, $v$ is incident only to one edge $e_1$, and we associate $e_1$ with $v_{1,0}$. In the second case, $v$ has two incident edges $e_1$ and $e_2$ such that $\phi(e_1, e_2) = 90^\circ$ in $\RR$, and we associate $e_1$ with $v_{1,0}$ and associate $e_2$ with $v_{0,1}$. The remaining cases are similar.

 \begin{figure}[ht!]
\centering
\includegraphics[width=0.5\textwidth]{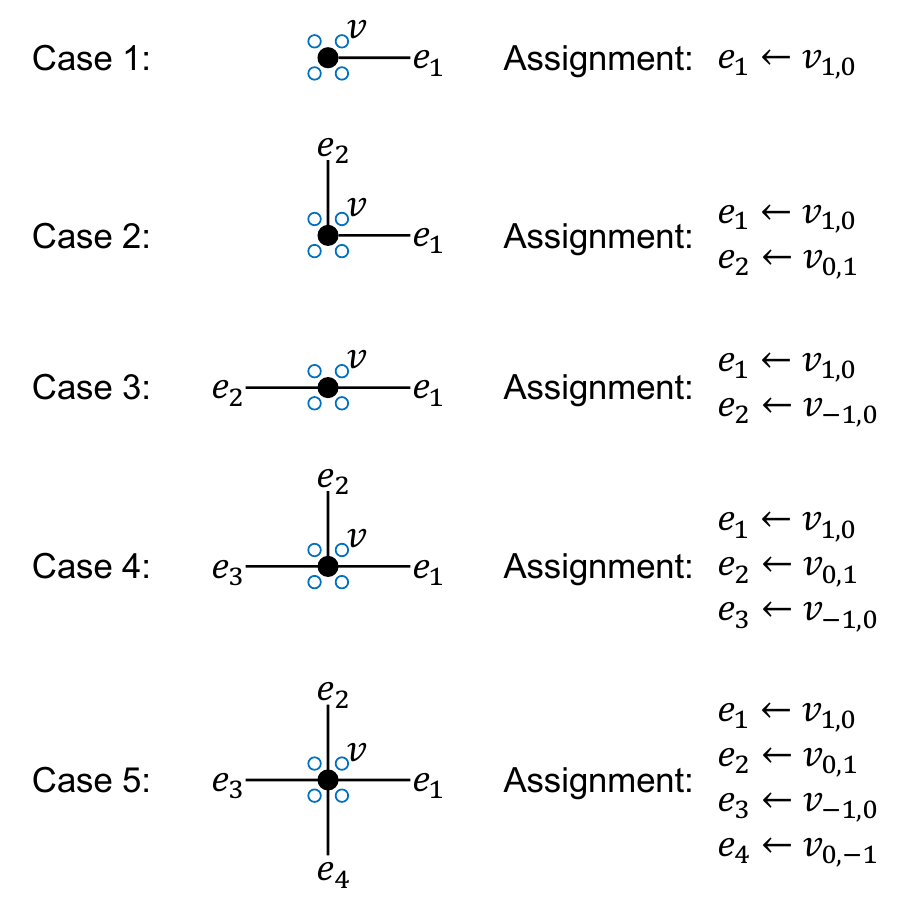}
\caption{Associating each edge incident to $v$ with a distinct neighbor of $v_{0,0}$.}\label{fig:association}
\end{figure}


  For each $e = \{u,v\}\in E$, we add edges to $E'$ to connect $X_u$ and $X_v$, as follows. Suppose that $e \leftarrow u_{i,j}$ and  $e \leftarrow v_{k,l}$ in the above assignment. Then we connect $X_u$ and $X_v$ by adding the three edges $e_1 = \{x_1, y_1\}$, $e_2 = \{x_2, y_2\}$, and $e_3 = \{x_3, y_3\}$ to connect  $P_{i,j}^u = (x_1, x_2, x_3)$ and $\overline{P_{k,l}^v} = (y_1, y_2, y_3)$. See \Cref{fig:connect} for an illustration of the case where $e \leftarrow u_{1,0}$ and $e \leftarrow v_{0,1}$.

   \begin{figure}[ht!]
\centering
\includegraphics[width=0.45\textwidth]{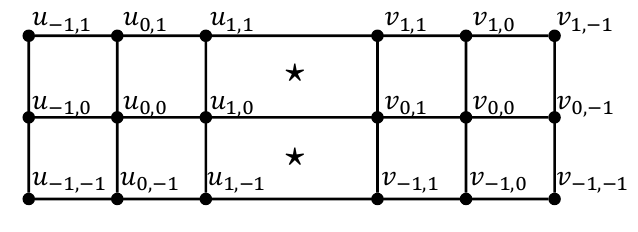}
\caption{Adding edges to $E'$ to connect $X_u$ and $X_v$.}\label{fig:connect}
\end{figure}

The addition of the three edges $e_1$, $e_2$, and $e_3$ create two 4-cycles, highlighted by $\star$ in \Cref{fig:connect}. Similarly, the interior angles in these two 4-cycles  are all set to $90^\circ$ in $\RR'$ to ensure that they must be drawn as rectangles.

This finishes the construction of $(G',\EE')$, which is a biconnected simple plane graph.
All remaining angles in $\RR'$ are set in such a way that the sum of angles surrounding each vertex is $360^\circ$.

\paragraph{Validity of the reduction} To prove that the reduction is valid, we need to show that $\RR$ is drawable if and only if $\RR'$ is drawable. We start with the direction $\RR \rightarrow \RR'$. Suppose we are given an ortho-radial drawing realizing~$\RR$. Our goal is to find an ortho-radial drawing realizing~$\RR'$. 
For each vertex $v \in V$, let $(r_{v}, \theta_v)$ denote its coordinates in the given ortho-radial drawing. We define 
\[B_{\epsilon}(v) = \left\{ (r, \theta) \mid r\in[r_{v} - \epsilon, r_{v} + \epsilon] \; \text{and} \; \theta \in [\theta_v - \epsilon, \theta_v + \epsilon]\right\}.\]
For each edge $e = \{u,v\} \in E$, it is drawn as a horizontal line (a circular arc of some circle centered
at the origin) or a vertical line
 (a line segment of some straight line passing through the origin). That is, the drawing of $e$ is can be described by one of the following two functions, for some choices of the parameters $(r_e, \theta_{e,1}, \theta_{e,2})$ or $(r_{e,1}, r_{e,2}, \theta_e)$:
\begin{align*}
    & \left\{ (r, \theta) \mid r= r_e \; \text{and} \; \theta \in [\theta_{e,1}, \theta_{e,2}]\right\}, && \text{if $e$ is drawn as a horizontal line,}\\
    & \left\{ (r, \theta) \mid r\in[r_{e,1}, r_{e,2}] \; \text{and} \; \theta = \theta_e\right\},  && \text{if $e$ is drawn as a vertical line.}
\end{align*}
We define $B_{\epsilon}(e)$ as follows:
\[
B_{\epsilon}(e) = 
\begin{cases}
\left\{ (r, \theta) \mid r\in[r_e - \epsilon, r_e + \epsilon] \; \text{and} \; \theta\in(\theta_{e,1} + \epsilon, \theta_{e,2} - \epsilon)\right\}
& \text{if $e$ is drawn horizontally,}\\
\left\{ (r, \theta) \mid r\in(r_{e,1} + \epsilon, r_{e,2} - \epsilon) \; \text{and} \; \theta \in [\theta_e - \epsilon, \theta_e + \epsilon]\right\}
& \text{if $e$ is drawn vertically.}
\end{cases}
\]
By choosing $\epsilon > 0$ to be small enough, we can make sure that the sets $B_\epsilon(v)$ and $B_\epsilon(e)$ are non-empty and pairwise disjoint, over all $v\in V$ and $e\in E$. For each vertex $v\in V$, we draw the grid $X_v$ in $B_\epsilon(v)$ by drawing the following grid-lines:
\begin{align*}
    & \left\{ (r, \theta) \mid r= r_v +  s \cdot \epsilon \; \text{and} \; \theta \in [\theta_v - \epsilon, \theta_v + \epsilon]\right\},  && \text{for $s \in \{-1, 0, 1\}$,}\\
    & \left\{ (r, \theta) \mid r\in[r_{v} - \epsilon, r_{v} + \epsilon] \; \text{and} \; \theta = \theta_v + s \cdot \epsilon \right\},  && \text{for $s \in \{-1, 0, 1\}$.}
\end{align*}
For each $e=\{u,v\} \in E$, we draw the three edges connecting $X_u$ and $X_v$ by drawing the following three lines in $B_\epsilon(e)$:
\begin{align*}
    & \left\{ (r, \theta) \mid r= r_e +  s \cdot \epsilon \; \text{and} \; \theta \in (\theta_{e,1} + \epsilon, \theta_{e,2} - \epsilon)\right\}, \text{for $s \in \{-1, 0, 1\}$, if $e$ is drawn horizontally,}\\
    & \left\{ (r, \theta) \mid r\in(r_{e,1} + \epsilon, r_{e,2} - \epsilon) \; \text{and} \; \theta = \theta_e +  s \cdot \epsilon\right\},  \text{for $s \in \{-1, 0, 1\}$, if $e$ is drawn vertically.}
\end{align*}
The validity of this drawing of $\RR'$ follows from the disjointness of the sets $B_\epsilon(v)$ and $B_\epsilon(e)$, over all $v\in V$ and $e\in E$.
 
Next, we consider the other direction $\RR \rightarrow \RR'$.  Suppose we are given an ortho-radial drawing realizing $\RR'$. Our goal is to find an ortho-radial drawing realizing $\RR$. For each vertex $v \in V$, we put $v$ at the position of $v_{0,0}$ in the given drawing of $\RR'$. 
To draw each edge $e=\{u,v\} \in E$, consider the assignment $e \leftarrow u_{i,j}$ and  $e \leftarrow v_{k,l}$ described in the reduction. The path $P = (u_{0,0}, u_{i,j}, v_{k,l}, v_{0,0})$ must be drawn as a straight line, due to the angle assignment of $\RR'$ described in our reduction. That is, in the given drawing of $\RR'$, $P$ is drawn as either a circular arc of some circle centered
at the origin or a line segment of some straight line passing through the origin. Therefore, we may use the drawing of $P$ to embed $e$, and this gives us a desired drawing of $\RR$.

\section{Conclusions}\label{sect:conclusions}

In this paper, we presented a near-linear time algorithm to decide whether a given ortho-radial representation is drawable, improving upon the previous quadratic-time algorithm~\cite{barth2023topology}. If the representation is drawable, then our algorithm outputs an ortho-radial drawing realizing the representation. Otherwise, our algorithm outputs a strictly monotone cycle to certify the non-existence of such a drawing. Given the broad applications of the topology-shape-metric framework in orthogonal drawing, we anticipate that our new ortho-radial drawing algorithm will be relevant and useful in future research in this field.

While there has been extensive research in orthogonal drawing, much remains unknown about the computational complexity of basic optimization problems in ortho-radial drawing. In particular, the problem of finding an ortho-radial representation that minimizes the number of bends has only been addressed by a practical algorithm~\cite{niedermann2020integer} that has no provable guarantees. It remains an intriguing open question to determine to what extent bend minimization is polynomial-time solvable for ortho-radial drawing.\footnote{The paper~\cite{barth2023topology} also lists the computation of a bend-minimized ortho-radial representation for a given plane graph as an open question. In the preliminary version (\href{https://arxiv.org/abs/2106.05734v1}{arXiv:2106.05734v1}) of the same paper~\cite{barth2023topology}, the authors claimed this problem to be NP-hard (Theorem 7), seemingly resolving the problem. However, the result was removed in the journal version, indicating a potential issue with their reduction, so the problem remains open.} To the best of our knowledge, even deciding whether a given plane graph admits an ortho-radial drawing \emph{without bends} is not known to be polynomial-time solvable.

Given an ortho-radial representation, can we find an ortho-radial drawing with the smallest number of layers (i.e., the number of concentric circles) in polynomial time? As discussed in \cref{sect:drawing}, if a good sequence is given, then our algorithm can output a layer-minimized drawing. For the general case where a good sequence might not exist, our algorithm does not have the layer-minimization guarantee, as there is some flexibility in the choice of virtual edges to add, and selecting different virtual edges results in different good sequences. 

There was a series of work in finding \emph{compact} orthogonal drawings according to various complexity measures~\cite{JGAA-595,bridgeman2000turn,JGAA-263,klau1999optimal,patrignani2001complexity}. To what extent can the ideas developed in these works be applied to ortho-radial drawings?

\printbibliography

\appendix

\end{document}